\documentclass[12pt,a4paper]{article}%
\usepackage{amsfonts}
\usepackage{amssymb}
\usepackage{amsmath}
\usepackage{amsthm}
\usepackage{color}
\usepackage{graphics}
\usepackage{epsfig,amssymb,latexsym,verbatim}
\usepackage{amsthm,amscd,graphics,latexsym}
\usepackage{epstopdf}
\usepackage{graphicx}
\usepackage{multirow}
\usepackage{multicol}
\usepackage{float}
\usepackage{hyperref}
\usepackage{graphicx, amssymb}
\usepackage{breakcites}
\usepackage{bm}
\usepackage[normalem]{ulem}
\usepackage{xr}%
\providecommand{\U}[1]{\protect\rule{.1in}{.1in}}
\floatstyle{ruled}
\newfloat{algorithm}{tbp}{loa}
\floatname{algorithm}{Algorithm}

\newtheorem{theorem}{\textbf{Theorem}}

\newtheorem{proposition}{\textbf{Proposition}}
\newtheorem{lemma}{\textbf{Lemma}}

\def \be{\begin{equation}}
\def \ee{\end{equation}}
\def \bea{\begin{eqnarray}}
\def \eea{\end{eqnarray}}
\begin{document}

\title{{\normalsize {{\Large \textbf{Estimation in Semiparametric Quantile Factor
Models}} }}}
\author{Shujie Ma\\University of California at Riverside\thanks{Ma's research was partially
supported by NSF grants DMS 1306972 and DMS 1712558.}
\and Oliver Linton\\University of Cambridge
\and Jiti Gao\\Monash University\thanks{Gao's research was supported by the Australian
Research Council Discovery Grants Program for its support under Grant numbers:
DP150101012 \& DP170104421.}}
\maketitle

\begin{abstract}
We propose an estimation methodology for a semiparametric quantile factor
panel model. We provide tools for inference that are robust to the existence
of moments and to the form of weak cross-sectional dependence in the
idiosyncratic error term. We apply our method to daily stock return data.

 \medskip


\end{abstract}

\newpage

\renewcommand{\theequation}{1.\arabic{equation}}

\setcounter{equation}{0}

\section{Introduction}

Factor models are widely used to capture the co-movement of a large number of time series and to model
covariance matrices. They provide useful dimensionality reduction in many applications from climate modelling to
finance. Perhaps the current state of the art for factor
modelling is Fan, Liao, and Micheva (2013), which allowed the idiosyncratic covariance matrix to be
non-diagonal but sparse, and used thresholding techniques {\color{black} (Cai and Liu, 2011)} to impose sparsity
and thereby obtain a better estimator of the covariance matrix and its inverse
in this big-data setting.
The usual approach ignores covariate information that can sometimes be informative.
Connor, Hagmann and Linton (2012) developed a semiparametric factor regression methodology
that introduces covariate information into the factor loading parameters. This model is well motivated
in finance applications where it can be understood as a properly formulated version of the popular Fama-French (1992)
approach to modelling returns with observable characteristics. The model also makes sense
in other contexts where covariate information is available. Their application was to monthly stock returns,
which is where the finance literature was focussed. {\color{black} Moreover, Fan, Liao and Wang (2016) proposed a Projected-PCA approach which employs principal component analysis to the projected data matrix onto a linear space spanned by covariates.
It is worth noting that most existing works in the literature of factor models require
at least four moments to establish their theoretical properties. See, for instance, Bai and Ng (2002), Bai and Li (2012), Lam and Yao (2012), Connor, Hagmann and Linton (2012), Fan, Liao, and Micheva (2013), Fan, Liao and Wang (2016), Li et al. (2017), among others.} This may not be a binding restriction for monthly stock returns, but for daily
stock returns this is a bit strong.

Quantile methods are widely used in statistics. They have the advantage of being robust to large observations.
They can also provide more information about the conditional distribution away from the centre, which is relevant in many applications.
In this paper, we propose estimation and inferential methodology for the
quantile version of the Connor, Hagmann and Linton (2012) model. Our contribution is summarized as follows.

First, we propose an estimation algorithm for this model. We use sieve
techniques to obtain preliminary estimators of the nonparametric beta
functions, see Chen (2011) for a review, and use these to estimate the factor return vector
at each time period. We then update the loading functions and factor returns
sequentially. We compute the estimator in two steps for computational reasons.
We have $J\times T$ unknown factor return parameters as well as $J\times
K_{N}$ sieve parameters to estimate, and simultaneous estimation of these parameters
without penalization would be challenging. Penalization of the factor returns
here is not well motivated so we do not pursue this. Instead we first
estimate the unrestricted additive quantile regression function for each time
period and then impose the factor structure in a sequential fashion.

Second, we derive the limiting properties of our estimated factor returns and
factor loading functions under the assumption that the included factors all
have non zero mean and under weak conditions on cross-section and temporal
dependence. A key consideration in the panel modelling of stock returns is
what position to take on the cross sectional dependence in the idiosyncratic
part of stock returns. Early studies assumed iid in the cross section, but
this turns out to be not necessary. More recent work has allowed for cross
sectional dependence in a variety of ways. Connor, Hagmann and Linton (2012)
imposed a known industry cluster/block structure where the number of
industries goes to infinity as do the number of members of the industry. Under
this structure one obtains a CLT and inference can be conducted by estimating
only the intra block covariances. Robinson and Thawornkaiwong (2012)
considered a linear process structure driven by independent shocks. Dong, Gao
and Peng (2015) introduced a spatial mixing structure to accommodate both
serial correlation and cross--sectional dependence for a general panel data
setting. Conley (1999) studied that under a lattice structure or some
observable or estimable distance function that determines the ordering, one
can consistently estimate the asymptotic covariance matrix. However, this type
of structure is hard to justify for stock returns, and in that case their
approach does not deliver consistent inference.
Connor and Koraczyck (1993) considered a different cross-sectional dependence
structure, namely, they supposed that there was an ordering of the cross
sectional units such that weak dependence of the alpha mixing variety was
held. They do not assume\ knowledge of the ordering as this was not needed for
their main results. We adopt and generalize their structure. In fact, we allow
for weak dependence simultaneously in the cross-section and time series
dependence. This structure affects the limiting distribution of the estimated
factor returns in a complicated fashion, and the usual Newey--West type of
standard errors can't be adapted to account for the cross-sectional dependence
here because the ordering is not assumed to be known. To conduct inference we
have to take account of the correlation structure. We use the so-called fix-b
asymptotics to achieve this, namely, we construct a test statistic based on an
inconsistent fixed-b estimator of the
correlation structure, as in Kiefer and Vogelsang (2002), and show that it has
a pivotal limiting distribution that is a functional of a Gaussian process.

Third, our estimation procedure only requires that the time series mean of
factor returns be non zero. A number of authors have noted that in the
presence of a weak factor, regression identification strategies can break down
(Bryzgalova, 2015). In view of this we provide a test of whether a given
factor is present or not in each time period.

Fourth, we apply our procedure to CRSP daily data and show how the factor
loading functions vary nonlinearly with state. The median regression
estimators are comparable to those of Connor, Hagmann and Linton (2012) and
can be used to test asset pricing theories under comparable quantile
restrictions, see for example, Bassett, Koenker and Kordas (2004), and to
design investment strategies. The lower quantile estimators could be used for
risk management purposes. The advantage of the quantile method is its
robustness to heavy tails in the response distribution, which may be present
in daily data. Indeed our theory does not require any moment conditions.

The organization of this paper is given as follows. Section 2 proposes the
main model and then discusses some identification issues. An estimation method
based on B--splines is then proposed in Section 3. Section 4 establishes an
asymptotic theory for the proposed estimation method. Section 5 discusses a
covariance estimation problem and then considers testing for the factors
involved in the main model. Section 6 gives an empirical application of the
proposed model and estimation theory to model the dependence of daily returns
on a set of characteristic variables. Section 7 concludes the paper with some
discussion. All the mathematical proofs of the main results are given in an appendix and on-line supplemental materials.

\renewcommand{\theequation}{2.\arabic{equation}}

\setcounter{equation}{0}

\section{The model and identification}

\label{SEC:intro}

We introduce some notations which will be used throughout the paper. For any
positive numbers $a_{n}$ and $b_{n}$, let $a_{n}\asymp b_{n}$ denote
lim$_{n\rightarrow\infty}a_{n}/b_{n}=c$, for a positive constant $c$, and let
$a_{n}\gg b_{n}$ denote $a_{n}^{-1}b_{n}=o(1)$. For any vector $\mathbf{a=(}%
a_{1},\ldots,a_{n})^{\intercal}\in\mathbb{R}^{n}$, denote $||\mathbf{a||}%
=\left(  \sum\nolimits_{i=1}^{n}a_{i}^{2}\right)  ^{1/2}$. For any symmetric
matrix $\mathbf{A}_{s\times s}$, denote its $L_{2}$ norm as $\left\Vert
\mathbf{A}\right\Vert =\max_{\mathbf{\zeta\in}\mathbb{R}^{s}\mathbf{,\zeta
\neq0}}\left\Vert \mathbf{A\zeta}\right\Vert \left\Vert \mathbf{\zeta
}\right\Vert ^{-1}$. We use $\left(  N,T\right)  \rightarrow\infty$ to denote
that $N$ and $T$ pass to infinity jointly.

We consider the following model for the $\tau^{\text{th}}$ conditional
quantile function of the response $y_{it}$ for the $i^{\text{th}}$ asset at
time $t$ given as
\begin{equation}
Q_{y_{it}}(\tau|X_{i},f_{t})=f_{ut}+\sum\nolimits_{j=1}^{J}g_{j}(X_{ji}%
)f_{jt},\label{EQ:factorquantile}%
\end{equation}
i.e., we suppose that
\begin{equation}
y_{it}=f_{ut}+\sum\nolimits_{j=1}^{J}g_{j}(X_{ji})f_{jt}+\varepsilon_{it},
\label{q}%
\end{equation}
for $i=1,\ldots,N$ and $t=1,\ldots,T$, where $y_{it}$ is the excess return to
security $i$ at time $t$; $f_{ut}$ and $f_{jt}$ are factor returns, which are
unobservable; $g_{j}(X_{ji})$ are the factor betas, which are unknown but
smooth functions of $X_{ji}$, where $X_{ji}$ are observable security
characteristics, and $X_{ji}$ lies in a compact set $\mathcal{X}_{ji}$. Let
$X_{i}=(X_{1i},\ldots,X_{Ji})^{\intercal}$ and $f_{t}=(f_{ut},f_{1t}%
,\ldots,f_{Jt})^{\intercal}$. The
error terms $\varepsilon_{it}$ are the asset-specific or idiosyncratic returns
and they satisfy that the conditional $\tau^{\text{th}}$ quantile of
$\varepsilon_{it}$ given $\left(  X_{i},f_{t}\right)  $ is zero. The factors
$f_{ut}$ and $f_{jt}$ and the factor betas $g_{j}(\cdot)$ should
be $\tau$ specific. For notational simplicity, we suppress the $\tau$
subscripts. For model identifiability, we assume that:

\textsc{Assumption A0}. \emph{For some probability measures }$P_{j}$\emph{ we
have} $\int g_{j}(x_{j})dP_{j}(x_{j})=0$ \emph{and} $\int\left(  g_{j}(x_{j})\right)  ^{2}dP_{j}(x_{j})=1$ for all $j=1,\ldots,J$. Furthermore,
$\lim\inf_{T\rightarrow\infty}\left\vert \sum_{t=1}^{T}f_{jt}/T\right\vert
>0$ for each $j$.

The case where $\tau=1/2$ corresponds to the conditional median, and is
comparable to the conditional mean model used in Connor, Hagmann and Linton (2012). The advantage of the median over the mean
is its robustness to heavy tails and outliers, which is especially important
with daily data. The case where $\tau=0.01,$ say, might be of interest for the
purposes of risk management, since this corresponds to a standard
Value-at-Risk threshold in which case (\ref{EQ:factorquantile}) gives the
conditional Value-at-Risk given the characteristics and the factor returns at
time $t.$ To obtain an ex-ante measure we should have to employ a forecasting
model for the factor returns.

Suppose that the $\tau^{\text{th}}$ conditional quantile function $Q_{y_{it}%
}(\tau|X_{i}=x)$ of the response $y_{it}$ at time $t$ given the covariate
$X_{i}=x$ is additive
\begin{equation}
H_{t}(\tau|x)=h_{ut}+\sum\nolimits_{j=1}^{J}h_{jt}(x_{j}),
\label{EQ:quantileinitial}%
\end{equation}
where $h_{jt}(\cdot)$ are unknown functions without loss of generality
satisfying $\int h_{jt}(x_{j})dP_{j}(x_{j})=0$ for $t=1,\ldots,T$
(Horowitz and Lee, 2005). Under the factor structure (\ref{EQ:factorquantile}%
), we have for all $j$%
\begin{equation}
\int\left(  \frac{1}{T}\sum_{t=1}^{T}h_{jt}(x_{j})\right)  ^{2}%
dP_{j}(x_{j})=\int g_{j}(x_{j})^{2}dP_{j}(x_{j})\times\left(  \frac{1}%
{T}\sum_{t=1}^{T}f_{jt}\right)  ^{2}=\left(  \frac{1}{T}\sum_{t=1}%
^{T}f_{jt}\right)  ^{2}. \label{m0}%
\end{equation}
Provided $\sum_{t=1}^{T}f_{jt}\neq0$, we can identify $g_{j}(x_{j})$
by
\begin{equation}
g_{j}(x_{j})=\frac{\frac{1}{T}\sum_{t=1}^{T}h_{jt}(x_{j})}{\sqrt
{\int\left(  \frac{1}{T}\sum_{t=1}^{T}h_{jt}(x_{j})\right)  ^{2}%
dP_{j}(x_{j})}}. \label{EQ:g0}%
\end{equation}
We will use this as the basis for the proposal of the estimation method in
Section 3 below.

\renewcommand{\theequation}{3.\arabic{equation}}

\setcounter{equation}{0}

\section{Estimation\label{estimation}}

\subsection{Factor returns and characteristic-beta functions \label{estimator}%
}

We propose an iterative algorithm to estimate the factor returns and the
characteristic-beta functions. The algorithm makes use of the structure in
(\ref{q}) so that it circumvents the \textquotedblleft curse of
dimensionality" (Bellman, 1961) while retaining flexibility of the nonparametric
regression. The right hand side of (\ref{EQ:factorquantile}) is bilinear in
unknown quantities, so it seems difficult to avoid such an algorithmic approach.

To estimate $g_{j}(\cdot)$, we first approximate them by B-spline
functions described as follows. Let $b_{j}(x_{j})=\{b_{j,1}(x_{j}%
),\ldots,b_{j,K_{N}}(x_{j})\}^{\intercal}$ be a set of normalized B-spline
functions of order $m$ (see, for example, de Boor (2001)), where $K_{N}%
=L_{N}+m$, and $L_{N}$ is the number of interior knots satisfying
$L_{N}\rightarrow\infty$ as $N\rightarrow\infty$. We adopt the centered
B-spline basis functions $B_{j}(x_{j})=\{B_{j,1}(x_{j}),\ldots,B_{j,K_{N}%
}(x_{j})\}^{\intercal},$ where
\[
B_{jk}(x_{j})=\sqrt{K_{N}}\left[  b_{j,k}(x_{j})-N^{-1}\sum\nolimits_{i=1}%
^{N}b_{j,k}(X_{ji})\right]  ,
\]
so that $N^{-1}\sum\nolimits_{i=1}^{N}B_{jk}(X_{ji})=0$ and var$\{B_{jk}%
(X_{j})\} \asymp1$. We first approximate the unknown functions $g_{j}(x_{j})$
by B-splines such that $g_{j}(x_{j})$ $\approx B_{j}(x_{j})^{\intercal
}\boldsymbol{\lambda}_{j}$, where $\boldsymbol{\lambda}_{j}=(\lambda
_{j,1},\ldots,\lambda_{j,K_{N}})^{\intercal}$ are spline coefficients. Hence
$N^{-1}\sum\nolimits_{i=1}^{N}B_{j}(X_{ji})^{\intercal}\boldsymbol{\lambda
}_{j}=0$. Denote $f_{t}=\{f_{ut},(f_{jt},1\leq j\leq J)^{\intercal
}\}^{\intercal}$. Let $\boldsymbol{\lambda}=(\boldsymbol{\lambda}%
_{1}^{\intercal},\ldots,\boldsymbol{\lambda}_{J}^{\intercal})^{\intercal}$ and
let $\rho_{\tau}(u)=u(\tau-I(u<0))$ be the quantile check function. The
iterative algorithm is described as follows:

1. Find the initial estimates $\widehat{f}^{[0]}$ and $\widehat{g}_{j}%
^{[0]}(\cdot)$.

2. For given $\widehat{f}^{[i]}$, we obtain
\[
\widehat{\boldsymbol{\lambda}}^{[i+1]}=\arg\min_{\boldsymbol{\lambda}%
\in\mathbb{R}^{JK_{N}}}\sum\nolimits_{i=1}^{N}\sum\nolimits_{t=1}^{T}%
\rho_{\tau}\left(  y_{it}-\widehat{f}_{ut}^{[i]}-\sum\nolimits_{j=1}^{J}%
B_{j}(X_{ji})^{\intercal}\boldsymbol{\lambda}_{j}\widehat{f}_{jt}%
^{[i]}\right)  .
\]
Let $\widehat{g}_{j}^{\ast\lbrack i+1]}(x_{j})$ $=B_{j}(x_{j})^{\intercal
}\widehat{\boldsymbol{\lambda}}_{j}^{[i+1]}$. The estimate for $g_{j}(x_{j})$
at the $(i+1)^{\text{th}}$ step is
\[
\widehat{g}_{j}^{[i+1]}(x_{j})=\frac{\widehat{g}_{j}^{\ast\lbrack i+1]}%
(x_{j})}{\sqrt{N^{-1}\sum\nolimits_{i=1}^{N}\widehat{g}_{j}^{\ast\lbrack
i+1]}(X_{ji})^{2}}}.
\]

3. For given $\widehat{g}_{j}^{[i+1]}(x_{j})$, we obtain for $t=1,\ldots,T$
\[
\widehat{f}_{t}^{[i+1]}=\arg\min_{f_{t}\in\mathbb{R}^{J+1}}\sum\nolimits_{i=1}%
^{N}\rho_{\tau}\left(  y_{it}-f_{ut}-\sum\nolimits_{j=1}^{J}\widehat{g}%
_{j}^{[i+1]}(X_{ji})f_{jt}\right)  .
\]

We repeat steps 2 and 3, and consider that the algorithm converges at the
$(i+1)^{\text{th}}$ step when $||\widehat{f}^{[i+1]}-\widehat{f}%
^{[i]}||<\epsilon$ and $||\widehat{\boldsymbol{\lambda}}^{[i+1]}%
-\widehat{\boldsymbol{\lambda}}^{[i]}||<\epsilon$ for a small positive value
$\epsilon$. Then the final estimates are $\widehat{f}_{t}=\widehat{f}%
_{t}^{[i+1]}$ and $\widehat{g}_{j}(x_{j})$ $=\widehat{g}_{j}^{[i+1]}(x_{j})$.
Our experience in numerical analysis suggests that the proposed method
converges well and rapidly using the consistent initial values proposed in
Section \ref{initial}. The algorithm stops after a finite number of iterations
by using the consistent initial values.

\subsection{Initial estimators\label{initial}}

We first approximate the unknown functions $h_{jt}(x_{j})$ by B-splines such
that $h_{jt}(x_{j})\approx B_{j}(x_{j})^{\intercal}\boldsymbol{\theta}_{jt}$,
where $\boldsymbol{\theta}_{jt}=(\theta_{jt,1},\ldots,\theta_{jt,K_{N}%
})^{\intercal}$ are spline coefficients. Let $\boldsymbol{\theta}%
_{t}=(\boldsymbol{\theta}_{1t}^{\intercal},\ldots,\boldsymbol{\theta}%
_{Jt}^{\intercal})^{\intercal}$. Then the estimators ($\widetilde{h}%
_{ut},\widetilde{\boldsymbol{\theta}}_{t}^{\intercal})^{\intercal}$ of
$(h_{ut},\boldsymbol{\theta}_{t}^{\intercal})^{\intercal}$ are obtained by
minimizing%
\[
\sum\nolimits_{i=1}^{N}\rho_{\tau}(y_{it}-h_{ut}-\sum\nolimits_{j=1}^{J}%
B_{j}(X_{ji})^{\intercal}\boldsymbol{\theta}_{jt})
\]
with respect to $(h_{ut},\boldsymbol{\theta}_{t}^{\intercal})^{\intercal}%
\in\mathbb{R}^{JK_{N}+1}.$ As a result, the estimator of $h_{jt}(x_{j})$
is $\widetilde{h}_{jt}(x_{j})=B_{j}(x_{j})^{\intercal}%
\widetilde{\boldsymbol{\theta}}_{jt}$. \ We then obtain the initial estimators
of $g_{j}(x_{j})$
\begin{equation}
\widehat{g}_{j}^{[0]}(x_{j})=\frac{T^{-1}\sum_{t=1}^{T}\widetilde{h}%
_{jt}(x_{j})}{\sqrt{N^{-1}\sum_{i=1}^{N}\left(  \frac{1}{T}\sum_{t=1}%
^{T}\widetilde{h}_{jt}(X_{ji})\right)  ^{2}}}.\label{EQ:ghat0}%
\end{equation}
The initial estimator of $f_{t}$ is
\begin{equation}
\widehat{f}_{t}^{[0]}=\arg\min_{f_{t}\in\mathbb{R}^{J+1}}\sum\nolimits_{i=1}%
^{N}\rho_{\tau}(y_{it}-f_{ut}-\sum\nolimits_{j=1}^{J}\widehat{g}_{j}%
^{[0]}(X_{ji})f_{jt}) \label{EQ:fhat0}%
\end{equation}
for $t=1,\ldots,T$.

\renewcommand{\theequation}{4.\arabic{equation}}

\setcounter{equation}{0}

\section{Asymptotic theory of the estimators\label{asymptotic}}

We suppose that there is some relabelling of the cross-sectional units
$i_{l_{1}},\ldots,i_{l_{N}},$ whose generic index we denote by $i^{\ast}$,
such that the cross sectional dependence decays with the distance $|i^{\ast
}-j^{\ast}|$. This assumption has been made in Connor and Korajczyk (1993) and
Lee and Robinson (2016). Our estimation procedure does not need
to know the ordering of the data. However, to develop a robust inference
procedure that accounts for heteroscedasticity and cross-sectional correlation
(HAC), we need to order the data across $i$. As discussed in Lee and Robinson
(2016), in some economic applications, data may be ordered according to some
explanatory variables. Such considerations are pursued in our real data
analysis with detailed discussions given in Section \ref{application}.
For notational simplicity, we denote the indices as $\{i,1\leq i\leq N\}$
after the ordering.

Let $g_{j}^{0}(\cdot)$ for $j=1,\ldots,J$ and $f_{t}^{0}=(f_{ut}^{0}%
,f_{1t}^{0},\ldots,f_{Jt}^{0})^{\intercal}$ be the true factor betas and
factor returns in model (\ref{q}). For model identifiability, assume
$E\{g_{j}^{0}(X_{ji})\}=0$ and $E\{g_{j}^{0}(X_{ji})\}^{2}=1$. Let $\mathbb{N}$ denote the collection of all positive
integers. We use a $\phi$-mixing coefficient to specify the dependence
structure. Let $\{W_{it}:1\leq i\leq N,1\leq t\leq T\},$ where $W_{it}%
=(X_{i}^{^{\intercal}},f_{t}^{^{\intercal}},\varepsilon_{it})^{^{\intercal}}$
and $\varepsilon_{it}=y_{it}-f_{ut}^{0}-\sum\nolimits_{j=1}^{J}g_{j}%
^{0}(X_{ji})f_{jt}^{0}.$ For $S_{1},S_{2}\subset\lbrack1,\ldots,N]\times
\lbrack1,\ldots,T]$, let%
\[
\phi(S_{1},S_{2})\equiv\sup\{|P(A|B)-P(A)|:A\in\sigma(W_{it},(i,t)\in
S_{1}),B\in\sigma(W_{it},(i,t)\in S_{2})\},
\]
where $\sigma\left(  \cdot\right)  $ denotes a $\sigma$-field. Then the $\phi
$-mixing coefficient of $\{W_{it}\}$ for any $k\in\mathbb{N}$ is defined as
\[
\phi(k)\equiv\sup\{\phi(S_{1},S_{2}):d(S_{1},S_{2})\geq k\},
\]
where
\[
d(S_{1},S_{2})\equiv\min\{\sqrt{|t-s|^{2}+|i-j|^{2}}:(i,t)\in S_{1},(j,s)\in
S_{2}\}.
\]
Without loss of generality, we assume that $\mathcal{X}_{ji}=[a,b]$. Denote
$h_{t}^{0}(x)=\{h_{jt}^{0}(x_{j}),1\leq j\leq J\}^{\intercal}$, where
$h_{jt}^{0}(\cdot)$ are the true unknown functions in
(\ref{EQ:quantileinitial}) and $x=(x_{1},\ldots,x_{J}%
)^{\intercal}$. Let $G_{i}^{0}(X_{i}%
)=\{1,g_{1}^{0}(X_{1i}),\ldots,g_{J}^{0}(X_{Ji})\}^{\intercal}.$ We make the
following assumptions.

\begin{enumerate}
\item[(C1)] $\{W_{it}\}$ is a random field of $\phi$-mixing random variables.
The $\phi$-mixing coefficient of $\{W_{it}\}$ satisfies $\phi(k)\leq
K_{1}e^{-\lambda_{1}k}$ for $K_{1},\lambda_{1}>0$. For each given $i$,
$\{W_{it}\}$ is a strictly stationary sequence.

\item[(C2)] The conditional density $p_{i}\left(  \varepsilon\left\vert
x_{i},f_{t}\right.  \right)  $ of $\varepsilon_{it}$ given $\left(
x_{i},f_{t}\right)  $ satisfies the Lipschitz condition of order $1$ and
$\inf_{1\leq i\leq N,1\leq t\leq T}p_{i}\left(  0\left\vert x_{i}%
,f_{t}\right.  \right)  >0$. For every $1\leq j\leq J$, the density function
$p_{X_{ji}}(\cdot)$ of $X_{ji}$ is bounded away from $0$ and satisfies the
Lipschitz condition of order $1$ on $[a,b]$. The density function $f_{X_{i}%
}(\cdot)$ of $X_{i}$ is absolutely continuous on $[a,b]^{J}$.

\item[(C3)] The functions $g_{j}^{0}$ and $h_{jt}^{0}$ are $r$-times
continuously differentiable on its support for some $r>2$. The spline order
satisfies $m\geq r$.

\item[(C4)] There exist some constants $0<c_{h}\leq C_{h}<\infty$ such that
$c_{h}\leq\left(  \frac{1}{T}\sum_{t=1}^{T}f_{jt}^{0}\right)  ^{2}\leq C_{h}$
for all $j$ with probability tending to one.

\item[(C5)] The eigenvalues of the $(J+1)\times(J+1)$ matrix $N^{-1}%
\sum\nolimits_{i=1}^{N}E(G_{i}^{0}(X_{i})G_{i}^{0}(X_{i})^{\intercal})$ are
bounded away from zero.

\item[(C6)] Let $\Omega_{N}^{0}$ be the covariance matrix of $N^{-1/2}%
\sum\nolimits_{i=1}^{N}G_{i}^{0}(X_{i})(\tau-I(\varepsilon_{it}<0))$. The
eigenvalues of $\Omega_{N}^{0}$ are bounded away from zero and infinity.
\end{enumerate}

We allow that $\{W_{it}\}$ are weakly dependent across $i$ and $t$, but need
to satisfy the strong mixing condition given in Condition (C1). Moreover,
Condition (C1) implies that $\{X_{i}\}$ is marginally cross-sectional mixing,
and $\{f_{t}\}$ is marginally temporally mixing. Similar assumptions are used
in Gao, Lu and Tj{\o }stheim (2006) for an alpha--mixing condition in a
spatial data setting, and Dong, Gao and Peng (2016) for introducing a spatial
mixing condition in a panel data setting. Conditions (C2)\ and (C3) are
commonly used in the nonparametric smoothing literature, see for example,
Horowitz and Lee (2005), and Ma, Song and Wang (2013). Conditions (C4)\ and
(C5)\ are similar to Conditions A2, A5 and A7 of Connor, Matthias and Linton (2012).

Define
\[
\Lambda_{Nt}^{0}=N^{-1}\sum\nolimits_{i=1}^{N}E\{p_{i}\left(  0\left\vert
X_{i},f_{t}\right.  \right)  G_{i}^{0}(X_{i})G_{i}^{0}(X_{i})^{\intercal}\}.
\]
and
\begin{equation}
\mathbf{\Sigma}_{Nt}^{0}=\tau(1-\tau)(\Lambda_{Nt}^{0})^{-1}\Omega_{N}%
^{0}(\Lambda_{Nt}^{0})^{-1}. \label{EQ:SigN}%
\end{equation}
The theorem below presents the asymptotic distribution of the final estimator
$\widehat{f}_{t}$. Define
\begin{equation}
\phi_{NT}=\sqrt{K_{N}/(NT)}+K_{N}^{3/2}N^{-3/4}\sqrt{\log NT}+K_{N}^{-r}.
\label{EQ:phNT}%
\end{equation}
Let $d_{NT}$ be a sequence satisfying
\begin{equation}
d_{NT}=O(\phi_{NT}). \label{dNT}%
\end{equation}

\begin{theorem}
\label{THM:fhat}Assume that Conditions (C1)-(C5) hold,\ and $K_{N}^{4}%
N^{-1}=o(1)$, $K_{N}^{-r+2}(\log T)=o(1)$ and $K_{N}^{-1}(\log NT)(\log
N)^{4}=o(1).$ Suppose that the algorithm in Section \ref{estimator} converges
within a finite number of iterations. Then, for any $t$ there is a
stochastically bounded sequence $\delta_{N,jt}$ such that as $N\rightarrow
\infty$,
\[
\sqrt{N}(\mathbf{\Sigma}_{Nt}^{0})^{-1/2}(\widehat{f}_{t}-f_{t}^{0}%
-d_{NT}\delta_{N,t})\overset{\mathcal{D}}{\rightarrow}\mathcal{N}%
(\mathbf{0},\mathbf{I}_{J+1}),
\]
where $\delta_{N,t}=(\delta_{N,jt},0\leq j\leq J)^{\intercal}$, $d_{NT}$ is
given in (\ref{dNT}), and $\mathbf{I}_{J+1}$ is the $(J+1)\times(J+1)$
identity matrix.
\end{theorem}


The next theorem establishes the rate of convergence of the final estimator
$\widehat{g}_{j}(x_{j})$.

\begin{theorem}
\label{THM:ghat}Suppose that the same conditions as given in Theorem
\ref{THM:fhat} hold. Then, for each $j$,
\begin{equation}
\left[  \int\{\widehat{g}_{j}(x_{j})-g_{j}^{0}(x_{j})\}^{2}dx_{j}\right]
^{1/2}=O_{p}(\phi_{NT})+o_{p}(N^{-1/2}), \label{l2rate}%
\end{equation}
where $\phi_{NT}$ is given in (\ref{EQ:phNT}).
\end{theorem}

\textbf{Remark 1:} \ The orders $\sqrt{K_{N}/(NT)}$ and
$K_{N}^{-r}$ are from the noise and bias terms for nonparametric estimation,
respectively, and the order $K_{N}^{3/2}N^{-3/4}\sqrt{\log NT}$ from the
approximation of the Bahadur representation in the quantile regression
setting. This says that if the order $K_{N}\asymp(NT)^{1/(2r+1)}$ is chosen,
and $T=O(N^{\alpha})$, where $\alpha<1/2-3/(2r+3)$, then the rate of $\phi_{NT}$ is $O_{P}((NT)^{-r/(2r+1)}),$ which is optimal, see for
example, Chen and Christensen (2015).

\textbf{Remark 2:} By using the asymptotic normality provided in Theorem
\ref{THM:fhat}, we can conduct inference for $f_{t}^{0}$ for each $t$, such as
constructing the confidence interval. Note that in the asymptotic distribution
in Theorem \ref{THM:fhat}, there is a bias term $d_{NT}\delta_{N,t}$ involved.
Let $K_{N}\asymp(NT)^{1/(2r+1)}$ and $T\asymp N^{\alpha}$, where
$1/(2r)<\alpha<1/2-3/(2r+3)$ and $r>3$. Then the asymptotic bias is negligible
and thus we have
\begin{equation}
\sqrt{N}(\mathbf{\Sigma}_{Nt}^{0})^{-1/2}(\widehat{f}_{t}-f_{t}^{0}%
)\rightarrow\mathcal{N}(\mathbf{0},\mathbf{I}_{J+1}). \label{EQ:CI}%
\end{equation}

\textbf{Remark 3}. It is possible to develop inferential results for $g_{j}$
following Chen and Liao (2012) and Chen and Pouzo (2015). As is usual in
nonparametric estimation, the weak cross-sectional and temporal dependence
does not affect the limiting distribution, and so standard techniques can be
applied. In fact, one may conclude the estimation algorithm with a kernel step
and demonstrate the oracle efficiency property, see Horowitz and Mammen (2011).

\renewcommand{\theequation}{5.\arabic{equation}}\setcounter{equation}{0}

\section{Covariance estimation and hypothesis testing for the
factors\label{covariance}}

In order to construct the confidence interval given in (\ref{EQ:CI}) we need
to estimate $\Omega_{N}^{0}$ and $\Lambda_{Nt}^{0}$, since they are unknown.
For estimation of $\Lambda_{Nt}^{0}$, if we use its sample analogue, the
conditional density $p_{i}\left(  0\left\vert X_{i},f_{t}\right.  \right)  $
needs to be estimated. Instead of using this direct way, we use the Powell's
kernel estimation idea in Powell (1991), and estimate $\Lambda_{Nt}^{0}$ by
\begin{equation}
\widehat{\Lambda}_{Nt}=(Nh)^{-1}\sum\nolimits_{i=1}^{N}K\left(  \frac
{y_{it}-\widehat{f}_{ut}-\sum\nolimits_{j=1}^{J}\widehat{g}_{j}(X_{ji}%
)\widehat{f}_{jt}}{h}\right)  \widehat{G}_{i}(X_{i})\widehat{G}_{i}%
(X_{i})^{\intercal}, \label{EQ:LamNhat}%
\end{equation}
where $\widehat{G}_{i}(X_{i})=\{1,\widehat{g}_{1}(X_{1i}),\ldots
,\widehat{g}_{J}(X_{Ji})\}^{\intercal},$ while $K(\cdot)$ is the uniform
kernel $K(u)=2^{-1}I(|u|\leq1)$ and $h$ is a bandwidth.

First, we show that the estimator $\widehat{\Lambda}_{Nt}$ is a consistent
estimator of $\Lambda_{Nt}^{0}$ given in the theorem below.

\begin{theorem}
\label{THM:consistentGAMMA}Suppose that the same conditions as given in
Theorem \ref{THM:fhat} hold, and $h\rightarrow0$, $h^{-1}\phi_{NT}=o(1)$,
$h^{-1}N^{-1/2}=O(1)$, where $\phi_{NT}$ is given in (\ref{EQ:phNT}). Then, we
have $||\widehat{\Lambda}_{Nt}-\Lambda_{Nt}^{0}||=o_{p}(1)$.
\end{theorem}

Moreover, the exact form of $\Omega_{N}^{0}$ defined in Condition (C6) is
given by{\small
\begin{align*}
\Omega_{N}^{0}  &  =(NT)^{-1}\sum\nolimits_{t=1}^{T}E\left[  \left\{
\sum\nolimits_{i=1}^{N}G_{i}^{0}(X_{i})(\tau-I(\varepsilon_{it}<0))\right\}
\left\{  \sum\nolimits_{i=1}^{N}G_{i}^{0}(X_{i})(\tau-I(\varepsilon
_{it}<0))\right\}  ^{\intercal}\right] \\
&  =\frac{\tau(1-\tau)}{N}\sum\nolimits_{i=1}^{N}E\{G_{i}^{0}(X_{i})G_{i}%
^{0}(X_{i})^{\intercal}\}+(NT)^{-1}\sum\nolimits_{t=1}^{T}\sum\nolimits_{i\neq
j}^{N}E(v_{it}v_{jt}^{\intercal}),
\end{align*}
where $v_{it}=G_{i}^{0}(X_{i})(\tau-I(\varepsilon_{it}<0))$} for
$i=1,\ldots,N$. To estimate $\Omega_{N}^{0}$, its sample analogue is not
consistent. Kernel-based robust estimators that account for HAC are developed
(Conley, 1999), and are shown to be consistent under a variety of sets of
conditions. It requires to use a truncation lag or \textquotedblleft
bandwidth\textquotedblright, which tends to infinity at a slower rate of $N$.
As pointed out by Kiefer and Vogelsang (2005), this is a convenient assumption
mathematically to ensure consistency, but it is unrealistic in finite sample
studies. Adopting the idea in Kiefer and Vogelsang (2005), we let the
bandwidth $M$ be proportional to the sample size $N$, i.e., $M=bN$ for
$b\in(0,1]$, and then we derive the fixed-b asymptotics (Kiefer and Vogelsang;
2005) for the HAC estimator of $\Omega_{N}^{0}$ under the quantile setting.
The HAC estimator is given as $\widehat{\Omega}_{N,M}=T^{-1}\sum
\nolimits_{t=1}^{T}\widehat{\Omega}_{Nt,M}$, where
\begin{equation}
\widehat{\Omega}_{Nt,M}=\frac{\tau(1-\tau)}{N}\sum\nolimits_{i=1}%
^{N}\widehat{G}_{i}(X_{i})\widehat{G}_{i}(X_{i})^{\intercal}+N^{-1}%
\sum\nolimits_{i\neq j}^{N}K^{\ast}\left(  \frac{i-j}{M}\right)
\widehat{v}_{it}\widehat{v}_{jt}^{\intercal}, \label{EQ:OmegaNt}%
\end{equation}
where $\widehat{v}_{it}=\widehat{G}_{i}(X_{i})(\tau-I(\widehat{\varepsilon
}_{it}<0))$ for $i=1,\ldots,N$, $\widehat{\varepsilon}_{it}=y_{it}%
-\widehat{f}_{ut}-\sum\nolimits_{j=1}^{J}\widehat{g}_{j}(X_{ji})\widehat{f}%
_{jt}$, $K^{\ast}(u)$ is a symmetric kernel weighting function satisfying
$K^{\ast}(0)=1$, and $|K^{\ast}(u)|\leq1$, and $M$ trims the sample
autocovariances and acts as a truncation lag. Consistency of $\widehat{\Omega
}_{N,M}$ needs that $M\rightarrow\infty$ and $M/N\rightarrow0$. \ The
following theorem provides the limiting distribution of $\widehat{\Omega
}_{N,M=bN}$ when $M=bN$ for $b\in(0,1]$.

Next, we will show asymptotic theory for the HAC covariance estimator under a
sequence where the smoothing parameter $M$ equals to $bN$. Let $\Omega
^{0}=\lim_{N\rightarrow\infty}\Omega_{N}^{0}$, and $\Omega^{0}$ can be written
as $\Omega^{0}=\Upsilon \Upsilon^{\intercal}$, where $\Upsilon$ is a lower
triangular matrix obtained from the Cholesky decomposition of $\Omega^{0}$.

\begin{theorem}
\label{THM:consistenOmega}Suppose that the same conditions as given in Theorem
\ref{THM:fhat} hold, and $\phi_{NT}N^{1/2}=o(1)$, and $K^{\ast\prime\prime
}(u)$ exists for $u\in\lbrack-1,1]$ and is continuous. Let $M=bN$ for
$b\in(0,1]$. Then as $N\rightarrow\infty$,%
\[
\widehat{\Omega}_{N,M=bN}\overset{\mathcal{D}}{\rightarrow}\Upsilon
\int\nolimits_{0}^{1}\int\nolimits_{0}^{1}-\frac{1}{b^{2}}K^{\ast\prime\prime
}\left(  \frac{r-s}{b}\right)  B_{J+1}(r)B_{J+1}(s)^{\intercal}drds\Upsilon
^{\intercal}\text{,}%
\]
where $B_{J+1}(r)=W_{J+1}(r)-rW_{J+1}(1)$ denotes a $(J+1)\times1$ vector of
standard Brownian bridges, and $W_{J+1}(r)$ denotes a $(J+1)$-vector of
independent standard Wiener processes where $r\in\lbrack0,1]$.
\end{theorem}

Theorem \ref{THM:consistenOmega} establishes the limiting distribution of
$\widehat{\Omega}_{N,M=bN}$, although $\widehat{\Omega}_{N,M=bN}$ is an
inconsistent estimator of $\Omega^{0}$. However, it can be used to construct
asymptotically pivotal tests involving $f_{t}^{0}$.

Consider testing the null hypothesis $H_{0}$: $Rf_{t}^{0}=r$ against the
alternative hypothesis $H_{1}$: $Rf_{t}^{0}\neq r$, where $R$ is a
$q\times(J+1)$ matrix with rank $q$ and $r$ is a $q\times1$ vector. We
construct an $F$-type statistic given as
\[
F_{Nt,b}=N(R\widehat{f}_{t}-r)^{\intercal}\{R\tau(1-\tau)\widehat{\Lambda
}_{Nt}^{-1}\widehat{\Omega}_{N,M=bN}\widehat{\Lambda}_{Nt}^{-1}R^{\intercal
}\}^{-1}(R\widehat{f}_{t}-r)/q.
\]
When $q=1$, we can construct a $t$-type statistic:%
\[
T_{Nt,b}=\frac{N^{1/2}(R\widehat{f}_{t}-r)}{\sqrt{R\tau(1-\tau
)\widehat{\Lambda}_{Nt}^{-1}\widehat{\Omega}_{N,M=bN}\widehat{\Lambda}%
_{Nt}^{-1}\}^{-1}R^{\intercal}}}.
\]
The limiting distributions of $F_{Nt,b}$ and $T_{Nt,b}$ under the null
hypothesis are given in the following theorem.

\begin{theorem}
\label{THM:null}Suppose that the same conditions as given in Theorem
\ref{THM:fhat} hold, and $\phi_{NT}N^{1/2}=o(1)$, and $K^{\ast\prime\prime
}(u)$ exists for $u\in\lbrack-1,1]$ and is continuous. Let $M=bN$ for
$b\in(0,1]$. Then under the null hypothesis $H_{0}$: $Rf_{t}^{0}=r$, as
$N\rightarrow\infty$,%
\[
F_{Nt,b}\overset{\mathcal{D}}{\rightarrow}\{ \tau(1-\tau)\}^{-1}%
W_{q}(1)^{\intercal}\left\{  \int\nolimits_{0}^{1}\int\nolimits_{0}^{1}%
-\frac{1}{b^{2}}K^{\ast\prime\prime}\left(  \frac{r-s}{b}\right)
B_{q}(r)B_{q}(s)^{\intercal}drds\right\}  ^{-1}W_{q}(1)/q\text{.}%
\]
If $q=1$, then as $N\rightarrow\infty$,%
\[
T_{Nt,b}\overset{\mathcal{D}}{\rightarrow}\frac{W_{1}(1)}{\sqrt{\tau(1-\tau
)}\sqrt{\int\nolimits_{0}^{1}\int\nolimits_{0}^{1}-\frac{1}{b^{2}}%
K^{\ast\prime\prime}\left(  \frac{r-s}{b}\right)  B_{1}(r)B_{1}(s)drds}}.
\]

\end{theorem}

Let $\Lambda_{t}^{0}=\lim_{N\rightarrow\infty}\Lambda_{Nt}^{0}$. The limiting
distributions of $F_{Nt,b}$ and $T_{Nt,b}$ under the alternative hypothesis
$H_{1}$: $Rf_{t}^{0}=r+cN^{-1/2}$ are given in the following theorem.

\begin{theorem}
\label{THM:alternative}Let $\Upsilon_{t}^{\ast}=(R\Lambda_{t}^{-1}\Omega
^{0}\Lambda_{t}^{-1}R^{\intercal})^{1/2}$. Suppose that the same conditions as
given in Theorem \ref{THM:fhat} hold, and $\phi_{NT}N^{1/2}=o(1)$, and
$K^{\ast\prime\prime}(u)$ exists for $u\in\lbrack-1,1]$ and is continuous. Let
$M=bN$ for $b\in(0,1]$. Then under the alternative hypothesis $H_{1}$:
$Rf_{t}^{0}=r+cN^{-1/2}$, as $N\rightarrow\infty$,%
\begin{align*}
&  F_{Nt,b}\overset{\mathcal{D}}{\rightarrow}\{ \tau(1-\tau)\}^{-1}\{
\Upsilon_{t}^{\ast-1}c+W_{q}(1)\}^{\intercal}\times\\
&  \left\{  \int\nolimits_{0}^{1}\int\nolimits_{0}^{1}-\frac{1}{b^{2}}%
K^{\ast\prime\prime}\left(  \frac{r-s}{b}\right)  B_{q}(r)B_{q}(s)^{\intercal
}drds\right\}  ^{-1}\{ \Upsilon_{t}^{\ast-1}c+W_{q}(1)\}/q\text{.}%
\end{align*}
If $q=1$, then as $N\rightarrow\infty$,%
\[
T_{Nt,b}\overset{\mathcal{D}}{\rightarrow}\frac{\Upsilon_{t}^{\ast-1}%
c+W_{1}(1)}{\sqrt{\tau(1-\tau)}\sqrt{\int\nolimits_{0}^{1}\int\nolimits_{0}%
^{1}-\frac{1}{b^{2}}K^{\ast\prime\prime}\left(  \frac{r-s}{b}\right)
B_{1}(r)B_{1}(s)drds}}.
\]

\end{theorem}

\textbf{Remark. }If $K^{\ast}(x)$ is the Bartlett kernel, then
\begin{align*}
&  \int\nolimits_{0}^{1}\int\nolimits_{0}^{1}-\frac{1}{b^{2}}K^{\ast
\prime\prime}\left(  \frac{r-s}{b}\right)  B_{q}(r)B_{q}(s)^{\intercal}drds\\
=\frac{2}{b}\int\nolimits_{0}^{1}B_{q}(r)B_{q}(r)^{\intercal}dr  &  -\frac
{1}{b}\int\nolimits_{0}^{1-b}\{B_{q}(r+b)B_{q}(r)^{\intercal}+B_{q}%
(r)B_{q}(r+b)^{\intercal}\}dr.
\end{align*}
These results allow one to test whether the factors are zero in a particular
time period or not. Our tests are robust to the form of the cross-sectional
dependence in the idiosyncratic error.

\renewcommand{\theequation}{6.\arabic{equation}}

\setcounter{equation}{0}

\section{Application\label{application}}

In a series of important papers, Fama and French (hereafter denoted FF),
demonstrated that there have been large return
premia associated with size and value, which are observable characteristics of stocks.
They contended that these return premia can be ascribed to a rational asset pricing
paradigm in which the size and value characteristics proxy for assets'
sensitivities to pervasive sources of risk in the economy. FF (1993) used a simple portfolio sorting approach to estimating their factor model.
Connor, Hagmann, and Linton (2012) used kernel-based semiparametric regression methodology to capture the same phenomenon.

In our data analysis, we use all securities from Center for Research in
Security Prices (CRSP) which have complete daily return records from 2005 to
2013, and have two-digit Standard Industrial Classification code (from CRSP),
market capitalization (from Compustat) and book value (from Compustat)
records. We use daily returns in excess of the risk-free return of 347 stocks.
We consider the same four characteristic variables as given in Connor,
Hagmann and Linton (2012), and Fan, Liao and Wang (2016), which are size,
value, momentum and volatility. Connor, Hagmann and Linton (2012) provided
some detailed descriptions of these characteristics. They are calculated using
the same method as described in Fan, Liao and Wang (2016).

We fit the quantile factor model (\ref{EQ:factorquantile}) for each year, so
that there are $T=251$ observations. By taking the same strategy as He and Shi (1996), we select the number of interior knots $L_{N}$ by minimizing the
Bayesian information criterion (BIC) given as{\small
\[
\text{BIC(}L_{N})=\log\{(NT)^{-1}\sum\nolimits_{i=1}^{N}\sum\nolimits_{t=1}%
^{T}\rho_{\tau}(y_{it}-\widehat{f}_{ut}-\sum\nolimits_{j=1}^{J}\widehat{g}%
_{j}(X_{ji})\widehat{f}_{jt})\}+\frac{\log(NT)}{2NT}J(L_{N}+m).
\]
}For the estimator $\widehat{\Lambda}_{Nt}$ given in (\ref{EQ:LamNhat}), the
optimal order for the bandwidth $h$ is in the order of $N^{-1/5}$. We let $h=\kappa N^{-1/5}$ in our numerical analysis and
take different values for $\kappa$. For the estimator $\widehat{\Omega
}_{Nt,M=bN}$ given in (\ref{EQ:OmegaNt}), we use different values for $b$, and
use the Bartlett kernel as suggested in Kiefer and Vogelsang (2005).

Figures \ref{FIG:fig1}-\ref{FIG:fig3} show the plots of the four estimated
loading functions for the year of 2009, 2010, 2011, and 2012 at different
quantiles $\tau=0.2$, $0.5$ and $0.8$. We observe that the estimated loading
functions have similar shapes for these four years. Moreover, for the size,
value and momentum characteristics, the estimated functions show a clear
nonlinear pattern, and at different quantiles, the curves are different for
the same characteristic. For example, for the size characteristic, the
estimated loading function fluctuates around zero and it has a sharp drop
after the value of size variable exceeds certain value at the quantiles
$\tau=0.2$ and $0.8$. However, it has a smooth decreasing pattern for the
median with $\tau=0.5$. For the momentum characteristic, the estimated
function shows different curves at the three quantiles.

\begin{figure}
\caption{\small{The plots of the estimated loading functions for the year of 2009
(dotted-dashed red lines), 2010 (dotted magenta lines), 2011 (dashed blue
lines), and 2012 (solid black lines) at $\tau=0.2$.}}%
\label{FIG:fig1}
\centering
{\normalsize $%
\begin{array}
[c]{cc}%
\includegraphics[width=7cm,height=6cm]{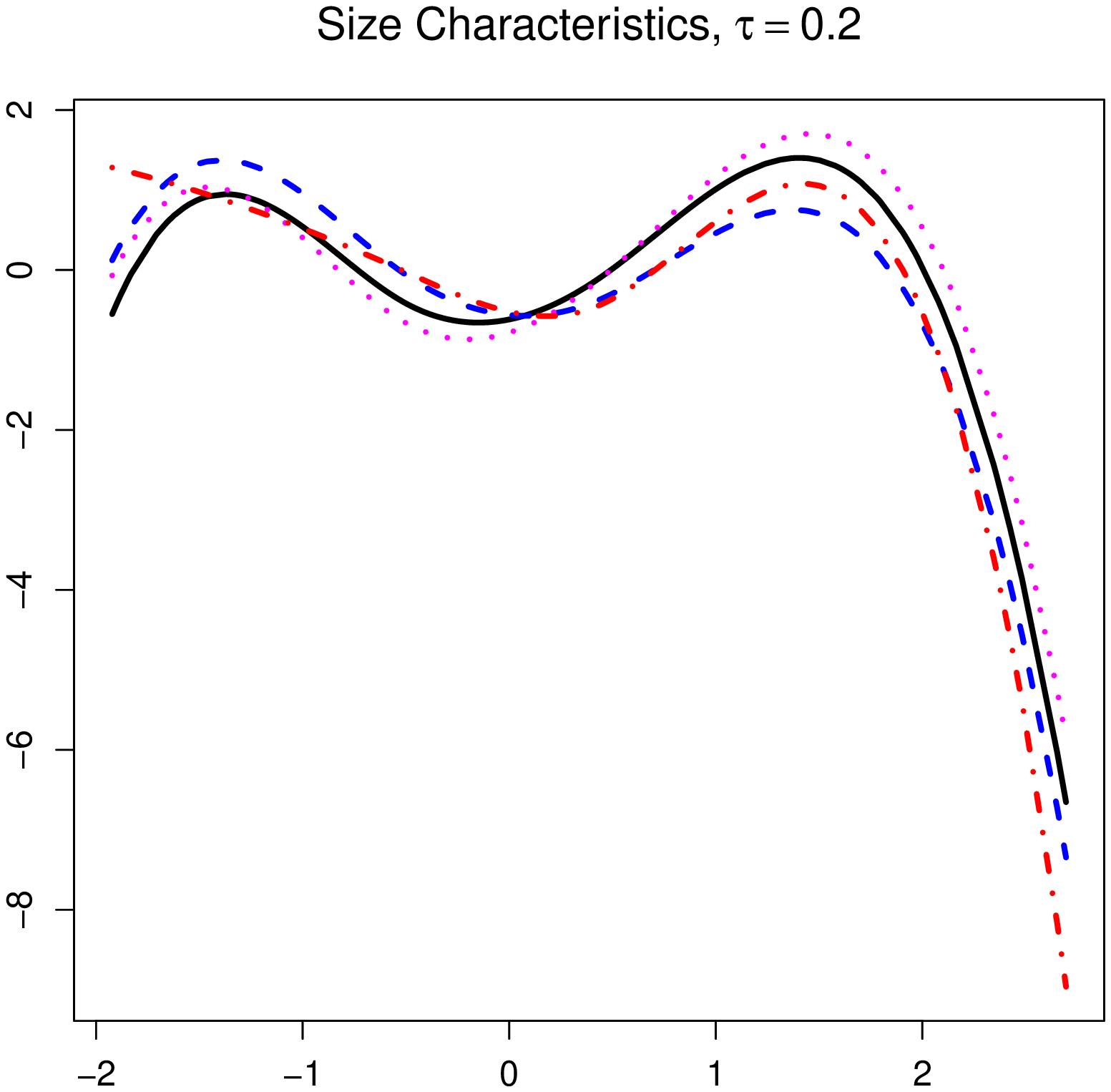} &
\includegraphics[width=7cm,height=6cm]{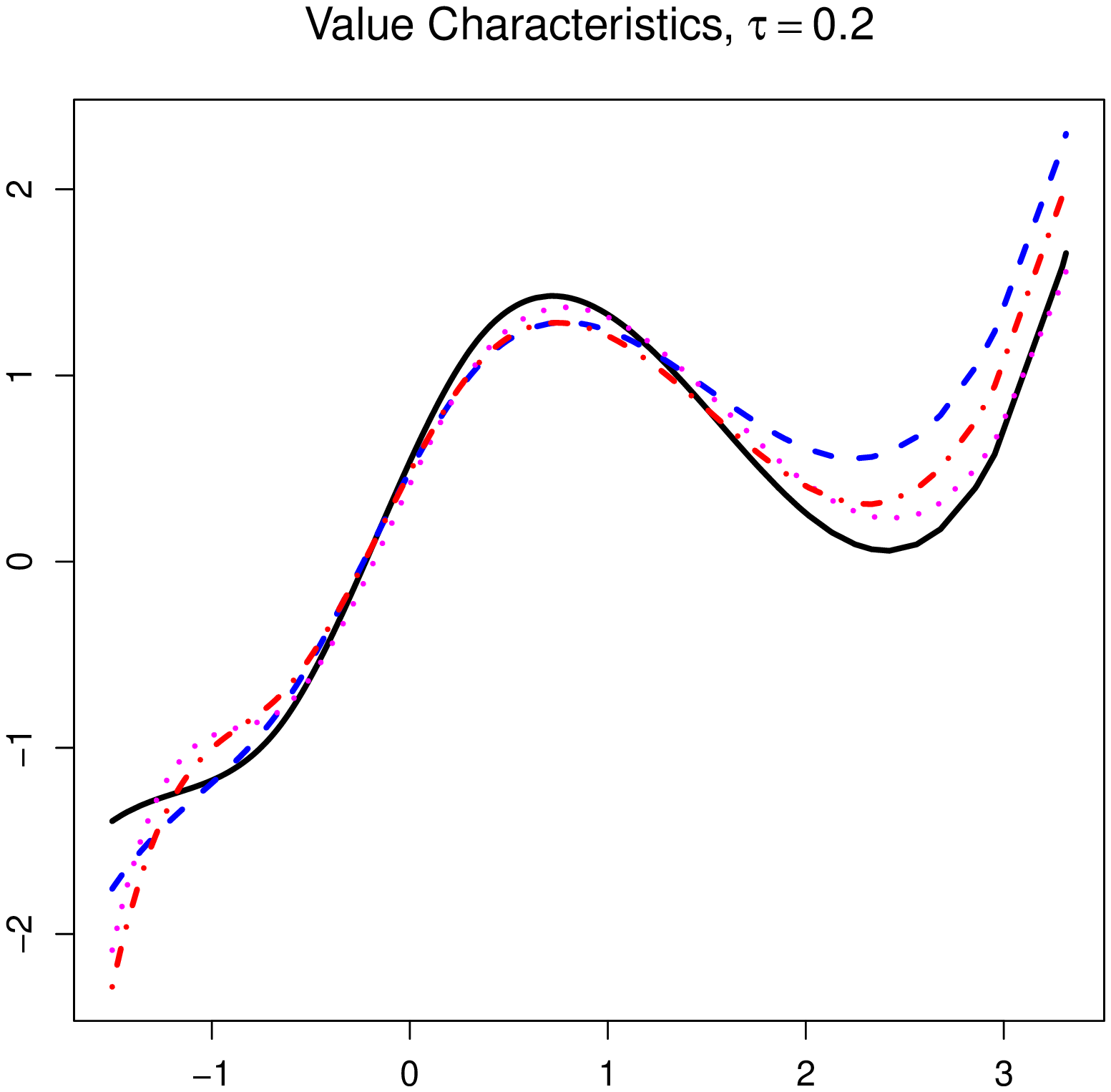}\\
\includegraphics[width=7cm,height=6cm]{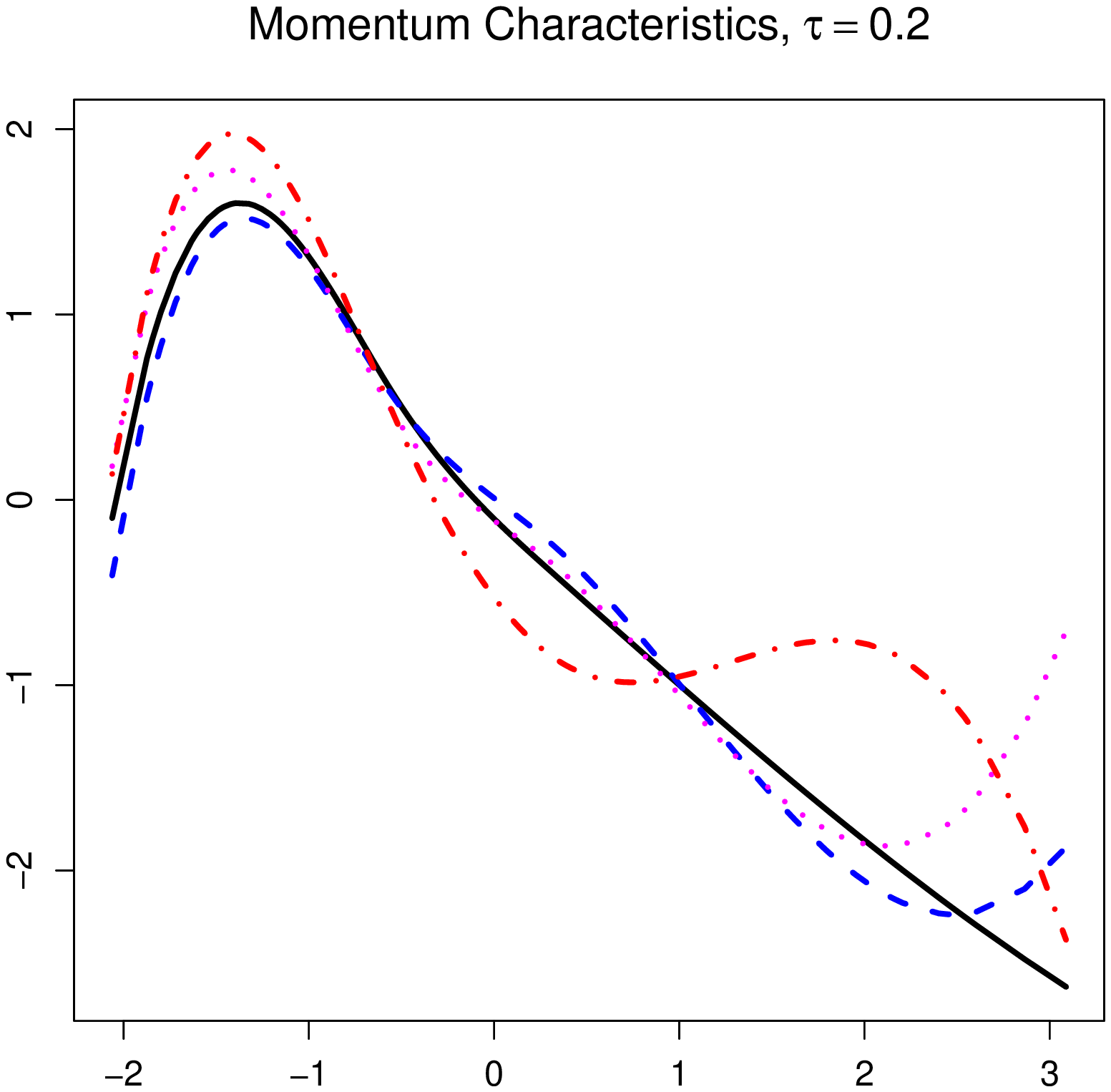} &
\includegraphics[width=7cm,height=6cm]{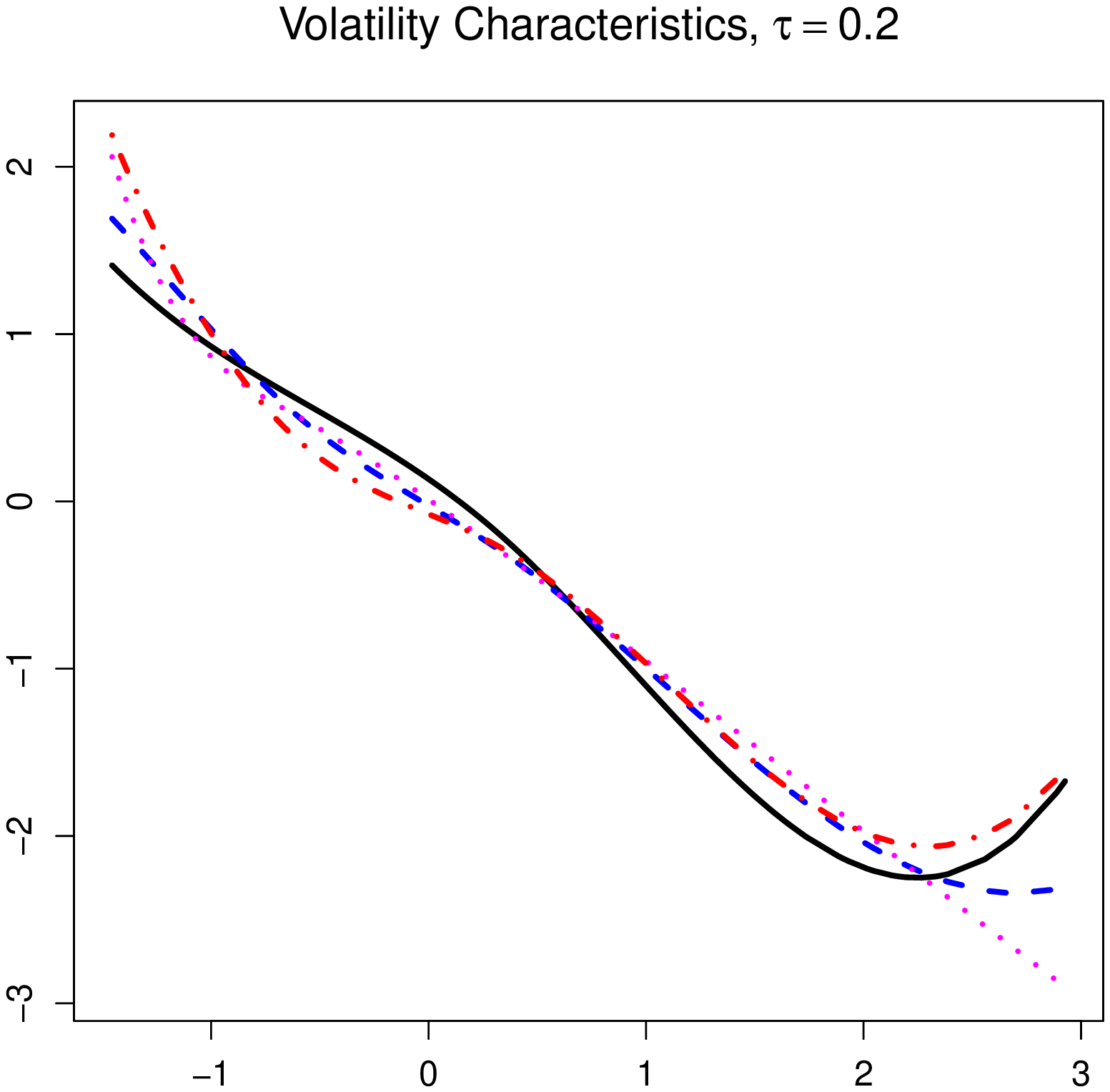}
\end{array}
$}
\end{figure}


\begin{figure}
\vspace{-1cm}
\caption{\small{The plots of the estimated loading functions for the year of 2009
(dotted-dashed red lines), 2010 (dotted magenta lines), 2011 (dashed blue
lines), and 2012 (solid black lines) at $\tau=0.5$.}}%
\label{FIG:fig2}
\centering
{\normalsize $%
\begin{array}
[c]{cc}%
\includegraphics[width=7cm,height=6cm]{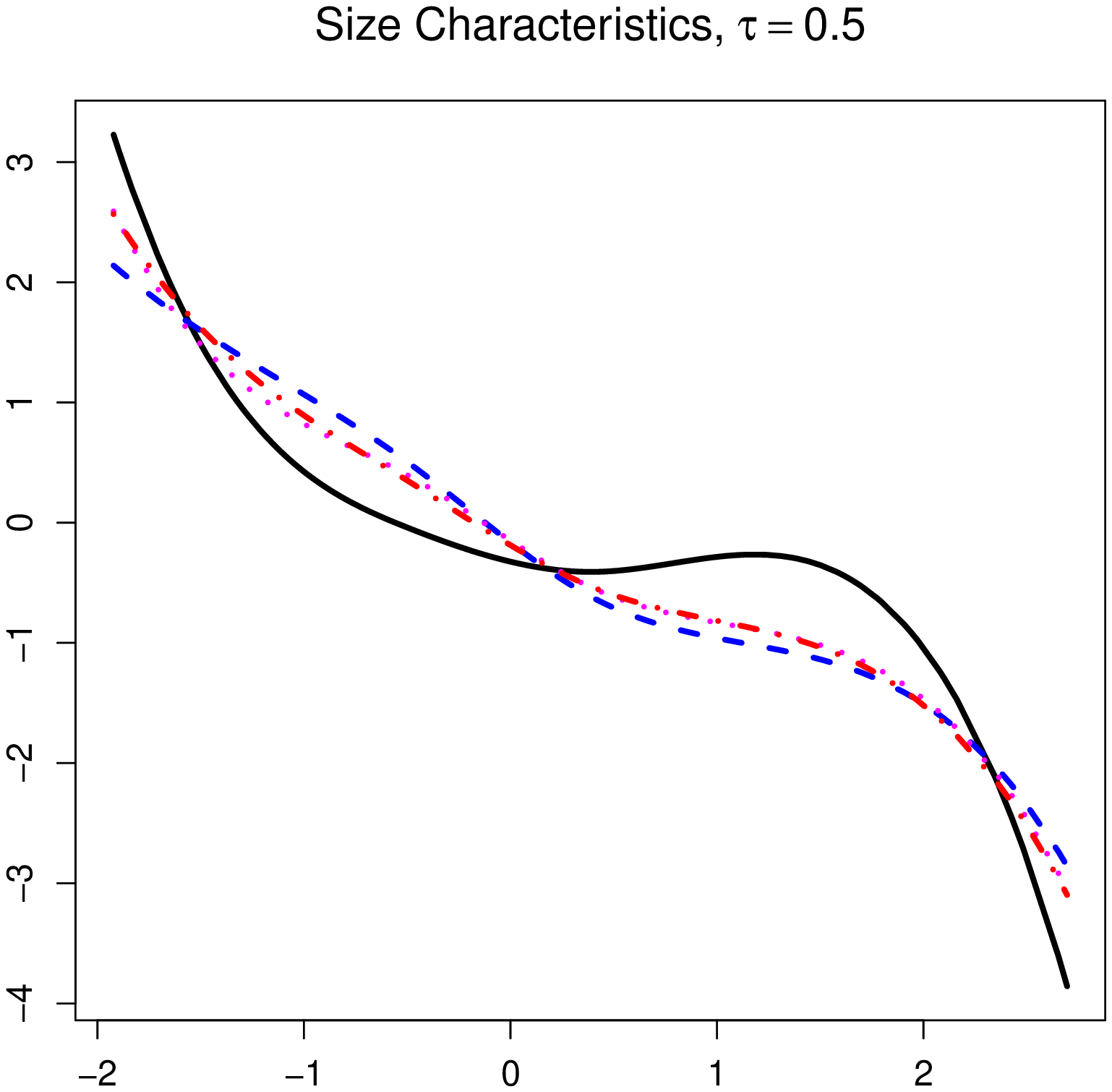} &
\includegraphics[width=7cm,height=6cm]{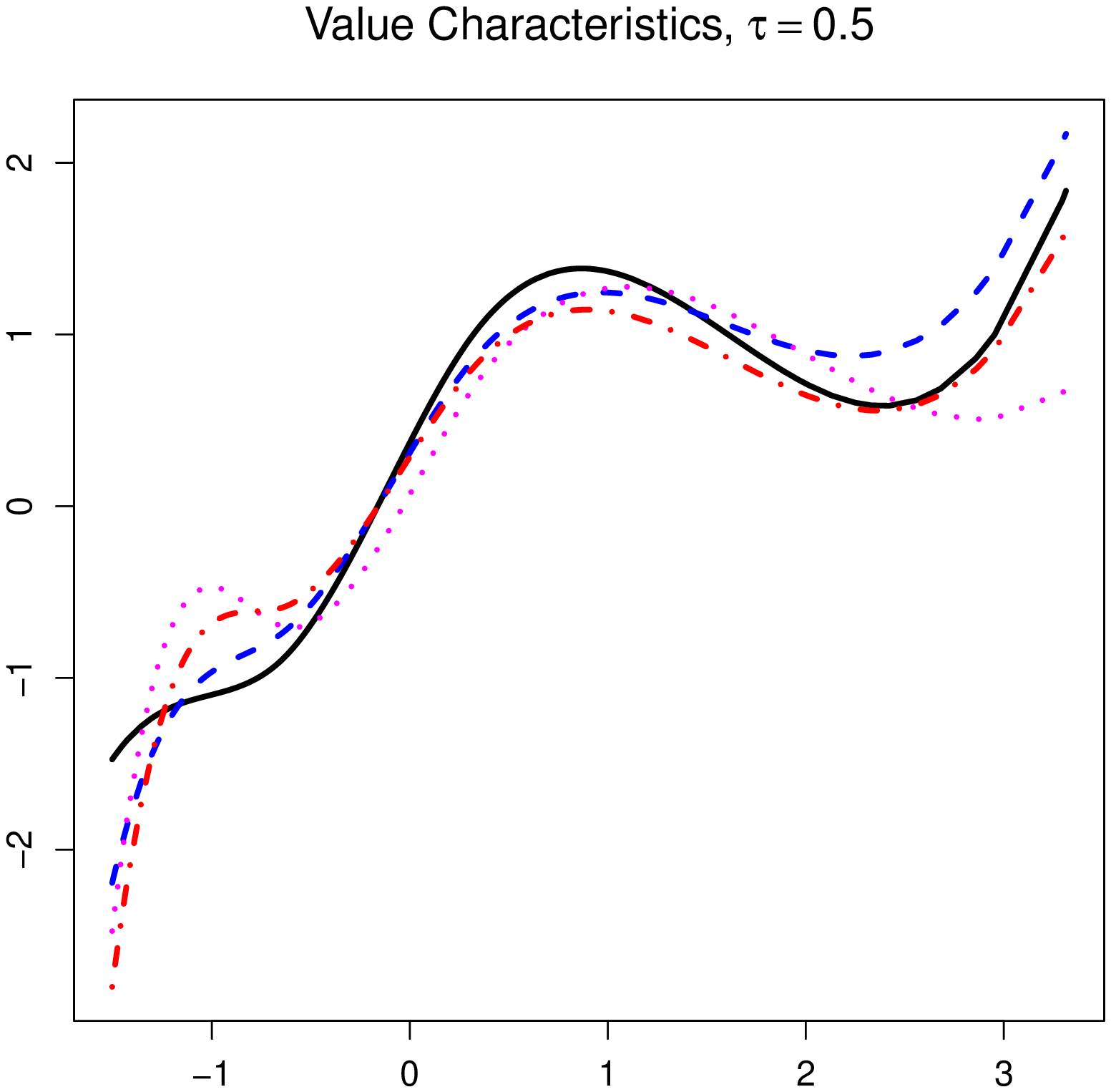}\\
\includegraphics[width=7cm,height=6cm]{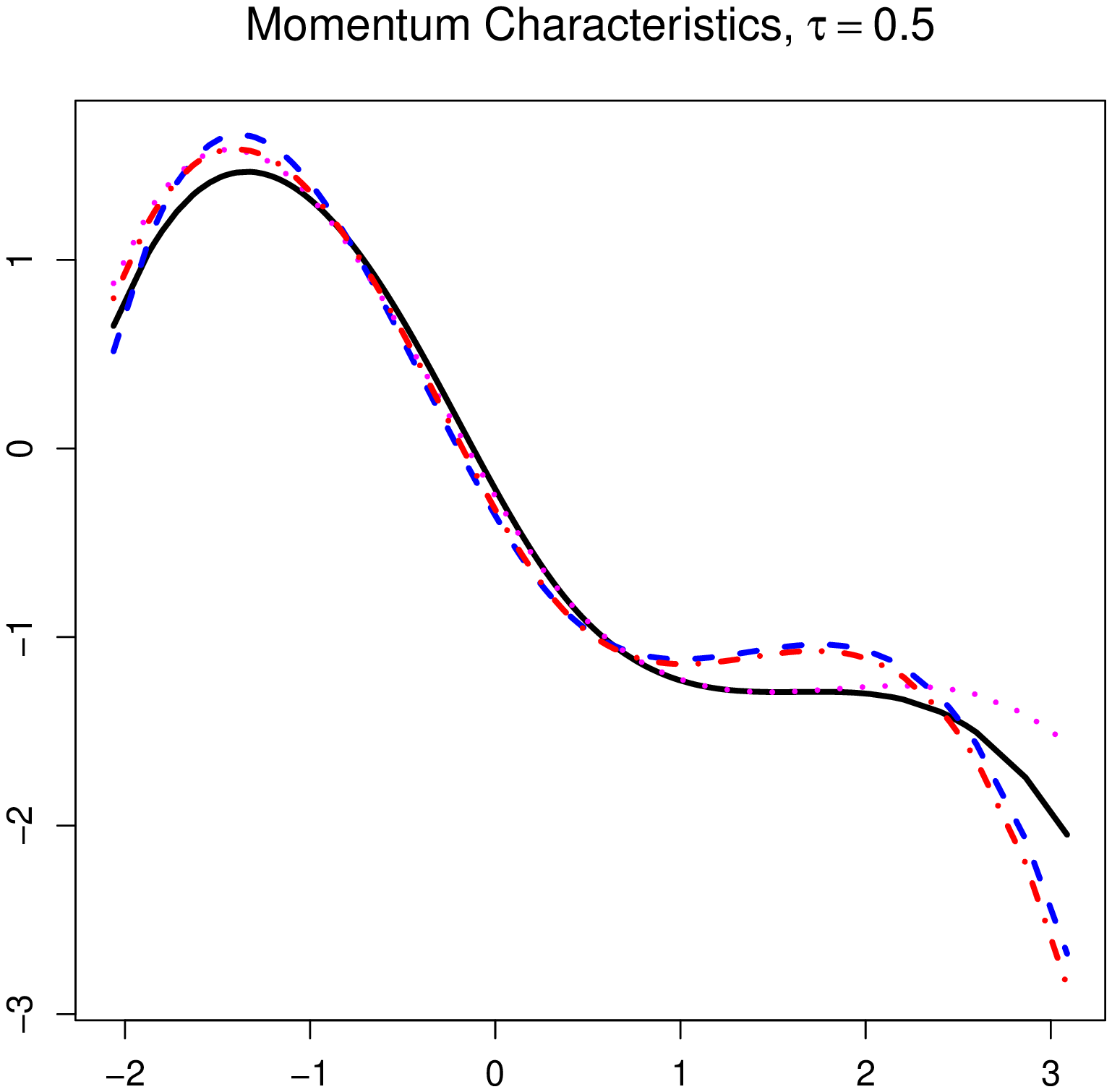} &
\includegraphics[width=7cm,height=6cm]{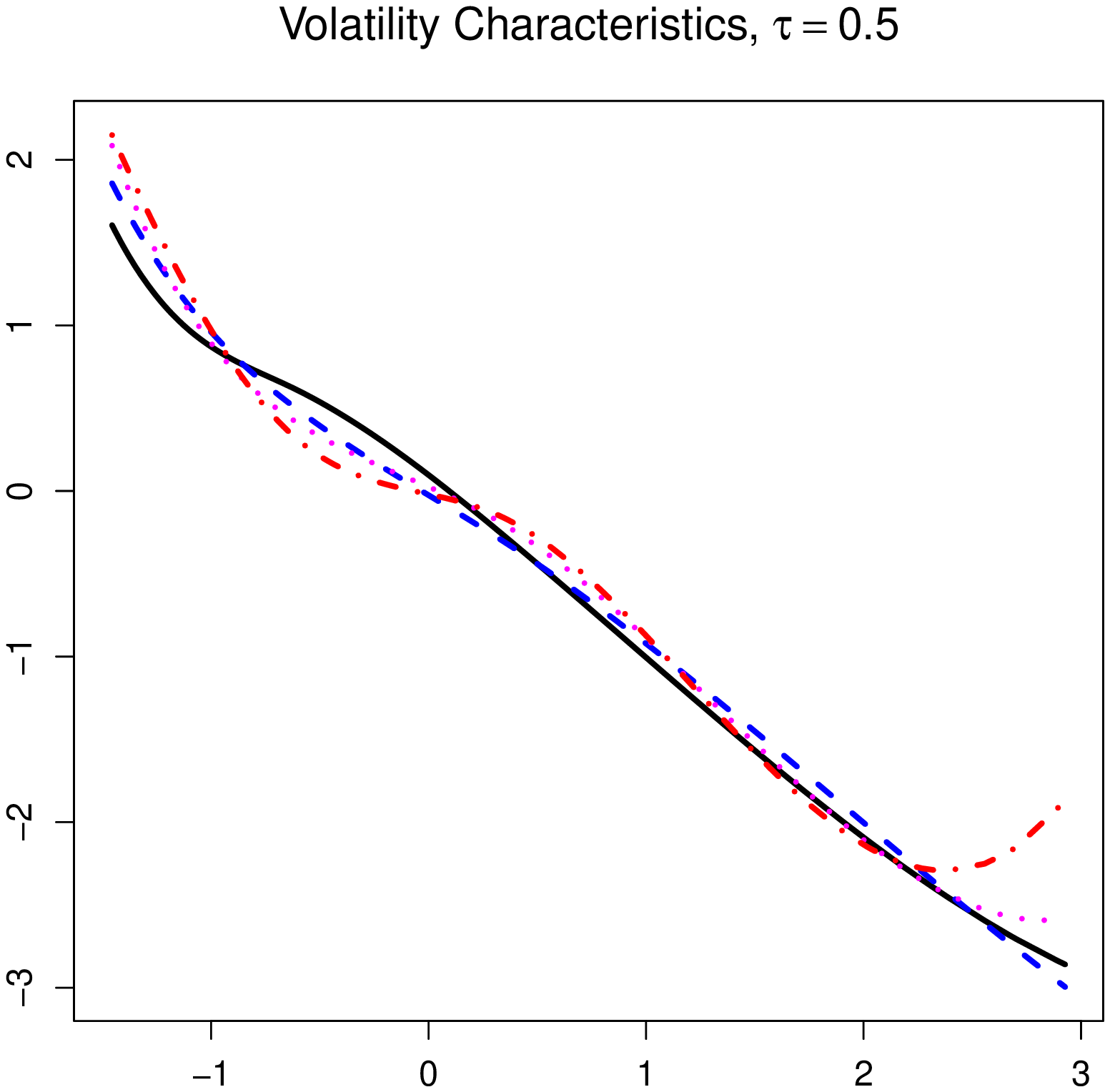}
\end{array}
$}

\vspace*{\floatsep}
\vspace*{-1.5cm}

\caption{\small{The plots of the estimated loading functions for the year of 2009
(dotted-dashed red lines), 2010 (dotted magenta lines), 2011 (dashed blue
lines), and 2012 (solid black lines) at $\tau=0.8$.}}%
\label{FIG:fig3}
\centering
{\normalsize $%
\begin{array}
[c]{cc}%
\includegraphics[width=7cm,height=6cm]{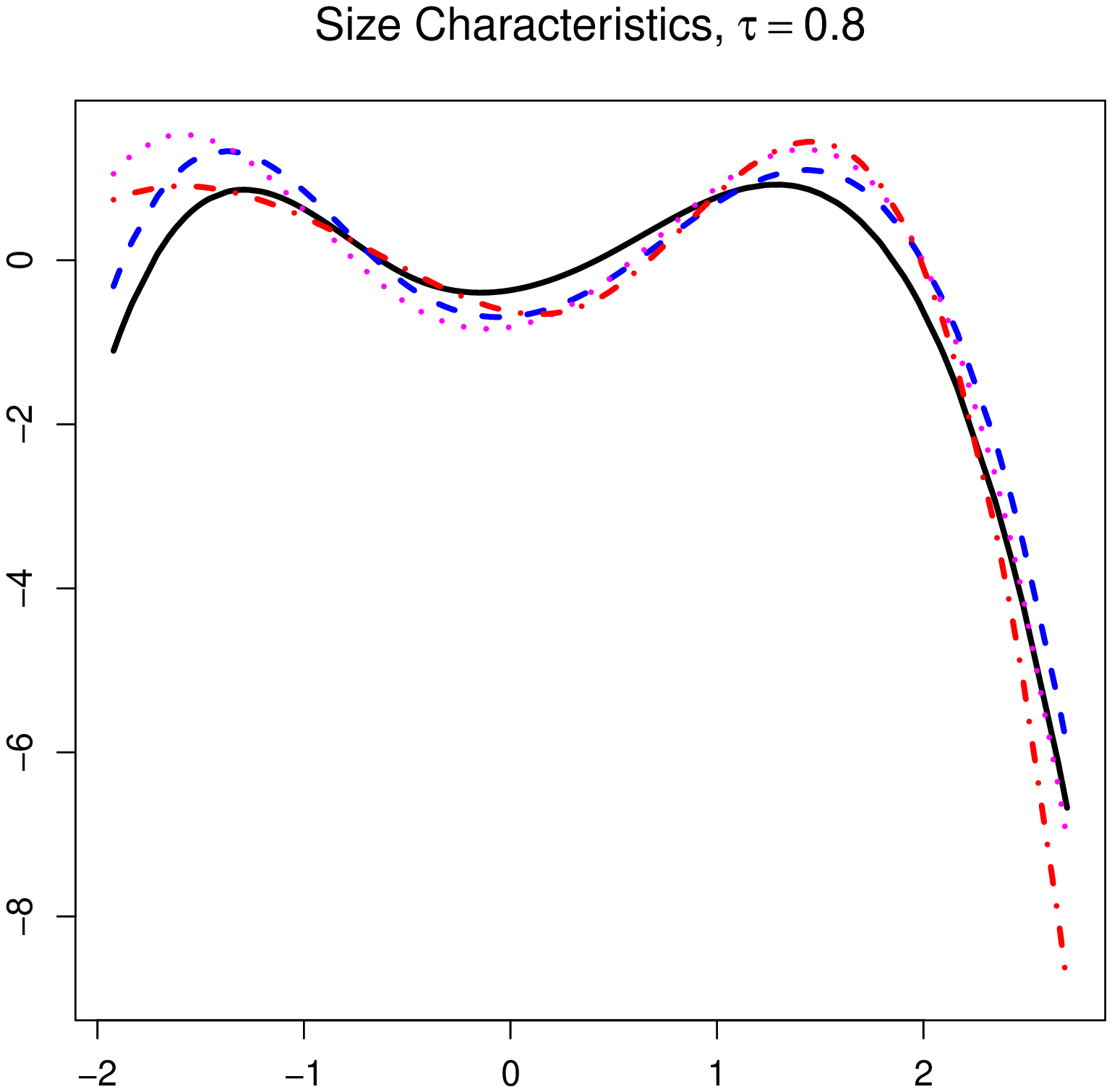} &
\includegraphics[width=7cm,height=6cm]{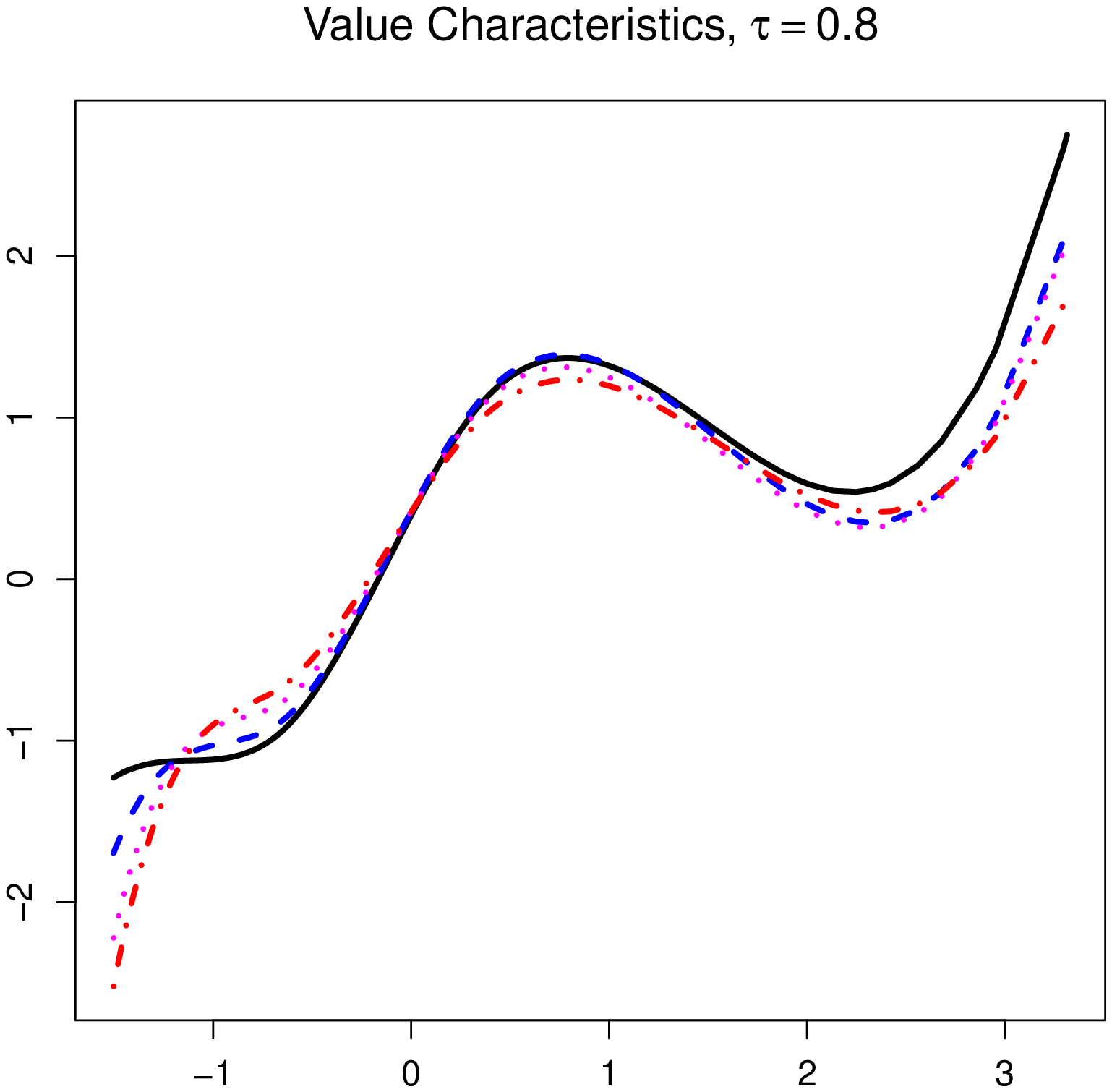}\\
\includegraphics[width=7cm,height=6cm]{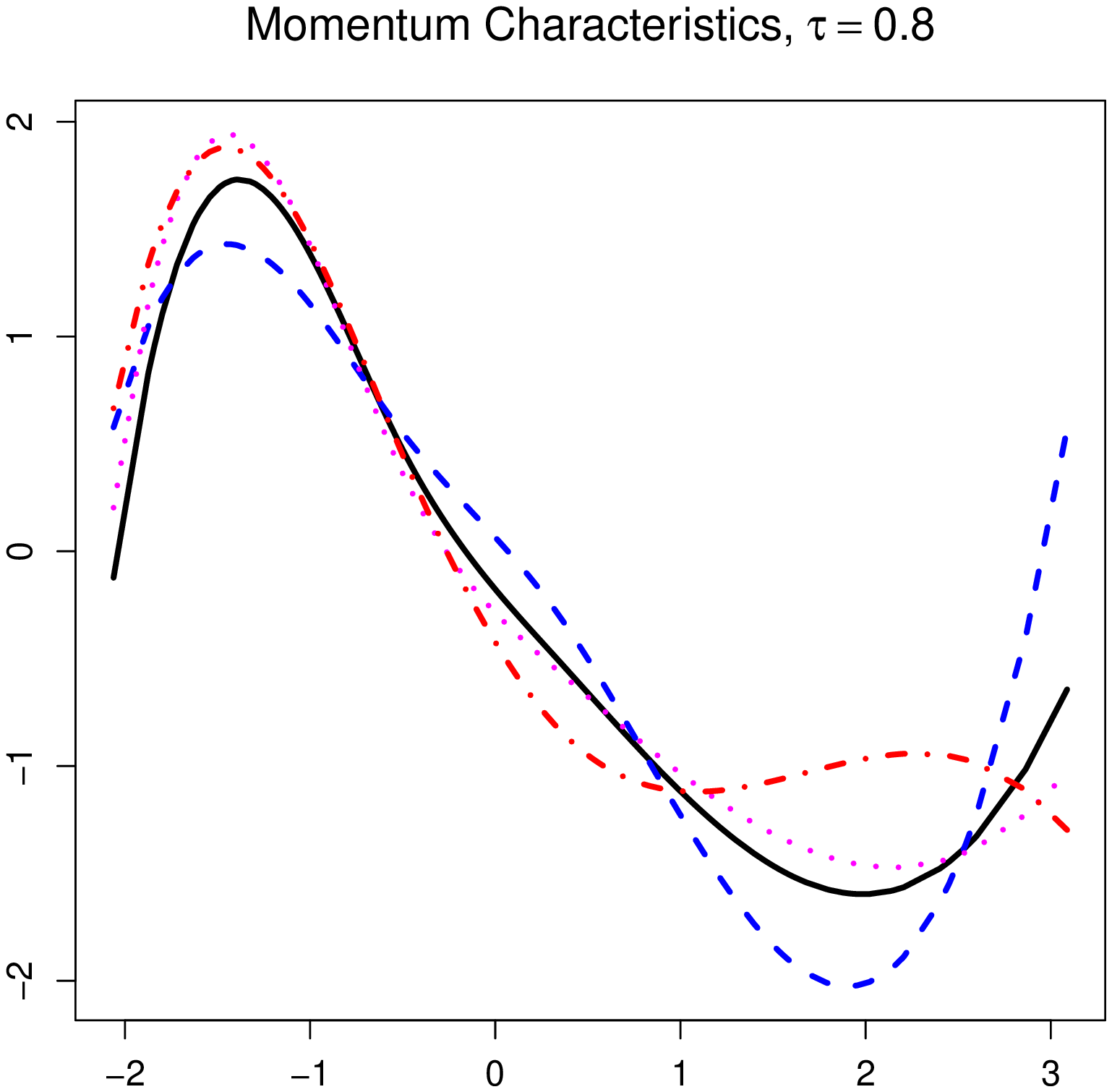} &
\includegraphics[width=7cm,height=6cm]{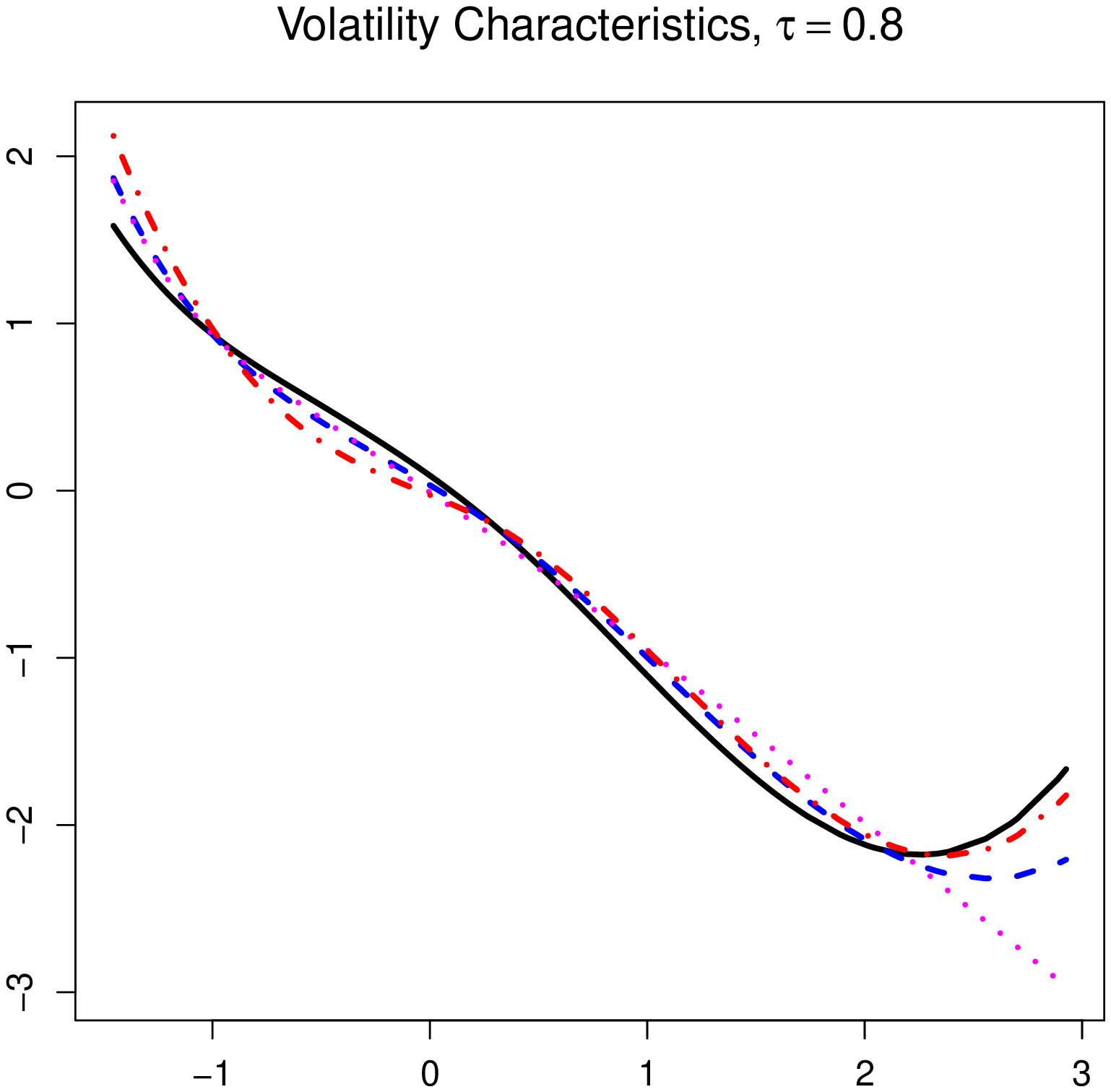}
\end{array}
$ }
\end{figure}

Next, we let $\kappa=0.5,1,1.5$ and $b=0.2,0.4,0.6$, respectively, for
calculation of $\widehat{\Lambda}_{Nt}$ and $\widehat{\Omega}_{Nt,M=bN}$. For
obtaining the robust estimator $\widehat{\Omega}_{Nt,M=bN}$, the data need to
be ordered across $i$. We consider two different orderings. First, we take the
same strategy as Lee and Robinson (2016) by ordering the data according to
firm size, since firms of similar size may be subject to similar shocks.
Second, we use the information of the four explanatory variables by ordering
the data according to the first principal component of the covariate matrix.
Using the year of 2012, we test for the statistical significance of each
factor at each time point, based on the proposed $t$-type statistic and its
distribution given in Theorem \ref{THM:null}. Then for each factor, we find
the percentage of the $t$-type statistics that are significant at a $95\%$
confidence level across the 251 time periods. Based on the two different
ordering strategies, Tables 1 and 2,
respectively, show the annualized standard deviations of the factor returns,
the percentage of significant $t$-type statistics for each factor, and the
median p-value at $\tau=0.5$. We can see that the results are consistent for
different values of $\kappa$ and $b$ and for the two different orderings of
the data. Moreover, all five factors are statistically significant with the
median p-value smaller than 0.05.

{\samepage{\small \begin{table}
\caption{Factor return statistics at $\tau=0.5$ for the year of 2012
when the data are ordered according to the firm size.} %
\label{TAB:factor}
\centering
{\small \hspace{-0.5cm}
\begin{tabular*}
{0.95\textwidth}[c]{ll@{\extracolsep{\fill}}ccccc}\hline
$(\kappa,b)$ &  & Intercept & Size & Value & Momentum & Volatility\\\hline
& Annualized volatility & $0.024$ & $0.023$ & $0.025$ & $0.025$ & $0.029$\\
$(0.5,0.2)$ & \% Periods significant & $92.43$ & $65.34$ & $62.95$ & $64.54$ &
$74.10$\\
& Overall p-value & $<0.001$ & $0.010$ & $0.016$ & $0.009$ & $0.001$\\\hline
& Annualized volatility & $0.020$ & $0.020$ & $0.022$ & $0.022$ & $0.026$\\
$(0.5,0.4)$ & \% Periods significant & $91.63$ & $58.17$ & $57.20$ & $58.17$ &
$66.93$\\
& Overall p-value & $<0.001$ & $0.020$ & $0.019$ & $0.020$ & $0.010$\\\hline
& Annualized volatility & $0.018$ & $0.019$ & $0.019$ & $0.020$ & $0.023$\\
$(0.5,0.6)$ & \% Periods significant & $90.84$ & $55.78$ & $56.40$ & $55.38$ &
$66.93$\\
& Overal p-value & $<0.001$ & $0.032$ & $0.028$ & $0.023$ & $0.006$\\\hline
& Annualized volatility & $0.026$ & $0.026$ & $0.027$ & $0.027$ & $0.032$\\
$(1.0,0.2)$ & \% Periods significant & $92.03$ & $62.95$ & $63.60$ & $62.15$ &
$71.31$\\
& Overall p-value & $<0.001$ & $0.014$ & $0.019$ & $0.011$ & $0.002$\\\hline
& Annualized volatility & $0.022$ & $0.023$ & $0.023$ & $0.024$ & $0.028$\\
$(1.0,0.4)$ & \% Periods significant & $90.44$ & $55.20$ & $56.40$ & $55.98$ &
$65.74$\\
& Overall p-value & $<0.001$ & $0.036$ & $0.030$ & $0.033$ & $0.011$\\\hline
& Annualized volatility & $0.019$ & $0.022$ & $0.021$ & $0.021$ & $0.025$\\
$(1.0,0.6)$ & \% Periods significant & $89.24$ & $56.20$ & $55.40$ & $58.80$ &
$62.95$\\
& Overall p-value & $<0.001$ & $0.032$ & $0.032$ & $0.026$ & $0.016$\\\hline
& Annualized volatility & $0.027$ & $0.028$ & $0.029$ & $0.029$ & $0.034$\\
$(1.5,0.2)$ & \% Periods significant & $92.03$ & $59.76$ & $55.38$ & $61.75$ &
$70.12$\\
& Overall p-value & $<0.001$ & $0.021$ & $0.032$ & $0.015$ & $0.003$\\\hline
& Annualized volatility & $0.023$ & $0.025$ & $0.025$ & $0.026$ & $0.031$\\
$(1.5,0.4)$ & \% Periods significant & $90.44$ & $56.57$ & $55.94$ & $55.94$ &
$63.75$\\
& Overall p-value & $<0.001$ & $0.030$ & $0.030$ & $0.036$ & $0.014$\\\hline
& Annualized volatility & $0.020$ & $0.019$ & $0.022$ & $0.022$ & $0.026$\\
$(1.5,0.6)$ & \% Periods significant & $88.44$ & $58.14$ & $56.80$ & $56.00$ &
$61.75$\\
& Overall p-value & $<0.001$ & $0.027$ & $0.028$ & $0.024$ & $0.018$\\\hline
\end{tabular*}
}
\end{table}

\begin{table}
\caption{Factor return statistics at $\tau=0.5$ for the year of 2012
when the data are ordered according to the first principal
component of the covariate matrix.}%
\label{TAB:factor2}
\centering
{\small \hspace{-0.5cm}
\begin{tabular*}
{0.95\textwidth}[c]{ll@{\extracolsep{\fill}}ccccc}\hline
$(\kappa,b)$ &  & Intercept & Size & Value & Momentum & Volatility\\\hline
& Annualized volatility & $0.023$ & $0.027$ & $0.025$ & $0.025$ & $0.027$\\
$(0.5,0.2)$ & \% Periods significant & $94.02$ & $62.15$ & $62.55$ & $67.73$ &
$75.30$\\
& Overall p-value & $<0.001$ & $0.023$ & $0.018$ & $0.011$ & $<0.001$\\\hline
& Annualized volatility & $0.019$ & $0.024$ & $0.022$ & $0.021$ & $0.023$\\
$(0.5,0.4)$ & \% Periods significant & $92.43$ & $57.60$ & $54.20$ & $58.96$ &
$70.92$\\
& Overall p-value & $<0.001$ & $0.023$ & $0.032$ & $0.019$ & $0.001$\\\hline
& Annualized volatility & $0.016$ & $0.021$ & $0.020$ & $0.019$ & $0.020$\\
$(0.5,0.6)$ & \% Periods significant & $92.83$ & $55.60$ & $56.40$ & $61.60$ &
$71.31$\\
& Overal p-value & $<0.001$ & $0.028$ & $0.028$ & $0.018$ & $0.004$\\\hline
& Annualized volatility & $0.025$ & $0.031$ & $0.027$ & $0.027$ & $0.030$\\
$(1.0,0.2)$ & \% Periods significant & $93.23$ & $56.18$ & $58.96$ & $66.14$ &
$73.71$\\
& Overall p-value & $<0.001$ & $0.036$ & $0.023$ & $0.014$ & $0.002$\\\hline
& Annualized volatility & $0.020$ & $0.026$ & $0.024$ & $0.024$ & $0.025$\\
$(1.0,0.4)$ & \% Periods significant & $92.03$ & $54.80$ & $56.20$ & $59.60$ &
$71.20$\\
& Overall p-value & $<0.001$ & $0.030$ & $0.030$ & $0.019$ & $0.002$\\\hline
& Annualized volatility & $0.016$ & $0.024$ & $0.022$ & $0.022$ & $0.022$\\
$(1.0,0.6)$ & \% Periods significant & $92.80$ & $56.20$ & $55.40$ & $56.80$ &
$68.80$\\
& Overall p-value & $<0.001$ & $0.027$ & $0.031$ & $0.029$ & $0.002$\\\hline
& Annualized volatility & $0.027$ & $0.030$ & $0.029$ & $0.028$ & $0.032$\\
$(1.5,0.2)$ & \% Periods significant & $92.03$ & $56.00$ & $54.40$ & $68.00$ &
$74.00$\\
& Overall p-value & $<0.001$ & $0.033$ & $0.032$ & $0.013$ & $0.002$\\\hline
& Annualized volatility & $0.021$ & $0.028$ & $0.026$ & $0.026$ & $0.026$\\
$(1.5,0.4)$ & \% Periods significant & $92.03$ & $56.60$ & $55.90$ & $55.20$ &
$68.00$\\
& Overall p-value & $<0.001$ & $0.028$ & $0.028$ & $0.030$ & $0.002$\\\hline
& Annualized volatility & $0.018$ & $0.025$ & $0.024$ & $0.023$ & $0.024$\\
$(1.5,0.6)$ & \% Periods significant & $92.03$ & $58.10$ & $54.80$ & $56.00$ &
$67.60$\\
& Overall p-value & $<0.001$ & $0.027$ & $0.030$ & $0.029$ & $0.003$\\\hline
\end{tabular*}
}
\end{table}}}

\renewcommand{\theequation}{7.\arabic{equation}}

\setcounter{equation}{0}

\section{Conclusions and discussion}

We have taken for granted that the $J$ factors are present in the sense that
\begin{equation}
\underset{T\rightarrow\infty}{\mathrm{p}\lim}\frac{1}{T}\sum_{t=1}^{T}%
f_{jt}^{0}\neq0 \label{mf}%
\end{equation}
for $j=1,\ldots,J$. For the factors in our application this is quite a
standard assumption, but in some cases one might wish to test this because if
this condition fails, then the right hand side of (\ref{m0}) is close to zero
and this equation cannot identify $g_{j}^{0}(x_{j}).$ We outline below a test
of the hypothesis (\ref{mf}) based on the unstructured additive quantile
regression model (\ref{EQ:quantileinitial}). A more limited objective is to
test whether for a given time period $t,$ $f_{jt}=0$.

We are interested in testing the hypothesis that
\begin{equation}
H_{0_{A_{j}}}:\lim_{T\rightarrow\infty}\frac{1}{T}\sum_{t=1}^{T}h_{jt}%
(x_{j})=0\text{ for all }x_{j}, \label{ho}%
\end{equation}
against the general alternative that $\lim_{T\rightarrow\infty}\frac{1}{T}%
\sum_{t=1}^{T}h_{jt}(x_{j})=\mu_{j}(x_{j})$ with $\int\mu_{j}(x_{j})^{2}%
dP_{j}(x_{j})>0.$ We also may be interested in a joint test $H_{0}=\cap_{j\in
I_{J}}H_{0_{A_{j}}},$ where $I_{J}$ is a set of integers, which is a subset of
$\{1,2,\ldots,J\}.$ These are tests of the presence of a factor.

We let%
\[
\widehat{\tau}_{j,N,T}=\frac{\int\left(  \frac{1}{T}\sum_{t=1}^{T}%
\widehat{h}_{jt}(x_{j})\right)  ^{2}dP_{j}(x_{j})-a_{N,T}}{s_{N,T}},
\]
where $\widehat{h}_{jt}(\cdot)$ is an estimator of the additive component
function $h_{jt}(\cdot)$ from the quantile additive model at time $t,$ while
$a_{N,T}$ and $s_{N,T}$ are constants to be determined. Under the null
hypothesis (\ref{ho}) we may show that%
\[
\widehat{\tau}_{j,n,T}\overset{\mathcal{D}}{\rightarrow}\mathcal{N}(0,1),
\]
while under the alternative we have $\widehat{\tau}_{j,n,T}\rightarrow\infty$
with probability approaching one. To ensure that $\widehat{\tau
}_{j,n,T}$ has an asymptotic distribution, we may need a two-step estimator
for the additive functions $h_{jt}(\cdot)$ as given in Horowitz and Mammen
(2011) or Ma and Yang (2011). This interesting and challenging technical
problem deserves further investigation, and it can be a good future research
topic.

\section{Acknowledge}
Ma's research was partially supported by NSF grants DMS 1306972 and DMS 1712558. Gao's research was supported by the Australian
Research Council Discovery Grants Program for its support under Grant numbers:
DP150101012 \& DP170104421.

\renewcommand{\theequation}{A.\arabic{equation}}

\setcounter{equation}{0}

\section{Appendix}

\label{SEC:appendix}

We first introduce some notations which will be used throughout the
Appendix. Let $\lambda_{\max}\left(  \mathbf{A}\right)  $ and $\lambda_{\min
}\left(  \mathbf{A}\right)  $ denote the largest and smallest eigenvalues of a
symmetric matrix $\mathbf{A}$, respectively. For an $m\times n$ real matrix $\mathbf{A}$, we denote $\left\Vert \mathbf{A}\right\Vert _{\infty
}=\max_{1\leq i\leq m}\sum_{j=1}^{n}\left\vert A_{ij}\right\vert $. For any
vector $\mathbf{a=(}a_{1},\ldots,a_{n})^{\intercal}\in\mathbb{R}^{n}$, denote
$||\mathbf{a||}_{\infty}=\max_{1\leq i\leq n}|a_{i}|$. We first study the
asymptotic properties of the initial estimators $\widehat{g}_{j}^{[0]}(x_{j})$
of $g_{j}^{0}(x_{j})$. The following proposition gives the convergence rate of
$\widehat{g}_{j}^{[0]}(x_{j})$ that will be used in
the proofs of Theorems \ref{THM:fhat} and \ref{THM:ghat}.

\begin{proposition}
{\normalsize \label{THM:ghat0} Let Conditions (C1)-(C4) hold. If, in addition,
$K_{N}^{4}N^{-1}=o(1)$, $K_{N}^{-r+2}(\log T)=o(1)$ and $K_{N}^{-1}(\log
NT)(\log N)^{4}=o(1)$, then for every $1\leq j\leq J$,
\begin{align*}
\sup\nolimits_{x_{j}\in\lbrack a,b]}|\widehat{g}_{j}^{[0]}(x_{j})-g_{j}%
^{0}(x_{j})|  &  =O_{p}(K_{N}/\sqrt{NT}+K_{N}^{2}N^{-3/4}\sqrt{\log NT}%
+K_{N}^{-r})+o_{p}(N^{-1/2}),\\
\left[  \int\{\widehat{g}_{j}^{[0]}(x_{j})-g_{j}^{0}(x_{j})\}^{2}%
dx_{j}\right]  ^{1/2}  &  =O_{p}(\sqrt{K_{N}/(NT)}+K_{N}^{3/2}N^{-3/4}%
\sqrt{\log NT}+K_{N}^{-r})+o_{p}(N^{-1/2}).
\end{align*}
}
\end{proposition}

\subsection{{\protect\normalsize Proof of Proposition \ref{THM:ghat0}}}
\label{SEC:ghat0}

According to the result on page 149
of de Boor (2001), for $h_{jt}^{0}$ satisfying the smoothness condition given
in (C2), there exists $\boldsymbol{\theta}_{jt}^{0}\in\mathbb{R}^{K_{n}}$ such
that $h_{jt}^{0}(x_{j})=\widetilde{h}_{jt}^{0}(x_{j})+b_{jt}(x_{j})$
\begin{equation}
\widetilde{h}_{jt}^{0}(x_{j})=B_{j}(x_{j})^{\intercal}\boldsymbol{\theta}%
_{jt}^{0}\text{ and }\sup_{j,t}\sup_{x_{j}\in\lbrack a,b]}|b_{jt}%
(x_{j})|=O(K_{N}^{-r}). \label{EQ:Rjt}%
\end{equation}
Denote $\widetilde{h}_{t}^{0}(x)=\{\widetilde{h}_{jt}^{0}(x_{j}),1\leq j\leq
J\}^{\intercal}$, and
\[
b_{t}(x)=\sum\nolimits_{j=1}^{J}h_{jt}^{0}(x_{j})-B(x)^{\intercal
}\boldsymbol{\theta}_{t}^{0},
\]
where $B(x)=\{B_{1}(x_{1})^{\intercal},\ldots,B_{J}(x_{J})^{\intercal
}\}^{\intercal}$ and $\boldsymbol{\theta}_{t}^{0}=(\boldsymbol{\theta}%
_{1t}^{0\intercal},\ldots,\boldsymbol{\theta}_{Jt}^{0\intercal})^{\intercal}$.
Then by (\ref{EQ:Rjt}), we have
\[
\sup\nolimits_{x\in\lbrack a,b]^{J}}|b_{t}(x)|=O(K_{N}^{-r}).
\]
Then $\mathbb{B}(x)(\widetilde{h}_{ut},\widetilde{\boldsymbol{\theta}}%
_{t}^{\intercal})^{\intercal}=(\widetilde{h}_{ut},\widetilde{h}_{t}%
(x)^{\intercal})^{\intercal}$ and $\mathbb{B}(x)(h_{ut}^{0},\boldsymbol{\theta
}_{t}^{0\intercal})^{\intercal}=(h_{ut}^{0},\widetilde{h}_{t}^{0}%
(x)^{\intercal})^{\intercal}$, where
\begin{equation}
\mathbb{B}(x)=\left[  \text{diag\{}1,B_{1}(x_{1})^{\intercal},\ldots
,B_{J}(x_{J})^{\intercal}\}\right]  _{(1+J)\times(1+JK_{N})}, \label{EQ:Bx}%
\end{equation}
$\widetilde{h}_{t}(x)=\{\widetilde{h}_{jt}(x_{j}),1\leq j\leq
J\}^{\intercal}$, and $\widetilde{h}_{jt}(\cdot)$ are the estimators given in Section \ref{initial}.
We first give the Bernstein inequality for a $\phi$-mixing sequence, which is
used through our proof.

\begin{lemma}
{\normalsize \label{LEM:Bernstein} Let $\left\{  \xi_{i}\right\}  $ be a sequence
of centered real-valued random variables. Let $S_{n}=\sum_{i=1}^{n}\xi_{i}$.
Suppose the sequence has the $\phi$-mixing coefficient satisfying $\phi
(k)\leq\exp(-2ck)$ for some $c>0$ and $\sup_{i\geq1}|\xi_{i}|\leq M$. Then
there is a positive constant $C_{1}$ depending only on $c$ such that for all
$n\geq2$%
\[
P(\left\vert S_{n}\right\vert \geq\varepsilon)\leq\exp(-\frac{C_{1}%
\varepsilon^{2}}{v^{2}n+M^{2}+\varepsilon M(\log n)^{2}}),
\]
where $v^{2}=\sup_{i>0}($\textrm{var}$\left(  \xi_{i}\right)  +2\sum_{j>i}%
|$\textrm{cov}$(\xi_{i},\xi_{j})|)$. }
\end{lemma}

\begin{proof}
{\normalsize The result of Lemma \ref{LEM:Bernstein} is given in Theorem 2 on page
275 of Merlev\'{e}de, Peligrad and Rio (2009) when the sequence $\left\{
\xi_{i}\right\}  $ has the $\alpha$-mixing coefficient satisfying
$\alpha(k)\leq\exp(-2ck)$ for some $c>0$. Thus, this result also holds for the
sequence having the $\phi$-mixing coefficient\ satisfying $\phi(k)\leq
\exp(-2ck)$, since $\alpha(k)\leq\phi(k)\leq\exp(-2ck)$. }
\end{proof}



Denote $B(X_{i})=\{B_{1}%
(X_{1i})^{\intercal},\ldots,B_{J}(X_{Ji})^{\intercal}\}^{\intercal}$ and
$Z_{i}=[\{1,B(X_{i})^{\intercal}\}^{\intercal}]_{(1+JK_{N})\times1}$.
Denote $\boldsymbol{\vartheta}_{t}=(h_{ut},\boldsymbol{\theta}_{t}^{\intercal
})^{\intercal}$ and $\boldsymbol{\vartheta}_{t}^{0}=(h_{ut}^{0}%
,\boldsymbol{\theta}_{t}^{0\intercal})^{\intercal}$. Define
\begin{align*}
G_{tN,i}(\boldsymbol{\vartheta}_{t})  &  =[\tau-I\{\varepsilon_{it}\leq
Z_{i}^{\intercal}(\boldsymbol{\vartheta}_{t}-\boldsymbol{\vartheta}_{t}%
^{0})-b_{t}(X_{i})\}]Z_{i},\\
G_{tN,i}^{\ast}(\boldsymbol{\vartheta}_{t})  &  =[\tau-F_{i}[\{Z_{i}%
^{\intercal}(\boldsymbol{\vartheta}_{t}-\boldsymbol{\vartheta}_{t}^{0}%
)-b_{t}(X_{i})\}|X_{i},f_{t}]]Z_{i},
\end{align*}
where $F_{i}(\varepsilon|X_{i},f_{t})=P(\varepsilon_{it}\leq\varepsilon
|X_{i},f_{t})$, and $\widetilde{G}_{tN,i}(\boldsymbol{\vartheta}_{t}%
)=G_{tN,i}(\boldsymbol{\vartheta}_{t})-G_{tN,i}^{\ast}(\boldsymbol{\vartheta
}_{t})$. \ Let $d(N)=(1+JK_{N})$.

Let $\Psi_{Nt}=N^{-1}\sum\nolimits_{i=1}^{N}p_{i}\left(  0\left\vert
X_{i},f_{t}\right.  \right)  Z_{i}Z_{i}^{\intercal}$. By the same reasoning as
the proofs for (ii) of Lemma A.7 in Ma and Yang (2011), we have with
probability approaching 1, as $N\rightarrow\infty$, there exist constants
$0<C_{1}\leq C_{2}<\infty$ such that%
\begin{equation}
C_{1}\leq\lambda_{\min}(\Psi_{Nt})\leq\lambda_{\max}(\Psi_{Nt})\leq C_{2},
\label{EQ:ZZ}%
\end{equation}
uniformly in $t=1,...,T$.

Next lemma presents the Bahadur representation for
$\widetilde{\boldsymbol{\vartheta}}_{t}=(\widetilde{h}_{ut}%
,\widetilde{\boldsymbol{\theta}}_{t}^{\intercal})^{\intercal}$ using the results in Lemmas
\ref{LEM:Gtildatheta-theta0}-\ref{LEM:Gstar} given in the Supplemental Materials.

\begin{lemma}
\label{LEM:bahardur}Under Conditions (C1)-(C3), and $K_{N}^{3}%
N^{-1}=o(1)$, $K_{N}^{2}N^{-1}(\log NT)^{2}(\log N)^{8}=o(1)$ and
$K_{N}^{-r+1}(\log T)=o(1)$,%
\begin{equation}
\widetilde{\boldsymbol{\vartheta}}_{t}-\boldsymbol{\vartheta}_{t}^{0}%
=D_{Nt,1}+D_{Nt,2}+R_{Nt}, \label{varthetatilda}%
\end{equation}
where
\begin{equation}
D_{Nt,1}=\Psi_{Nt}^{-1}\left[
N^{-1}\sum\nolimits_{i=1}^{N}Z_{i}(\tau-I(\varepsilon_{it}<0))\right]  ,
\label{EQ:DNt1}%
\end{equation}%
\[
D_{Nt,2}=\Psi_{Nt}^{-1}\left[  N^{-1}\sum\nolimits_{i=1}^{N}Z_{i}%
\{p_{i}\left(  0\left\vert X_{i},f_{t}\right.  \right)  \sum\nolimits_{j=1}%
^{J}b_{jt}(X_{ji})\}\right]  ,
\]
uniformly in $t$, and the remaining term $R_{Nt}$ satisfies
\begin{align*}
\sup_{1\leq t\leq T}||R_{Nt}||  &  =O_{p}(K_{N}^{3/2}N^{-1}+K_{N}%
^{3/2}N^{-3/4}\sqrt{\log NT}+K_{N}^{1/2-2r}+N^{-1/2}K_{N}^{-r/2+1/2}\sqrt{\log
K_{N}T})\\
&  =O_{p}(K_{N}^{3/2}N^{-3/4}\sqrt{\log NT}+K_{N}^{1/2-2r})+o_{p}(N^{-1/2}).
\end{align*}
\end{lemma}

\begin{proof}
By Lemma \ref{LEM:Gstar} in the Supplemental Materials, we have
\[
\widetilde{\boldsymbol{\vartheta}}_{t}-\boldsymbol{\vartheta}_{t}^{0}%
=N^{-1}\Psi_{Nt}^{-1}\sum\nolimits_{i=1}^{N}p_{i}\left(  0\left\vert
X_{i},f_{t}\right.  \right)  Z_{i}b_{t}(X_{i})-\Psi_{Nt}^{-1}G_{tN,i}^{\ast
}(\widetilde{\boldsymbol{\vartheta}}_{t})+R_{Nt}^{\ast}.
\]
Moreover,
\[
\Psi_{Nt}^{-1}G_{tN,i}^{\ast}(\widetilde{\boldsymbol{\vartheta}}_{t}%
)=\Psi_{Nt}^{-1}G_{tN,i}(\widetilde{\boldsymbol{\vartheta}}_{t})-\Psi
_{Nt}^{-1}\widetilde{G}_{tN,i}(\boldsymbol{\vartheta}_{t}^{0})-\Psi_{Nt}%
^{-1}[\widetilde{G}_{tN,i}(\widetilde{\boldsymbol{\vartheta}}_{t}%
)-\widetilde{G}_{tN,i}(\boldsymbol{\vartheta}_{t}^{0})].
\]
Thus,%
\begin{equation}
\widetilde{\boldsymbol{\vartheta}}_{t}-\boldsymbol{\vartheta}_{t}^{0}%
=\Psi_{Nt}^{-1}N^{-1}\sum\nolimits_{i=1}^{N}\widetilde{G}_{tN,i}%
(\boldsymbol{\vartheta}_{t}^{0})+\Psi_{Nt}^{-1}N^{-1}\sum\nolimits_{i=1}%
^{N}p_{i}\left(  0\left\vert X_{i},f_{t}\right.  \right)  Z_{i}b_{t}%
(X_{i})+R_{Nt}^{\ast\ast},\label{EQ:vartheattilda}%
\end{equation}
where
\begin{equation}
R_{Nt}^{\ast\ast}=-\Psi_{Nt}^{-1}N^{-1}\sum\nolimits_{i=1}^{N}G_{tN,i}%
(\widetilde{\boldsymbol{\vartheta}}_{t})+\Psi_{Nt}^{-1}N^{-1}\sum
\nolimits_{i=1}^{N}[\widetilde{G}_{tN,i}(\widetilde{\boldsymbol{\vartheta}%
}_{t})-\widetilde{G}_{tN,i}(\boldsymbol{\vartheta}_{t}^{0})]+R_{Nt}^{\ast
}.\label{EQ:Rstarstar}%
\end{equation}
By\ Lemmas \ref{LEM:Gtildatheta-theta0} and \ref{LEM:Gtildatilda} in the Supplemental Materials and
(\ref{EQ:ZZ}), we have%
\begin{align*}
\sup_{1\leq t\leq T}||R_{Nt}^{\ast\ast}|| &  \leq\sup_{1\leq t\leq T}%
||\Psi_{Nt}^{-1}||\sup_{1\leq t\leq T}||N^{-1}\sum\nolimits_{i=1}^{N}%
G_{tN,i}(\widetilde{\boldsymbol{\vartheta}}_{t})||\\
&  +\sup_{1\leq t\leq T}||\Psi_{Nt}^{-1}||\sup_{1\leq t\leq T}||N^{-1}%
\sum\nolimits_{i=1}^{N}[\widetilde{G}_{tN,i}(\widetilde{\boldsymbol{\vartheta
}}_{t})-\widetilde{G}_{tN,i}(\boldsymbol{\vartheta}_{t}^{0})]||+\sup_{1\leq
t\leq T}||R_{Nt}^{\ast}||\\
&  =O_{p}(K_{N}^{3/2}N^{-1}+(K_{N}^{2}N)^{-3/4}\sqrt{\log NT}+K_{N}^{1/2-2r}).
\end{align*}
Define $\overline{G}_{tN,i\ell}(\boldsymbol{\vartheta}_{t}^{0})=\{\tau
-I(\varepsilon_{it}\leq0)\}Z_{i,\ell}$ and $\overline{G}_{tN,i}%
(\boldsymbol{\vartheta}_{t}^{0})=\{\overline{G}_{tN,i\ell}%
(\boldsymbol{\vartheta}_{t}^{0}),1\leq\ell\leq d(N)\}$. Then $E\{\widetilde{G}%
_{tN,i\ell}(\boldsymbol{\vartheta}_{t}^{0})-\overline{G}_{tN,i\ell
}(\boldsymbol{\vartheta}_{t}^{0})\}=0$. Moreover,
\[
E\{\widetilde{G}_{tN,i\ell}(\boldsymbol{\vartheta}_{t}^{0})-\overline
{G}_{tN,i\ell}(\boldsymbol{\vartheta}_{t}^{0})\}^{2}\leq E\left[
I\{\varepsilon_{it}\leq-b_{t}(X_{i})\}-I\{\varepsilon_{it}\leq0\}Z_{i,\ell
}\right]  ^{2}\leq CK_{N}^{-r}%
\]
for some constant $0<C<\infty$, and by Condition (C1), we have
\begin{align*}
&  E\{\widetilde{G}_{tN,i\ell}(\boldsymbol{\vartheta}_{t}^{0})-\overline
{G}_{tN,i\ell}(\boldsymbol{\vartheta}_{t}^{0})\}\{\widetilde{G}_{tN,i^{\prime
}\ell}(\boldsymbol{\vartheta}_{t}^{0})-\overline{G}_{tN,i^{\prime}\ell
}(\boldsymbol{\vartheta}_{t}^{0})\}\\
&  \leq2\times4^{2}\{\phi(|i^{\prime}-i|)\}^{1/2}[E\{\widetilde{G}_{tN,i\ell
}(\boldsymbol{\vartheta}_{t}^{0})-\overline{G}_{tN,i\ell}%
(\boldsymbol{\vartheta}_{t}^{0})\}^{2}E\{\widetilde{G}_{tN,i^{\prime}\ell
}(\boldsymbol{\vartheta}_{t}^{0})-\overline{G}_{tN,i^{\prime}\ell
}(\boldsymbol{\vartheta}_{t}^{0})\}^{2}]^{1/2}\\
&  \leq C^{\prime}K_{1}e^{-\lambda_{1}|i^{\prime}-i|/2}K_{N}^{-r}.
\end{align*}
Hence, by the above results, we have
\begin{align*}
&  E[N^{-1}\sum\nolimits_{i=1}^{N}\{\widetilde{G}_{tN,i\ell}%
(\boldsymbol{\vartheta}_{t}^{0})-\overline{G}_{tN,i\ell}(\boldsymbol{\vartheta
}_{t}^{0})\}]^{2}\\
&  \leq N^{-1}CK_{N}^{-r}+N^{-2}\sum\nolimits_{i\neq i^{\prime}}C^{\prime
}K_{1}e^{-\lambda_{1}|i^{\prime}-i|}K_{N}^{-r}\\
&  \leq CN^{-1}K_{N}^{-r}+C^{\prime}K_{1}N^{-2}N(1-e^{-\lambda_{1}/2}%
)^{-1}K_{N}^{-r}\leq C^{\prime\prime}N^{-1}K_{N}^{-r},
\end{align*}
for some constant $0<C^{\prime\prime}<\infty$. Thus
\begin{align*}
E||N^{-1}\sum\nolimits_{i=1}^{N}\{\widetilde{G}_{tN,i}(\boldsymbol{\vartheta
}_{t}^{0})-\overline{G}_{tN,i}(\boldsymbol{\vartheta}_{t}^{0})\}||^{2} &
=\sum\nolimits_{\ell=1}^{d(N)}E[N^{-1}\sum\nolimits_{i=1}^{N}\{\widetilde{G}%
_{tN,i\ell}(\boldsymbol{\vartheta}_{t}^{0})-\overline{G}_{tN,i\ell
}(\boldsymbol{\vartheta}_{t}^{0})\}]^{2}\\
&  \leq C^{\prime\prime}(1+JK_{N})N^{-1}K_{N}^{-r}.
\end{align*}
Therefore, by the Bernstein's inequality in Lemma \ref{LEM:Bernstein} and the union bound of probability, following the same procedure as the proof for Lemma \ref{LEM:Gtildatheta-theta0} given in the Supplemental Materials,
we have
\begin{equation}
\sup_{1\leq t\leq T}||N^{-1}\sum\nolimits_{i=1}^{N}\{\widetilde{G}%
_{tN,i}(\boldsymbol{\vartheta}_{t}^{0})-\overline{G}_{tN,i}%
(\boldsymbol{\vartheta}_{t}^{0})\}||=O_{p}(N^{-1/2}K_{N}^{-r/2+1/2}\sqrt{\log
K_{N}T}).\label{EQ;Gtilda-Gbar}%
\end{equation}
Therefore, by (\ref{EQ:vartheattilda}), (\ref{EQ:Rstarstar}) and
(\ref{EQ;Gtilda-Gbar}), we have $\widetilde{\boldsymbol{\vartheta}}%
_{t}-\boldsymbol{\vartheta}_{t}^{0}=D_{Nt,1}+D_{Nt,2}+R_{Nt}$, where
\[
\sup_{1\leq t\leq T}||R_{Nt}||=O_{p}(K_{N}^{3/2}N^{-1}+(K_{N}^{2}%
N)^{-3/4}\sqrt{\log NT}+K_{N}^{1/2-2r}+N^{-1/2}K_{N}^{-r/2+1/2}\sqrt{\log
K_{N}T}).
\]
\end{proof}

\begin{proof}
[Proof of Proposition \ref{THM:ghat0}]{\normalsize Let $1_{l}$ be the
$(J+1)\times1$ vector with the $l^{\text{th}}$ element as \textquotedblleft%
$1$\textquotedblright\ and other elements as \textquotedblleft$0$%
\textquotedblright. By (\ref{varthetatilda}) in Lemma \ref{LEM:bahardur}, we
have
\[
\widetilde{h}_{jt}(x_{j})-\widetilde{h}_{jt}^{0}(x_{j})=1_{j+1}^{\intercal
}\mathbb{B}(x)(D_{Nt,1}+D_{Nt,2})+1_{j+1}^{\intercal}\mathbb{B}(x)R_{Nt},
\]%
\begin{align*}
\sup_{1\leq t\leq T}\{N^{-1}\sum\nolimits_{i=1}^{N}(1_{j+1}^{\intercal
}\mathbb{B}(X_{i})R_{Nt})^{2}\}^{1/2} &  \leq\sup_{1\leq t\leq T}%
||R_{Nt}||[\lambda_{\max}\{N^{-1}\sum\nolimits_{i=1}^{N}B_{j}(X_{ji}%
)B_{j}(X_{ji})^{\intercal}\}]^{1/2}\\
&  =O_{p}(K_{N}^{3/2}N^{-3/4}\sqrt{\log NT}+K_{N}^{1/2-2r})+o_{p}(N^{-1/2}),
\end{align*}
and%
\begin{align*}
&  \sup\nolimits_{1\leq t\leq T}\sup\nolimits_{x\in\lbrack a,b]^{J}}%
|1_{j+1}^{\intercal}\mathbb{B}(x)R_{Nt}|\\
&  \leq\sup\nolimits_{x\in\lbrack a,b]^{J}}||\mathbb{B}(x)^{\intercal}%
1_{j+1}||\sup\nolimits_{1\leq t\leq T}||R_{Nt}||\\
&  =O(K_{N}^{1/2})O_{p}(K_{N}^{3/2}N^{-1}+K_{N}^{3/2}N^{-3/4}\sqrt{\log
NT}+K_{N}^{1/2-2r}+N^{-1/2}K_{N}^{-r/2+1/2}\sqrt{\log K_{N}T})\\
&  =O_{p}(K_{N}^{2}N^{-3/4}\sqrt{\log NT}+K_{N}^{1-2r})+o_{p}(N^{-1/2}),
\end{align*}
by the assumption that $K_{N}^{4}N^{-1}=o(1)$, $K_{N}^{-r+2}(\log T)=o(1)$ and
$r>2$. Since $h_{jt}^{0}(x_{j})=\widetilde{h}_{jt}^{0}(x_{j})+b_{jt}(x_{j})$,
then we have
\[
\widetilde{h}_{jt}(x_{j})-h_{jt}^{0}(x_{j})=1_{j+1}^{\intercal}\mathbb{B}%
(x)(D_{Nt,1}+D_{Nt,2})-b_{jt}(x_{j})+1_{j+1}^{\intercal}\mathbb{B}(x)R_{Nt}.
\]
Also by (\ref{EQ:Rjt}), we have $\sup_{1\leq t\leq T}\sup_{x\in\lbrack
a,b]^{J}}\left\vert 1_{j+1}^{\intercal}\mathbb{B}(x)D_{Nt,2}\right\vert
=O_{p}(K_{N}^{-r})$. Then $\widetilde{h}_{jt}(x_{j})-h_{jt}^{0}(x_{j})$ can be
written as
\begin{equation}
\widetilde{h}_{jt}(x_{j})-h_{jt}^{0}(x_{j})=1_{j+1}^{\intercal}\mathbb{B}%
(x)D_{Nt,1}+\eta_{N,jt}(x_{j}),\label{EQ:hjt}%
\end{equation}
where the remaining term $\eta_{N,jt}(x_{j})$ satisfies
\begin{equation}
\sup_{1\leq t\leq T}[N^{-1}\sum\nolimits_{i=1}^{N}\{\eta_{N,jt}(X_{ji}%
)\}^{2}]^{1/2}=O_{p}(K_{N}^{-r})+O_{p}(K_{N}^{3/2}N^{-3/4}\sqrt{\log
NT})+o_{p}(N^{-1/2}),\label{EQ:etaNjt}%
\end{equation}%
\begin{align}
\sup_{1\leq t\leq T}\{\int\eta_{N,jt}(x_{j})^{2}dx_{j}\}^{1/2} &  =O_{p}%
(K_{N}^{-r})+O_{p}(K_{N}^{3/2}N^{-3/4}\sqrt{\log NT})+o_{p}(N^{-1/2}%
),\nonumber\\
\sup_{1\leq t\leq T}\sup\nolimits_{x_{j}\in\lbrack a,b]}|\eta_{N,jt}(x_{j})|
&  =O_{p}(K_{N}^{-r})+O_{p}(K_{N}^{2}N^{-3/4}\sqrt{\log NT})+o_{p}%
(N^{-1/2}).\label{EQ:supeta}%
\end{align}
Moreover, by Berntein's inequality and following the same procedure as the
proof for Lemma \ref{LEM:Gtildatheta-theta0}, we have $\sup_{1\leq t\leq
T}||D_{Nt,1}||=O_{p}(\sqrt{K_{N}/N}\sqrt{\log K_{N}T})$. Hence,
\begin{align}
\sup_{1\leq t\leq T}\sup_{x\in\lbrack a,b]^{J}}|1_{j+1}^{\intercal}%
\mathbb{B}(x)D_{Nt,1}| &  =O_{p}(\sqrt{\log K_{N}T}K_{N}/\sqrt{N}),\nonumber\\
\sup_{1\leq t\leq T}\{N^{-1}\sum\nolimits_{i=1}^{N}(1_{j+1}^{\intercal
}\mathbb{B}(X_{i})D_{Nt,1})^{2}\}^{1/2} &  =O_{p}(\sqrt{\log K_{N}T}%
\sqrt{K_{N}/N}).\label{EQ:DNt11}%
\end{align}
Therefore, by (\ref{EQ:hjt}), (\ref{EQ:etaNjt}), (\ref{EQ:supeta}) and
(\ref{EQ:DNt11}), we have
\[
\sup_{1\leq t\leq T}N^{-1}\sum\nolimits_{i=1}^{N}\{\widetilde{h}_{jt}%
(X_{ji})-h_{jt}^{0}(X_{ji})\}^{2}=O_{p}((\log K_{N}T)K_{N}/N+N^{-2r}),
\]%
\begin{equation}
\sup_{1\leq t\leq T}\sup\nolimits_{x_{j}\in\lbrack a,b]}|\widetilde{h}%
_{jt}(x_{j})-h_{jt}^{0}(x_{j})|=O_{p}(\sqrt{\log K_{N}T}K_{N}N^{-1/2}%
+K_{N}^{-r}).\label{EQ:htilda-ho}%
\end{equation}
Moreover, by Conditions (C3) and (C4), we have with probability approaching
$1$, as $N\rightarrow\infty$, }%
\begin{equation}
{\normalsize c_{h}\leq N^{-1}\sum\nolimits_{i=1}^{N}(T^{-1}\sum\nolimits_{t=1}%
^{T}h_{jt}^{0}(X_{ji}))^{2}\leq C_{h}\ ,\ c_{h}\leq N^{-1}\sum\nolimits_{i=1}%
^{N}(T^{-1}\sum\nolimits_{t=1}^{T}\widetilde{h}_{jt}(X_{ji}))^{2}\leq C_{h}%
.}\label{htilda}%
\end{equation}
{\normalsize Hence, this result together with (\ref{EQ:hjt}) leads to that
with probability approaching $1$, as $N\rightarrow\infty$,}%
\begin{align}
&  \left\vert 1/\sqrt{N^{-1}\sum\nolimits_{i=1}^{N}(T^{-1}\sum\nolimits_{t=1}%
^{T}\widetilde{h}_{jt}(X_{ji}))^{2}}-1/\sqrt{N^{-1}\sum\nolimits_{i=1}%
^{N}(T^{-1}\sum\nolimits_{t=1}^{T}h_{jt}^{0}(X_{ji}))^{2}}\right\vert
\nonumber\\
&  =\left\vert M_{NT}N^{-1}\sum\nolimits_{i=1}^{N}T^{-1}\sum\nolimits_{t=1}%
^{T}\{\widetilde{h}_{jt}(X_{ji})-h_{jt}^{0}(X_{ji})\}T^{-1}\sum\nolimits_{t=1}%
^{T}\{\widetilde{h}_{jt}(X_{ji})+h_{jt}^{0}(X_{ji})\}\right\vert \nonumber\\
&  =|2M_{NT}N^{-1}\sum\nolimits_{i=1}^{N}T^{-1}\sum\nolimits_{t=1}%
^{T}\{\widetilde{h}_{jt}(X_{ji})-h_{jt}^{0}(X_{ji})\}\{T^{-1}\sum
\nolimits_{t=1}^{T}h_{jt}^{0}(X_{ji})\}\nonumber\\
&  +M_{NT}N^{-1}\sum\nolimits_{i=1}^{N}T^{-1}\sum\nolimits_{t=1}%
^{T}\{\widetilde{h}_{jt}(X_{ji})-h_{jt}^{0}(X_{ji})\}^{2}|\nonumber\\
&  \leq\left\vert 2M_{NT}N^{-1}\sum\nolimits_{i=1}^{N}T^{-1}\sum
\nolimits_{t=1}^{T}[1_{j+1}^{\intercal}\mathbb{B}(X_{i})D_{Nt,1}\{T^{-1}%
\sum\nolimits_{t=1}^{T}h_{jt}^{0}(X_{ji})\}+\varrho_{it}]\right\vert
\nonumber\\
&  +2M_{NT}N^{-1}\sum\nolimits_{i=1}^{N}T^{-1}\sum\nolimits_{t=1}^{T}%
\{1_{j+1}^{\intercal}\mathbb{B}(X_{i})D_{Nt,1}\}^{2}+\eta_{N,jt}%
(X_{ji})\nonumber\\
&  +2M_{NT}N^{-1}\sum\nolimits_{i=1}^{N}T^{-1}\sum\nolimits_{t=1}^{T}%
\{\eta_{N,jt}(X_{ji})\}^{2}\label{EQ:htilda-h0}%
\end{align}
for $M_{NT}$ satisfying $M_{NT}\in(c^{\prime},C^{\prime})$ {\normalsize for
some constants $0<c^{\prime}<C^{\prime}<\infty$, where }$\varrho_{it}%
${\normalsize $=\eta_{N,jt}(X_{ji})\{T^{-1}\sum\nolimits_{t=1}^{T}h_{jt}%
^{0}(X_{ji})\}$. Moreover by (\ref{EQ:etaNjt}), there exists a constant
}$C^{\ast}\in(0,\infty)$ such that
\begin{align}
|N^{-1}\sum\nolimits_{i=1}^{N}T^{-1}\sum\nolimits_{t=1}^{T}\varrho_{it}| &
\leq C^{\ast}\sup\nolimits_{1\leq t\leq T}N^{-1}\sum\nolimits_{i=1}^{N}%
|\eta_{N,jt}(X_{ji})|\nonumber\\
&  \leq C^{\ast}\sup\nolimits_{1\leq t\leq T}[N^{-1}\sum\nolimits_{i=1}%
^{N}\{\eta_{N,jt}(X_{ji})\}^{2}]^{1/2}\nonumber\\
&  =O_{p}(K_{N}^{-r})+O_{p}(K_{N}^{3/2}N^{-3/4}\sqrt{\log NT})+o_{p}%
(N^{-1/2}),\label{EQ:rhotN1}%
\end{align}
and%
\begin{equation}
N^{-1}\sum\nolimits_{i=1}^{N}T^{-1}\sum\nolimits_{t=1}^{T}\{\eta_{N,jt}%
(X_{ji})\}^{2}=O_{p}(K_{N}^{-2r})+O_{p}(K_{N}^{3}N^{-3/2}\log(NT))+o_{p}%
(N^{-1}).\label{EQ:rhotN2}%
\end{equation}
{\normalsize Define $\psi_{it}=\{\psi_{it,\ell}\}_{\ell=1}^{d(N)}=\Psi
_{Nt}^{-1}$}$Z_{i}(\tau-I(\varepsilon_{it}<0))${\normalsize . Then
$E(\psi_{it,\ell})=0$. Moreover, $E||\psi_{it}||^{2}\leq c_{1}$}$K_{N}%
${\normalsize for some constant $0<c_{1}<\infty$, and by Condition (C1), we
have
\begin{align*}
|E(\psi_{it}^{\intercal}\psi_{js})| &  \leq2\{\phi(\sqrt{|i-j|^{2}+|t-s|^{2}%
})\}^{1/2}\sum\nolimits_{\ell=1}^{d(N)}\{E(\psi_{it,\ell})^{2}E(\psi_{js,\ell
})^{2}\}^{1/2}\\
&  \leq\{\phi(\sqrt{|i-j|^{2}+|t-s|^{2}})\}^{1/2}(E||\psi_{it}||^{2}%
+E||\psi_{js}||^{2})\\
&  \leq2c_{1}K_{N}\{\phi(\sqrt{|i-j|^{2}+|t-s|^{2}})\}^{1/2}.
\end{align*}
Hence by Condition (C1), we have
\begin{align*}
&  E||(NT)^{-1}\sum\nolimits_{t=1}^{T}\sum\nolimits_{i=1}^{N}\psi_{it}||^{2}\\
&  =(NT)^{-2}\sum\nolimits_{t,t^{\prime}}\sum\nolimits_{i,i^{\prime}}%
E(\psi_{it}^{\intercal}\psi_{i^{\prime}t^{\prime}})\leq2c_{1}K_{N}%
(NT)^{-2}\sum\nolimits_{t,t^{\prime}}\sum\nolimits_{i,i^{\prime}}\{\phi
(\sqrt{|i-j|^{2}+|t-s|^{2}})\}^{1/2}\\
&  \leq2c_{1}K_{1}K_{N}(NT)^{-2}\sum\nolimits_{t,t^{\prime}}\sum
\nolimits_{i,i^{\prime}}e^{-\lambda_{1}\sqrt{|i-i^{\prime}|^{2}+|t-t^{\prime
}|^{2}}/2}\\
&  \leq2c_{1}(NT)^{-2}K_{1}K_{N}\sum\nolimits_{t,t^{\prime}}\sum
\nolimits_{i,i^{\prime}}e^{-(\lambda_{1}/2)(|i-i^{\prime}|+|t-t^{\prime}|)}\\
&  \leq2c_{1}K_{1}K_{N}(NT)^{-2}(NT)(\sum\nolimits_{k=0}^{T}e^{-(\lambda
_{1}/2)k})(\sum\nolimits_{k=0}^{N}e^{-(\lambda_{1}/2)k})\\
&  \leq2c_{1}K_{1}K_{N}(NT)^{-2}(NT)\{1-e^{-(\lambda_{1}/2)}\}^{-2}%
=2c_{1}K_{1}K_{N}\{1-e^{-(\lambda_{1}/2)}\}^{-2}(NT)^{-1}=O\{K_{N}(NT)^{-1}\}.
\end{align*}
Thus, by Markov's inequality,}%
\begin{equation}
||(NT)^{-1}\sum\nolimits_{t=1}^{T}\sum\nolimits_{i=1}^{N}\psi_{it}%
||=O_{p}[\{K_{N}(NT)^{-1}\}^{1/2}].\label{EQ:Ez}%
\end{equation}
{\normalsize Moreover, by the definition of $D_{Nt,1}$ given in (\ref{EQ:DNt1}%
), we have $D_{Nt,1}=$}$N^{-1}\sum\nolimits_{i=1}^{N}${\normalsize $\psi_{it}%
$. Therefore, }%
\begin{gather*}
|N^{-1}\sum\nolimits_{i=1}^{N}T^{-1}\sum\nolimits_{t=1}^{T}1_{j+1}^{\intercal
}\mathbb{B}(X_{i})D_{Nt,1}|=|N^{-1}\sum\nolimits_{i=1}^{N}1_{j+1}^{\intercal
}\mathbb{B}(X_{i})(NT)^{-1}\sum\nolimits_{t=1}^{T}\sum\nolimits_{i=1}^{N}%
\psi_{it}|\\
\leq||(NT)^{-1}\sum\nolimits_{t=1}^{T}\sum\nolimits_{i=1}^{N}\psi
_{it}||[\lambda_{\max}\{N^{-1}\sum\nolimits_{i=1}^{N}B_{j}(X_{ji})B_{j}%
(X_{ji})^{\intercal}\}]^{1/2}=O_{p}[\{K_{N}(NT)^{-1}\}^{1/2}].
\end{gather*}
By (\ref{EQ:DNt11}) and $\log(K_{N}T)K_{N}N^{-1/2}=o(1)$, we have
\[
N^{-1}\sum\nolimits_{i=1}^{N}T^{-1}\sum\nolimits_{t=1}^{T}\{1_{j+1}%
^{\intercal}\mathbb{B}(X_{i})D_{Nt,1}\}^{2}=\{O_{p}(\sqrt{\log K_{N}T}%
\sqrt{K_{N}/N})^{2}\}=o_{p}(N^{-1/2}).
\]
Therefore, the above results together with (\ref{EQ:htilda-h0}),
(\ref{EQ:rhotN1}) and (\ref{EQ:rhotN2}) lead to \
\begin{align*}
& \left\vert 1/\sqrt{N^{-1}\sum\nolimits_{i=1}^{N}(T^{-1}\sum\nolimits_{t=1}%
^{T}\widetilde{h}_{jt}(X_{ji}))^{2}}-1/\sqrt{N^{-1}\sum\nolimits_{i=1}%
^{N}(T^{-1}\sum\nolimits_{t=1}^{T}h_{jt}^{0}(X_{ji}))^{2}}\right\vert \\
& =O_{p}[\{K_{N}(NT)^{-1}\}^{1/2}]+O_{p}(K_{N}^{-r})+O_{p}(K_{N}^{3/2}%
N^{-3/4}\sqrt{\log NT})+o_{p}(N^{-1/2}).
\end{align*}
Denote $\varpi_{NT}=\sqrt{N^{-1}\sum\nolimits_{i=1}^{N}(T^{-1}\sum
\nolimits_{t=1}^{T}\widetilde{h}_{jt}(X_{ji}))^{2}}$ and $\varpi_{NT}%
^{0}=\sqrt{N^{-1}\sum\nolimits_{i=1}^{N}(T^{-1}\sum\nolimits_{t=1}^{T}%
h_{jt}^{0}(X_{ji}))^{2}}$. Then
\begin{align}
&  T^{-1}\sum\nolimits_{t=1}^{T}\{\widetilde{h}_{jt}(x_{j})/\varpi_{NT}%
-h_{jt}^{0}(x_{j})/\varpi_{NT}^{0}\}\nonumber\\
&  =T^{-1}\sum\nolimits_{t=1}^{T}\{\widetilde{h}_{jt}(x_{j})/\varpi
_{NT}-h_{jt}^{0}(x_{j})/\varpi_{NT}\}+T^{-1}\sum\nolimits_{t=1}^{T}%
\{h_{jt}^{0}(x_{j})/\varpi_{NT}-h_{jt}^{0}(x_{j})/\varpi_{NT}^{0}\}\nonumber\\
&  =T^{-1}\sum\nolimits_{t=1}^{T}\{\widetilde{h}_{jt}(x_{j})-h_{jt}^{0}%
(x_{j})\}/\varpi_{NT}+T^{-1}\sum\nolimits_{t=1}^{T}h_{jt}^{0}(x_{j}%
)\{1/\varpi_{NT}-1/\varpi_{NT}^{0}\}.\nonumber
\end{align}
By the above result and Condition (C3), we have
\begin{align*}
& \sup\nolimits_{x_{j}\in\lbrack a,b]}|T^{-1}\sum\nolimits_{t=1}^{T}h_{jt}%
^{0}(x_{j})\{1/\varpi_{NT}-1/\varpi_{NT}^{0}\}|\\
& =O_{p}[\{K_{N}(NT)^{-1}\}^{1/2}]+O_{p}(K_{N}^{-r})+O_{p}(K_{N}^{3/2}%
N^{-3/4}\sqrt{\log NT})+o_{p}(N^{-1/2}).
\end{align*}
Moreover, {\normalsize (\ref{EQ:hjt}) leads to that with probability
approaching $1$, as $N\rightarrow\infty$,}%
\begin{align*}
& T^{-1}\sum\nolimits_{t=1}^{T}\{\widetilde{h}_{jt}(x_{j})-h_{jt}^{0}%
(x_{j})\}/\varpi_{NT}\\
& =T^{-1}\sum\nolimits_{t=1}^{T}1_{j+1}^{\intercal}\mathbb{B}(x)D_{Nt,1}%
/\varpi_{NT}+T^{-1}\sum\nolimits_{t=1}^{T}\eta_{N,jt}(x_{j})/\varpi_{NT}%
=\Phi_{NTj,1}(x_{j})+\Phi_{NTj,2}(x_{j}).
\end{align*}
By (\ref{EQ:supeta}) and (\ref{htilda}),
\begin{align*}
\sup\nolimits_{x_{j}\in\lbrack a,b]}|\Phi_{NTj,2}(x_{j})| &  =O_{p}(K_{N}%
^{-r})+O_{p}(K_{N}^{2}N^{-3/4}\sqrt{\log NT})+o_{p}(N^{-1/2}),\\
\left\{  \int\Phi_{NTj,2}(x_{j})^{2}dx_{j}\right\}  ^{1/2} &  =O_{p}%
(K_{N}^{-r})+O_{p}(K_{N}^{3/2}N^{-3/4}\sqrt{\log NT})+o_{p}(N^{-1/2}).
\end{align*}
By (\ref{htilda}) and (\ref{EQ:Ez}),
\begin{align*}
\sup\nolimits_{x_{j}\in\lbrack a,b]}|\Phi_{NTj,1}(x_{j})| &  \leq
{\normalsize c_{h}^{-1}\{}||(NT)^{-1}\sum\nolimits_{t=1}^{T}\sum
\nolimits_{i=1}^{N}\psi_{it}||^{2}\sup\nolimits_{x_{j}\in\lbrack a,b]}%
||B_{j}(x_{j})||^{2}\}^{1/2}\\
&  =O_{p}\{K_{N}(NT)^{-1/2}\},\\
\left\{  \int\Phi_{NTj,1}(x_{j})^{2}dx_{j}\right\}  ^{1/2} &  \leq
{\normalsize c_{h}^{-1}C_{2}\{}||(NT)^{-1}\sum\nolimits_{t=1}^{T}%
\sum\nolimits_{i=1}^{N}\psi_{it}||^{2}\}^{1/2}=O_{p}\{K_{N}^{1/2}%
(NT)^{-1/2}\}.
\end{align*}

{\normalsize Hence, the results in Proposition \ref{THM:ghat0} follow from the
above results directly. }
\end{proof}

\subsection{{\protect\normalsize Proofs of Theorems \ref{THM:fhat} and
\ref{THM:ghat}}}\label{SEC:fhatghat}

We first present the following several lemmas that will be used in the
proofs of Theorems 1 and 2. We define the infeasible estimator $f_{t}^{\ast
}=\{f_{ut}^{\ast},(f_{jt}^{\ast},1\leq j\leq J)^{\intercal}\}^{\intercal}$ as
the minimizer of
\begin{equation}
\sum\nolimits_{i=1}^{N}\rho_{\tau}(y_{it}-f_{ut}-\sum\nolimits_{j=1}^{J}%
g_{j}^{0}(X_{ji})f_{jt}). \label{EQ:fastarobj}%
\end{equation}

\begin{lemma}
\label{LEM:fstar} Under Conditions (C1), (C2), (C4), (C5) and
(C6),\ we have as $N\rightarrow\infty$,%
\[
\sqrt{N}(\mathbf{\Sigma}_{Nt}^{0})^{-1/2}(f_{t}^{\ast}-f_{t}^{0}%
)\rightarrow\mathcal{N}(\mathbf{0},\mathbf{I}_{J+1}),
\]
where $\mathbf{\Sigma}_{Nt}^{0}$ is given in (\ref{EQ:SigN}).
\end{lemma}

\begin{proof}
By Bahadur representation for the $\phi$-mixing case (see Babu
(1989)), we have%
\begin{equation}
f_{t}^{\ast}-f_{t}^{0}=\Lambda_{Nt}^{-1}\{N^{-1}\sum\nolimits_{i=1}^{N}%
G_{i}^{0}(X_{i})(\tau-I(\varepsilon_{it}<0))\}+\upsilon_{Nt}%
,\label{EQ:fstar-f0}%
\end{equation}
and $||\upsilon_{Nt}||=o_{p}(N^{-1/2})$ for every $t$, where $\Lambda
_{Nt}=N^{-1}\sum\nolimits_{i=1}^{N}p_{i}\left(  0\left\vert X_{i}%
,f_{t}\right.  \right)  G_{i}^{0}(X_{i})G_{i}^{0}(X_{i})^{\intercal}$. By
Conditions (C2) and (C5), we have that the eigenvalues of $\Lambda
_{Nt}^{0}$ are bounded away from zero and infinity. By similar reasoning to
the proof for Theorem 2 in Lee and Robinson (2016), we have $\left\Vert
\Lambda_{Nt}^{-1}\right\Vert =O_{p}(1)$ and $\left\Vert \Lambda_{Nt}%
-\Lambda_{Nt}^{0}\right\Vert =o_{p}(1)$. Thus, the asymptotic distribution in
Lemma \ref{LEM:fstar} can be obtained directly by Condition (C6).
\end{proof}

Recall that the initial estimator $\widehat{f}_{t}^{[0]}$ given in
(\ref{EQ:fhat0}) is defined in the same way as $f_{t}^{\ast}$ with $g_{j}%
^{0}(X_{ji})$ replaced by $\widehat{g}_{j}^{[0]}(X_{ji})$ in
(\ref{EQ:fastarobj}).\ Then we have the following result for $\widehat{f}%
_{t}^{[0]}$.

\begin{lemma}
\label{LEM:fhat0} Let Conditions (C1)-(C5) hold. If, in addition,
$K_{N}^{4}N^{-1}=o(1)$, $K_{N}^{-r+2}(\log T)=o(1)$ and $K_{N}^{-1}(\log
NT)(\log N)^{4}=o(1)$, then for any $t$ there is a stochastically bounded
sequence $\delta_{N,jt}$ such that as $N\rightarrow\infty$,
\[
\sqrt{N}(\widehat{f}_{t}^{[0]}-f_{t}^{\ast}-d_{NT}\delta_{N,t})=o_{p}(1),
\]
where $\delta_{N,t}=(\delta_{N,jt},0\leq j\leq J)^{\intercal}$ and $d_{NT}$ is
given in (\ref{dNT}).
\end{lemma}

\begin{proof}
Denote $g=\{g_{j}(\cdot),1\leq j\leq J\}$. Define
\begin{align*}
L_{Nt}(f_{t},g) &  =N^{-1}\sum\nolimits_{i=1}^{N}\rho_{\tau}(y_{it}%
-f_{ut}-\sum\nolimits_{j=1}^{J}g_{j}(X_{ji})f_{jt})\\
&  -N^{-1}\sum\nolimits_{i=1}^{N}\rho_{\tau}(y_{it}-f_{ut}^{0}-\sum
\nolimits_{j=1}^{J}g_{j}(X_{ji})f_{jt}^{0}),
\end{align*}
so that $f_{t}^{\ast}$ and $\widehat{f}_{t}^{[0]}$ are the minimizers of
$L_{Nt}(f_{t},g^{0})$ and $L_{Nt}(f_{t},\widehat{g}^{[0]})$, respectively,
where $\widehat{g}^{[0]}=\{\widehat{g}_{j}^{[0]}(\cdot),1\leq j\leq J\}$ and
$g^{0}=\{g_{j}^{0}(\cdot),1\leq j\leq J\}$. According to the result on page
149 of de Boor (2001), for $g_{j}^{0}$ satisfying the smoothness condition
given in (C2), there exists $\boldsymbol{\lambda}_{j}^{0}\in R^{K_{n}}$ such
that $g_{j}^{0}(x_{j})=\widetilde{g}_{j}^{0}(x_{j})+r_{j}(x_{j}),$
\[
\widetilde{g}_{j}^{0}(x_{j})=B_{j}(x_{j})^{\intercal}\boldsymbol{\lambda}%
_{j}^{0}\text{ and }\sup_{j}\sup\nolimits_{x_{j}\in\lbrack a,b]}|r_{j}%
(x_{j})|=O(K_{N}^{-r}).
\]
By Proposition \ref{THM:ghat0}, there exists $\boldsymbol{\lambda}_{j,NT}\in
R^{K_{N}}$ such that $\widehat{g}_{j}^{[0]}(x_{j})=B_{j}(x_{j})^{\intercal
}\boldsymbol{\lambda}_{j,NT}$ and $||\boldsymbol{\lambda}_{j,NT}%
-\boldsymbol{\lambda}_{j}^{0}||=O_{p}(d_{NT})+o_{p}(N^{-1/2})$. Let %
$d_{NT}^{\prime}$ be a sequence satisfying $d_{NT}^{\prime}=o(N^{-1/2})$ and
let $d_{NT}^{\ast}=d_{NT}+d_{NT}^{\prime}$. In the following, we will
show that
\begin{equation}
\widetilde{f}_{t}-f_{t}^{0}-d_{NT}\delta_{N,t}=\Lambda_{Nt}^{-1}\{N^{-1}%
\sum\nolimits_{i=1}^{N}G_{i}^{0}(X_{i})(\tau-I(\varepsilon_{it}<0))\}+o_{p}%
(N^{-1/2}),\label{EQ:ftilda-f0}%
\end{equation}
uniformly in $||\boldsymbol{\lambda}_{j}-\boldsymbol{\lambda}_{j}^{0}%
||\leq\widetilde{C}$$d_{NT}^{\ast}$ for some constant
$0<\widetilde{C}<\infty$, where $\widetilde{f}_{t}$ is the minimizer of
$L_{Nt}(f_{t},g)$ and $g_{j}(x_{j})=B_{j}(x_{j})^{\intercal}%
\boldsymbol{\lambda}_{j}$. Hence the result in Lemma \ref{LEM:fhat0} follows
from (\ref{EQ:fstar-f0}) and (\ref{EQ:ftilda-f0}).

We have $||\widetilde{f}_{t}-f_{t}^{\ast}||=o_{p}(1)$, since
\begin{align*}
&  |L_{Nt}(f_{t},g)-L_{Nt}(f_{t},g^{0})|\\
&  \leq2N^{-1}\sum\nolimits_{i=1}^{N}|\sum\nolimits_{j=1}^{J}\{g_{j}%
(X_{ji})-g_{j}^{0}(X_{ji})\}f_{jt}|+2N^{-1}\sum\nolimits_{i=1}^{N}%
|\sum\nolimits_{j=1}^{J}\{g_{j}(X_{ji})-g_{j}^{0}(X_{ji})\}f_{jt}^{0}|\\
&  \leq C_{L}\widetilde{C}\{d_{NT}+o(N^{-1/2})\}=o(1),
\end{align*}
for some constant $0<C_{L}<\infty$, where the first inequality follows from
the fact that $|\rho_{\tau}(u-v)-\rho_{\tau}(u)|\leq2|v|$. Thus
$||\widetilde{f}_{t}-f_{t}^{0}||=o_{p}(1)$. Let $\mathbf{X=(}X_{1}%
,\ldots,X_{N})^{\intercal}$, $G_{i}(X_{i})=\{1,g_{1}(X_{1i}),\ldots
,g_{J}(X_{Ji})\}^{\intercal}$ and $\mathbf{F}=\{f_{1},\ldots,f_{T}%
\}^{\intercal}$. Let
\[
\psi_{\tau}(\varepsilon)=\tau-I(\varepsilon<0).
\]
For $\boldsymbol{\lambda}_{j}\in R^{K_{n}}$ satisfying
$||\boldsymbol{\lambda}_{j}-\boldsymbol{\lambda}_{j}^{0}||\leq\widetilde{C}$$d_{NT}^{\ast}$ and $f_{t}$ in a neighborhood of $f_{t}^{0}$, write
\begin{align}
&  L_{Nt}(f_{t},g)=E\{L_{Nt}(f_{t},g)|\mathbf{X,F}\}-(f_{t}-f_{t}%
^{0})^{\intercal}\{W_{Nt,1}-E(W_{Nt,1}|\mathbf{X,F})\}\nonumber\\
&  +W_{Nt,2}(f_{t},g)-E(W_{Nt,2}(f_{t},g)|\mathbf{X,F}),\label{EQ:LNtftg}%
\end{align}
where $g_{j}(x_{j})=B_{j}(x_{j})^{\intercal}\boldsymbol{\lambda}_{j}$, and
\begin{align}
W_{Nt,1} &  =N^{-1}\sum\nolimits_{i=1}^{N}G_{i}(X_{i})\psi_{\tau}(y_{it}%
-f_{t}^{0\intercal}G_{i}(X_{i})),\label{EQ:WNt1}\\
W_{Nt,2}(f_{t},g) &  =N^{-1}\sum\nolimits_{i=1}^{N}\{\rho_{\tau}(y_{it}%
-f_{t}^{\intercal}G_{i}(X_{i}))-\rho_{\tau}(y_{it}-f_{t}^{0\intercal}%
G_{i}(X_{i}))\label{EQ:WNt2}\\
&  +(f_{t}-f_{t}^{0})^{\intercal}G_{i}(X_{i})\psi_{\tau}(y_{it}-f_{t}%
^{0\intercal}G_{i}(X_{i}))\}.\nonumber
\end{align}
In Lemma \ref{LEM:LNtfg} in the Supplemental Materials, we show that as $N\rightarrow\infty$,
\[
E\{L_{Nt}(f_{t},g)|\mathbf{X,F}\}=-(f_{t}-f_{t}^{0})^{\intercal}%
E(W_{Nt,1}|\mathbf{X,F)+}\frac{1}{2}(f_{t}-f_{t}^{0})^{\intercal}\Lambda
_{Nt}^{0}(f_{t}-f_{t}^{0})+o_{p}(||f_{t}-f_{t}^{0}||^{2}),
\]
uniformly in $||\boldsymbol{\lambda}_{j}-\boldsymbol{\lambda}_{j}^{0}%
||\leq\widetilde{C}$$d_{NT}^{\ast}$ and $||f_{t}-f_{t}^{0}%
||\leq\varpi_{N}$, where $\varpi_{N}$ is any sequence of positive numbers
satisfying $\varpi_{N}=o(1)$. Substituting this into (\ref{EQ:LNtftg}), we
have with probability approaching $1$,
\begin{align*}
L_{Nt}(f_{t},g) &  =-(f_{t}-f_{t}^{0})^{\intercal}W_{Nt,1}\mathbf{+}\frac
{1}{2}(f_{t}-f_{t}^{0})^{\intercal}\Lambda_{N}^{0}(f_{t}-f_{t}^{0})\\
&  +W_{Nt,2}(f_{t},g)-E(W_{Nt,2}(f_{t},g)|\mathbf{X,F)+}o(||f_{t}-f_{t}%
^{0}||^{2}).
\end{align*}
In Lemma \ref{LEM:WNt2} in the Supplemental Materials, we show that
\[
W_{Nt,2}(f_{t},g)-E(W_{Nt,2}%
(f_{t},g)|\mathbf{X,F)=}o_{p}(||f_{t}-f_{t}^{0}||^{2}+N^{-1}),
\]
uniformly in
$||\boldsymbol{\lambda}_{j}-\boldsymbol{\lambda}_{j}^{0}||\leq\widetilde{C}%
d_{NT}$ and $||f_{t}-f_{t}^{0}||\leq\varpi_{N}$. Thus, we have $\widetilde{f}%
_{t}-f_{t}^{0}=(\Lambda_{Nt}^{0})^{-1}W_{Nt,1}+o_{p}(N^{-1/2})$. Since
$||(\Lambda_{Nt}^{0})^{-1}-(\Lambda_{Nt})^{-1}||=o_{p}(1)$, we have
\begin{equation}
\widetilde{f}_{t}-f_{t}^{0}=\Lambda_{Nt}^{-1}W_{Nt,1}+o_{p}(N^{-1/2}%
).\label{EQ:ftilda}%
\end{equation}
In Lemma \ref{LEM:WNt1} in the Supplemental Materials, we show that for any $t$ there is a
stochastically bounded sequence $\delta_{N,jt}$ such that as $N\rightarrow
\infty$,
\begin{equation}
W_{Nt,1}=N^{-1}\sum\nolimits_{i=1}^{N}G_{i}^{0}(X_{i})\psi_{\tau}%
(\varepsilon_{it})+d_{NT}\delta_{N,t}+o_{p}(N^{-1/2}).\label{EQ:WNt1ap}%
\end{equation}
where $\delta_{N,t}=(\delta_{N,jt},0\leq j\leq J)^{\intercal}$ and
$g_{j}(x_{j})=B_{j}(x_{j})^{\intercal}\boldsymbol{\lambda}_{j}$, uniformly in
$||\boldsymbol{\lambda}_{j}-\boldsymbol{\lambda}_{j}^{0}||\leq\widetilde{C}$$d_{NT}^{\ast}$. Hence, result (\ref{EQ:ftilda-f0}) follows from
(\ref{EQ:ftilda}) and (\ref{EQ:WNt1ap}) directly. Then the proof is complete.
\end{proof}

Let $\boldsymbol{\lambda}=(\boldsymbol{\lambda}_{1}^{\intercal}%
,\ldots,\boldsymbol{\lambda}_{J}^{\intercal})^{\intercal}$. For given
$\widehat{f}^{[0]}$, we obtain
\[
\widehat{\boldsymbol{\lambda}}^{[1]}=(\boldsymbol{\lambda}_{1}^{[1]\intercal
},\ldots,\boldsymbol{\lambda}_{J}^{[1]\intercal})^{\intercal}=\arg
\min_{\boldsymbol{\lambda}}\{(NT)^{-1}\sum\nolimits_{i=1}^{N}\sum
\nolimits_{t=1}^{T}\rho_{\tau}(y_{it}-\widehat{f}_{ut}^{[0]}-\sum
\nolimits_{j=1}^{J}B_{j}(X_{ji})^{\intercal}\boldsymbol{\lambda}%
_{j}\widehat{f}_{jt}^{[0]})\}.
\]
Let $\widehat{g}_{j}^{\ast\lbrack1]}(x_{j})$ $=B_{j}(x_{j})^{\intercal
}\widehat{\boldsymbol{\lambda}}_{j}^{[1]}$. The estimator for $g_{j}(x_{j})$
at the $1^{\text{st}}$ step is
\[
\widehat{g}_{j}^{[1]}(x_{j})=\widehat{g}_{j}^{\ast\lbrack1]}(x_{j}%
)/\sqrt{N^{-1}\sum\nolimits_{i=1}^{N}\widehat{g}_{j}^{\ast\lbrack1]}%
(X_{ji})^{2}}.
\]
We define the infeasible estimator of $\boldsymbol{\lambda}$ as
\[
\boldsymbol{\lambda}^{\ast}=(\boldsymbol{\lambda}_{1}^{\ast\intercal}%
,\ldots,\boldsymbol{\lambda}_{J}^{\ast\intercal})^{\intercal}=\arg
\min_{\boldsymbol{\lambda}}\{(NT)^{-1}\sum\nolimits_{i=1}^{N}\sum
\nolimits_{t=1}^{T}\rho_{\tau}(y_{it}-f_{ut}^{0}-\sum\nolimits_{j=1}^{J}%
B_{j}(X_{ji})^{\intercal}\boldsymbol{\lambda}_{j}f_{jt}^{0})\}.
\]
Let $g_{j}^{\ast}(x_{j})$ $=B_{j}(x_{j})^{\intercal}\boldsymbol{\lambda}%
_{j}^{\ast}$ and $\widetilde{g}_{j}^{\ast}(x_{j})$ $=g_{j}^{\ast}(x_{j}%
)/\sqrt{N^{-1}\sum\nolimits_{i=1}^{N}g_{j}^{\ast}(X_{ji})^{2}}$.

\begin{lemma}
\label{LEM:gjtilda} Let Conditions (C1)--(C5) hold. If, in addition,
$K_{N}^{4}N^{-1}=o(1)$, $K_{N}^{-r+2}(\log T)=o(1)$ and $K_{N}^{-1}(\log
NT)(\log N)^{4}=o(1)$, then for every $1\leq j\leq J$,
\begin{equation}
\left[  \int\{\widetilde{g}_{j}^{\ast}(x_{j})-g_{j}^{0}(x_{j})\}^{2}%
dx_{j}\right]  ^{1/2}=O_{p}(K_{N}^{1/2}(NT)^{-1/2}+K_{N}^{-r}),
\label{EQ:gtilda1}%
\end{equation}
and
\begin{equation}
\int\{\widehat{g}_{j}^{[1]}(x_{j})(x_{j})-\widetilde{g}_{j}^{\ast}%
(x_{j})\}^{2}dx_{j}=O_{p}(d_{NT}^{2})+o_{p}(N^{-1/2}). \label{EQghat1}%
\end{equation}
Therefore, for every $1\leq j\leq J$,
\begin{equation}
\int\{\widehat{g}_{j}^{[1]}(x_{j})-g_{j}^{0}(x_{j})\}^{2}dx_{j}=O_{p}%
(d_{NT}^{2})+o_{p}(N^{-1/2}). \label{EQ:ghat1xj}%
\end{equation}
\end{lemma}

\begin{proof}
{\normalsize Denote $\widetilde{g}^{0}(x)=\{\widetilde{g}_{j}^{0}(x_{j}),1\leq
j\leq J\}^{\intercal}$ and $g^{\ast}(x)=\{g_{j}^{\ast}(x_{j}),1\leq j\leq
J\}^{\intercal}$. Let $\boldsymbol{\lambda}^{0}=(\boldsymbol{\lambda}%
_{1}^{0\intercal},\ldots,\boldsymbol{\lambda}_{J}^{0\intercal})^{\intercal}$.
Let $\mathbb{B}^{\ast}(x)=\left[  \text{diag}\left[  B_{1}(x_{1})^{\intercal
},\ldots,B_{J}(x_{J})^{\intercal}\right]  \right]  _{J\times JK_{N}}$. Then
$\mathbb{B}^{\ast}(x)\boldsymbol{\lambda}^{\ast}=g^{\ast}(x)$ and
$\mathbb{B}^{\ast}(x)\boldsymbol{\lambda}^{0}=\widetilde{g}^{0}(x)$. Let
$Q_{it}^{0}=\{B_{j}(X_{ji})^{\intercal}f_{jt}^{0},1\leq j\leq J\}^{\intercal}%
$,
\begin{equation}
\Psi_{NT}=(NT)^{-1}\sum\nolimits_{i=1}^{N}\sum\nolimits_{t=1}^{T}%
f_{\varepsilon}\left(  0\left\vert X_{i},f_{t}\right.  \right)  Q_{it}%
^{0}Q_{it}^{0\intercal},\label{EQ:CphNT}%
\end{equation}
and $r_{j,it}^{\ast}=$ $r_{j}(X_{ji})f_{jt}^{0}$. Moreover, define
\begin{align}
U_{NT,1} &  =(NT)^{-1}\sum\nolimits_{i=1}^{N}\sum\nolimits_{t=1}^{T}Q_{it}%
^{0}(\tau-I(\varepsilon_{it}<0)),\label{EQ:UNT1}\\
U_{NT,2} &  =(NT)^{-1}\sum\nolimits_{i=1}^{N}\sum\nolimits_{t=1}^{T}Q_{it}%
^{0}f_{\varepsilon}(0|X_{i},f_{t})\left(  \sum\nolimits_{j=1}^{J}%
r_{j,it}^{\ast}\right)  .\nonumber
\end{align}
By the same procedure as the proof of Lemma \ref{LEM:bahardur}, for $K_{N}%
^{4}(\log(NT))^{2}(NT)^{-1}=o(1)$, we obtain the Bahadur representation for
$\boldsymbol{\lambda}^{\ast}-\boldsymbol{\lambda}^{0}$ as
\begin{equation}
\boldsymbol{\lambda}^{\ast}-\boldsymbol{\lambda}^{0}=\Psi_{NT}^{-1}%
(U_{N,1}+U_{N,2})+R_{NT}^{\ast},\label{EQ:lambdastar}%
\end{equation}
and the remaining term $R_{NT}^{\ast}$ satisfies
\begin{align*}
||R_{NT}^{\ast}|| &  =O_{p}(K_{N}^{3/2}(NT)^{-1}+K_{N}^{3/2}(NT)^{-3/4}%
\sqrt{\log(NT)}+K_{N}^{1/2-2r}+(NT)^{-1/2}K_{N}^{-r/2+1/2})\\
&  =O_{p}(K_{N}^{3/2}(NT)^{-3/4}\sqrt{\log(NT)}+K_{N}^{1/2-2r})+o_{p}%
((NT)^{-1/2}).
\end{align*}
By (\ref{EQ:lambdastar}) and following the same reasoning as the proof for
(\ref{EQ:htilda-ho}), we have $\sup_{x_{j}\in\lbrack a,b]}|g_{j}^{\ast}%
(x_{j})-g_{j}^{0}(x_{j})|=O_{p}(K_{N}(NT)^{-1/2}+K_{N}^{-r})$, $[\int%
\{g_{j}^{\ast}(x_{j})-g_{j}^{0}(x_{j})\}^{2}dx_{j}]^{1/2}=O_{p}(K_{N}%
^{1/2}(NT)^{-1/2}+K_{N}^{-r})$, and $[N^{-1}\sum\nolimits_{i=1}^{N}%
\{g_{j}^{\ast}(X_{ji})-g_{j}^{0}(X_{ji})\}^{2}]^{1/2}=O_{p}(K_{N}%
^{1/2}(NT)^{-1/2}+K_{N}^{-r})$. Therefore, we have
\[
\{\sqrt{N^{-1}\sum\nolimits_{i=1}^{N}g_{j}^{\ast}(X_{ji})^{2}}\}^{-1}%
-\{\sqrt{N^{-1}\sum\nolimits_{i=1}^{N}g_{j}^{0}(X_{ji})^{2}}\}^{-1}%
=O_{p}(K_{N}^{1/2}(NT)^{-1/2}+K_{N}^{-r}),
\]
and thus
\begin{align*}
\sup_{x_{j}\in\lbrack a,b]}|\widetilde{g}_{j}^{\ast}(x_{j})-g_{j}^{0}(x_{j})|
&  =O_{p}(K_{N}(NT)^{-1/2}+K_{N}^{-r}),\\
\lbrack\int\{\widetilde{g}_{j}^{\ast}(x_{j})-g_{j}^{0}(x_{j})\}^{2}%
dx_{j}]^{1/2} &  =O_{p}(K_{N}^{1/2}(NT)^{-1/2}+K_{N}^{-r}).
\end{align*}
Then the result (\ref{EQ:gtilda1}) is proved. Define
\begin{align*}
L_{NT}^{\ast}(f,\boldsymbol{\lambda}) &  =(NT)^{-1}\sum\nolimits_{i=1}^{N}%
\sum\nolimits_{t=1}^{T}\rho_{\tau}(y_{it}-f_{ut}-\sum\nolimits_{j=1}^{J}%
B_{j}(X_{ji})^{\intercal}\boldsymbol{\lambda}_{j}f_{jt})\\
&  -(NT)^{-1}\sum\nolimits_{i=1}^{N}\sum\nolimits_{t=1}^{T}\rho_{\tau}%
(y_{it}-f_{ut}-\sum\nolimits_{j=1}^{J}B_{j}(X_{ji})^{\intercal}%
\boldsymbol{\lambda}_{j}^{0}f_{jt}).
\end{align*}
Hence, $\widehat{\boldsymbol{\lambda}}^{[1]}$ and $\boldsymbol{\lambda}^{\ast
}$ are the minimizers of $L_{NT}^{\ast}(\widehat{f}^{[0]},\boldsymbol{\lambda
})$ and $L_{NT}^{\ast}(f^{0},\boldsymbol{\lambda})$, respectively. In Lemma
\ref{LEM:VNt11} in the Supplemental Materials, we show that
\begin{equation}
||\widehat{\boldsymbol{\lambda}}^{[1]}-\boldsymbol{\lambda}^{0}-\Psi_{NT}%
^{-1}U_{N,1}||=O_{p}(d_{NT})+o_{p}(N^{-1/2}).\label{EQ:lambdatilda-lamda0}%
\end{equation}
Hence, by (\ref{EQ:lambdastar}), (\ref{EQ:lambdatilda-lamda0}) and
$||\Psi_{NT}^{-1}U_{N,2}||=O(K_{N}^{-r})$, we have
\begin{equation}
||\widehat{\boldsymbol{\lambda}}^{[1]}-\boldsymbol{\lambda}^{\ast}%
||=O_{p}(d_{NT})+o_{p}(N^{-1/2}).\label{EQ:lambdatilda}%
\end{equation}
Then we have $\int\{\widehat{g}_{j}^{\ast\lbrack1]}(x_{j})-g_{j}^{\ast}%
(x_{j})\}^{2}dx_{j}=O_{p}(d_{NT}^{2})$}$+o_{p}(N^{-1})$ {\normalsize  and
$N^{-1}\sum\nolimits_{i=1}^{N}\{\widehat{g}_{j}^{\ast\lbrack1]}(X_{ji}%
)-g_{j}^{\ast}(X_{ji})\}^{2}=O_{p}(d_{NT}^{2})$}$+o_{p}(N^{-1})$.
Thus,
\[
\{\sqrt{N^{-1}\sum\nolimits_{i=1}^{N}\widehat{g}_{j}^{\ast\lbrack1]}%
(X_{ji})^{2}}\}^{-1}-\{\sqrt{N^{-1}\sum\nolimits_{i=1}^{N}g_{j}^{\ast}%
(X_{ji})^{2}}\}^{-1}={\small O_{p}(d_{NT})}+o_{p}(N^{-1/2}).
\]
Therefore, the result (\ref{EQghat1}) follows from the above results directly.
\end{proof}

\begin{proof}
[Proofs of Theorems \ref{THM:fhat} and \ref{THM:ghat}] Based on
(\ref{EQ:ghat1xj}) in Lemma \ref{LEM:gjtilda}, the result in Lemma
\ref{LEM:fhat0} holds for $\widehat{f}_{t}^{[1]}$ with a different bounded
sequence. Then the result (\ref{EQ:ghat1xj}) in Lemma \ref{LEM:gjtilda} holds
for $\widehat{g}_{j}^{[2]}(x_{j})$. This process can be continued for any
finite number of iterations. By assuming that the algorithm in Section
\ref{estimator} converges at the $(i+1)^{\text{th}}$ step for any finite
number $i$, the results in Lemmas \ref{LEM:fhat0} and \ref{LEM:gjtilda} hold
for $\widehat{f}_{t}=\widehat{f}_{t}^{[i+1]}$ and $\widehat{g}_{j}%
=\widehat{g}_{j}^{[i+1]}(x_{j})$. Hence, Theorem \ref{THM:fhat} for
$\widehat{f}_{t}$ follows from Lemmas \ref{LEM:fstar} and \ref{LEM:fhat0},
directly, and Theorem \ref{THM:ghat} for $\widehat{g}_{j}$ is proved by using
Lemma \ref{LEM:gjtilda}.
\end{proof}

\subsection{{\protect\normalsize Proofs of Theorem \ref{THM:consistentGAMMA}}}

\begin{proof}
We prove the consistency of $\widehat{\Lambda}_{Nt}$. Define
\[
\widetilde{\Lambda}_{Nt}=(Nh)^{-1}\sum\nolimits_{i=1}^{N} K\left(
\frac{y_{it}-(f_{ut}^{0}+\sum\nolimits_{j=1}^{J}g_{j}^{0}(X_{ji})f_{jt}^{0}%
)}{h}\right)  G_{i}^{0}(X_{i})G_{i}^{0}(X_{i})^{\intercal},
\]
and
\[
\widehat{\Lambda}_{Nt}=(Nh)^{-1}\sum\nolimits_{i=1}^{N}K\left(  \frac
{y_{it}-(\widehat{f}_{ut}+\sum\nolimits_{j=1}^{J}\widehat{g}_{j}%
(X_{ji})\widehat{f}_{jt})}{h}\right)  \widehat{G}_{i}(X_{i})\widehat{G}%
_{i}(X_{i})^{\intercal}.
\]
We will show $||\widehat{\Lambda}_{Nt}-\widetilde{\Lambda}_{Nt}||=o_{p}(1)$
and $||\widetilde{\Lambda}_{Nt}-\Lambda_{Nt}^{0}||=o_{p}(1)$, respectively.
Let $\widehat{d}_{it}(X_{i})=\{ \widehat{f}_{ut}+\sum\nolimits_{j=1}%
^{J}\widehat{g}_{j}(X_{ji})\widehat{f}_{jt}\}-\{f_{ut}^{0}+\sum\nolimits_{j=1}%
^{J}g_{j}^{0}(X_{ji})f_{jt}^{0}\}$. Then,
\[
\widehat{\Lambda}_{Nt}-\widetilde{\Lambda}_{Nt}=D_{Nt,1}+D_{Nt,2},
\]
where
\begin{align*}
D_{Nt,1}  &  =(2Nh)^{-1}\sum\nolimits_{i=1}^{N}\{I(|\varepsilon_{it}|\leq
h)-I(|\varepsilon_{it}-\widehat{d}_{it}(X_{i})|\leq h)\}G_{i}^{0}(X_{i}%
)G_{i}^{0}(X_{i})^{\intercal},\\
D_{Nt,2}  &  =(2Nh)^{-1}\sum\nolimits_{i=1}^{N}I(|\varepsilon_{it}%
-\widehat{d}_{it}(X_{i})|\leq h)\{ \widehat{G}_{i}(X_{i})\widehat{G}_{i}%
(X_{i})^{\intercal}-G_{i}^{0}(X_{i})G_{i}^{0}(X_{i})^{\intercal}\}.
\end{align*}
Since there exist some constants $0<c_{f},c_{1}<\infty$ such that with
probability approaching 1,
\[
E\{ \widehat{d}_{it}(X_{i})\}^{2}=\int\widehat{d}_{it}^{2}(x)f_{X_{i}%
}(x)dx\leq c_{f}\int\widehat{d}_{it}^{2}(x)dx\leq c_{1}\phi_{NT}^{2}%
+o(N^{-1}),
\]
where $\phi_{NT}$ is given in (\ref{EQ:phNT}), and the last inequality follows
from the result in Theorem \ref{THM:ghat}, then\ there exists some constant
$0<c<\infty$ such that with probability approaching 1,
\begin{align*}
E||\widehat{\Lambda}_{Nt}-\widetilde{\Lambda}_{Nt}||  &  \leq c(2Nh)^{-1}%
\sum\nolimits_{i=1}^{N}E|\widehat{d}_{it}(X_{i})|\times||G_{i}^{0}(X_{i}%
)G_{i}^{0}(X_{i})^{\intercal}||\\
&  \leq c(2Nh)^{-1}\sum\nolimits_{i=1}^{N}E\{ \widehat{d}_{it}(X_{i}%
)\}^{2}E||G_{i}^{0}(X_{i})G_{i}^{0}(X_{i})^{\intercal}||^{2}\}^{1/2}\\
&  \leq cc_{1}^{1/2}(2Nh)^{-1}(\sqrt{K_{N}/(NT)}+K_{N}^{3/2}N^{-3/4}\sqrt{\log
N}+K_{N}^{-r})\times\\
&  \sum\nolimits_{i=1}^{N}\{E||G_{i}^{0}(X_{i})G_{i}^{0}(X_{i})^{\intercal
}||^{2}\}^{1/2}.
\end{align*}
By Condition (C3), we have $\sup_{x_{j}\in\lbrack a,b]}|g_{j}^{0}(x_{j})|\leq
C^{\prime}$ for all $j$, for any vector $\mathbf{a\in}R^{J+1}$ and
$||\mathbf{a||}^{2}=1$, we have
\begin{align*}
\mathbf{a}^{\intercal}G_{i}^{0}(X_{i})G_{i}^{0}(X_{i})^{\intercal}\mathbf{a}
&  \mathbf{=\{}a_{0}+\sum\nolimits_{j=1}^{J}g_{j}^{0}(X_{ji})a_{j}\}^{2}%
\leq(J+1)\{a_{0}^{2}+g_{j}^{0}(X_{ji})^{2}a_{j}^{2}\}\\
&  \leq(J+1)\{a_{0}^{2}+(C^{\prime})^{2}a_{j}^{2}\} \leq C_{a}%
\end{align*}
for some constant $0<C_{a}<\infty$. Hence, $||G_{i}^{0}(X_{i})G_{i}^{0}%
(X_{i})^{\intercal}||\leq C_{a}$, and thus we have
\begin{align*}
&  E||\widehat{\Lambda}_{Nt}-\widetilde{\Lambda}_{Nt}||\leq cc_{1}%
^{1/2}(2Nh)^{-1}(\phi_{NT}+o(N^{-1/2}))\sum\nolimits_{i=1}^{N}C_{a}\\
&  =2^{-1}cc_{1}^{1/2}C_{a}h^{-1}(\phi_{NT}+o(N^{-1/2}))=o(1)
\end{align*}
by the assumption that $h^{-1}\phi_{NT}=o(1)$ and $h^{-1}N^{-1/2}=O(1)$.
Hence, we have $||D_{Nt,1}||=o_{p}(1)$. Moreover, for any vector
$\mathbf{a\in}R^{J+1}$ and $||\mathbf{a||}^{2}=1$, we have with probability
approaching 1, there exists a constant $0<C<\infty$ such that
\begin{align*}
|\mathbf{a}^{\intercal}D_{Nt,2}\mathbf{a|}  &  \mathbf{\leq}(2Nh)^{-1}%
\sum\nolimits_{i=1}^{N}|\mathbf{\{}a_{0}+\sum\nolimits_{j=1}^{J}%
\widehat{g}_{j}(X_{ji})a_{j}\}^{2}-\mathbf{\{}a_{0}+\sum\nolimits_{j=1}%
^{J}g_{j}^{0}(X_{ji})a_{j}\}^{2}|\\
&  \leq C(2Nh)^{-1}\sum\nolimits_{i=1}^{N}\sum\nolimits_{j=1}^{J}|\{
\widehat{g}_{j}(X_{ji})-g_{j}^{0}(X_{ji})\}a_{j}|\\
&  \leq C(2h)^{-1}\sum\nolimits_{j=1}^{J}\{N^{-1}\sum\nolimits_{i=1}^{N}\{
\widehat{g}_{j}(X_{ji})-g_{j}^{0}(X_{ji})\}^{2}a_{j}^{2}\}^{1/2}\\
&  =O(h^{-1})\{O(\phi_{NT})+o(N^{-1/2})\}=o(1).
\end{align*}
Hence, we have $||D_{Nt,2}||=o_{p}(1)$. Therefore, $||\widehat{\Lambda}%
_{Nt}-\widetilde{\Lambda}_{Nt}||\leq||D_{Nt,1}||+||D_{Nt,2}||=o_{p}(1)$. Next,
we will show $||\widetilde{\Lambda}_{Nt}-\Lambda_{Nt}^{0}||=o_{p}(1)$. Since
\begin{align*}
&  |E\left\{  (2h)^{-1}I(|\varepsilon_{it}|\leq h)-p_{i}\left(  0\left\vert
X_{i}\right.  \right)  |X_{i},f_{t}\right\}  |\\
&  =|(2h)^{-1}h\{p_{i}\left(  h^{\ast}\left\vert X_{i},f_{t}\right.  \right)
+p_{i}\left(  -h^{\ast\ast}\left\vert X_{i},f_{t}\right.  \right)
\}-p_{i}\left(  0\left\vert X_{i},f_{t}\right.  \right)  |\\
&  =|2^{-1}[\{p_{i}\left(  h^{\ast}\left\vert X_{i},f_{t}\right.  \right)
-p_{i}\left(  0\left\vert X_{i},f_{t}\right.  \right)  \}+\{p_{i}\left(
-h^{\ast\ast}\left\vert X_{i},f_{t}\right.  \right)  -p_{i}\left(  0\left\vert
X_{i},f_{t}\right.  \right)  \}]|\leq c^{\prime}h
\end{align*}
for some constant $0<c^{\prime}<\infty$, where $h^{\ast}$ and $h^{\ast\ast}$
are some values between $0$ and $h$, and the last inequality follows from
Condition (C2), then by the above result and Condition (C5),
\begin{align}
||E(\widetilde{\Lambda}_{Nt}-\Lambda_{Nt}^{0})||  &  =||N^{-1}\sum
\nolimits_{i=1}^{N}E[\{(2h)^{-1}I(|\varepsilon_{it}|\leq h)-p_{i}\left(
0\left\vert X_{i},f_{t}\right.  \right)  \}G_{i}^{0}(X_{i})G_{i}^{0}%
(X_{i})^{\intercal}]||\nonumber\\
&  \leq c^{\prime}h||N^{-1}\sum\nolimits_{i=1}^{N}EQ_{i}^{0}(X_{i})G_{i}%
^{0}(X_{i})^{\intercal}||=O(h)=o(1). \label{EQ:ELamtilda}%
\end{align}
Moreover, by Conditions (C1), we have $E\{I(|\varepsilon_{it}|\leq h)\}
\leq2C^{\ast}h$ for some constant $C^{\ast}\in(0,\infty)$, and then for any
vector $\mathbf{a}\in R^{(J+1)}$ with $||\mathbf{a||=}1$, by Conditions (C1),
(C2) and (C3), we have
\begin{align}
&  \text{var}(\mathbf{a}^{\intercal}\widetilde{\Lambda}_{Nt}\mathbf{a)}%
\nonumber\\
&  =(2Nh)^{-2}\text{var}\left(  \sum\nolimits_{i=1}^{N}I(|\varepsilon
_{it}|\leq h)\mathbf{\{}a_{0}+\sum\nolimits_{j=1}^{J}g_{j}^{0}(X_{ji}%
)a_{j}\}^{2}\right) \nonumber\\
&  \leq(2Nh)^{-2}\sum\nolimits_{i,i^{\prime}}2\{ \phi(|i-i^{\prime}%
|)\}^{1/2}\times\nonumber\\
&  \left(  E\left[  I\left(  |\varepsilon_{it}|\leq h\right)  \mathbf{\{}%
a_{0}+\sum\nolimits_{j=1}^{J}g_{j}^{0}(X_{ji})a_{j}\}^{4}\right]  \right)
^{1/2}\left(  E\left[  I(|\varepsilon_{i^{\prime}t^{\prime}}|\leq
h)\mathbf{\{}a_{0}+\sum\nolimits_{j=1}^{J} g_{j}^{0}(X_{ji^{\prime}}%
)a_{j}\}^{4}\right]  \right)  ^{1/2}\nonumber\\
&  \leq(J+1)^{2}\{a_{0}^{2}+C^{\prime2}a_{j}^{2}\}(2Nh)^{-2}(2C^{\ast}%
h)^{2}\sum\nolimits_{i,i^{\prime}}2\{ \phi(|i-i^{\prime}|)\}^{1/2}\nonumber\\
&  \leq(J+1)^{2}\{a_{0}^{2}+C^{\prime2}a_{j}^{2}\}N^{-2}2C^{\ast2}K_{1}%
\sum\nolimits_{i,i^{\prime}}e^{-(\lambda_{1}/2)(|i-i^{\prime}|)}\nonumber\\
&  \leq(J+1)^{2}\{a_{0}^{2}+C^{\prime2}a_{j}^{2}\}2C^{\ast2}K_{1}%
N^{-1}\{1-e^{-(\lambda_{1}/2)}\}=O(N^{-1})=o(1). \label{EQ:VLamtilda}%
\end{align}
By (\ref{EQ:ELamtilda}) and (\ref{EQ:VLamtilda}), we have
$||\widetilde{\Lambda}_{Nt}-\Lambda_{Nt}^{0}||=o_{p}(1)$. Hence,
$||\widehat{\Lambda}_{Nt}-\Lambda_{Nt}^{0}||\leq||\widehat{\Lambda}%
_{Nt}-\widetilde{\Lambda}_{Nt}||+||\widetilde{\Lambda}_{Nt}-\Lambda_{Nt}%
^{0}||=o_{p}(1)$.
\end{proof}

\subsection{{\protect\normalsize Proofs of Theorem \ref{THM:consistenOmega}}}

\begin{proof}
Let $S_{[rN]t}=\sum\nolimits_{i=1}^{[rN]}G_{i}^{0}(X_{i}%
)(\tau-I(\varepsilon_{it}<0))$, where $[a]$ denotes the largest integer no
greater than $a$. Let $M=bN$. Define $\Lambda_{Nt}(r)=N^{-1}\sum
\nolimits_{i=1}^{[rN]}p_{i}\left(  0\left\vert X_{i},f_{t}\right.  \right)
G_{i}^{0}(X_{i})G_{i}^{0}(X_{i})^{\intercal}$, $\digamma_{Nt}(r)=N^{-1/2}%
S_{[rN]t}$, and
\[
D_{bN}(r)=N^{2}\left(  K^{\ast}\left(  \frac{[rN]+1}{bN}\right)  -K^{\ast
}\left(  \frac{[rN]}{bN}\right)  \right)  -\left(  K^{\ast}\left(  \frac
{[rN]}{bN}\right)  -K^{\ast}\left(  \frac{[rN]-1}{bN}\right)  \right)  .
\]
Denote $K_{ij}^{\ast}=K^{\ast}(\frac{i-j}{bN})$, and $\widehat{w}_{Nt}%
=\frac{\tau(1-\tau)}{N}\sum\nolimits_{i=1}^{N}\widehat{G}_{i}(X_{i}%
)\widehat{G}_{i}(X_{i})^{\intercal}-N^{-1}\sum\nolimits_{i=1}^{N}%
\widehat{v}_{it}\widehat{v}_{it}^{\intercal}$. Then
\begin{align*}
\widehat{\Omega}_{Nt,N} &  =N^{-1}\sum\nolimits_{i=1}^{N}\sum\nolimits_{j=1}%
^{N}\widehat{v}_{it}K_{ij}^{\ast}\widehat{v}_{jt}^{\intercal}+\widehat{w}%
_{Nt}\\
&  =N^{-1}\sum\nolimits_{i=1}^{N}(\widehat{v}_{it}\sum\nolimits_{j=1}%
^{N}K_{ij}^{\ast}\widehat{v}_{jt}^{\intercal})+\widehat{w}_{Nt}.
\end{align*}
Define $\widehat{S}_{nt}=\sum\nolimits_{i=1}^{n}\widehat{v}_{it}$. By the
assumptions in Theorem \ref{THM:fhat}, $\phi_{NT}N^{1/2}=o(1)$ and by the
results in Lemmas \ref{LEM:fstar}-\ref{LEM:gjtilda}, we have
\begin{equation}
\widehat{f}_{t}-f_{t}^{0}=\Lambda_{Nt}^{-1}\{N^{-1}\sum\nolimits_{i=1}%
^{N}G_{i}^{0}(X_{i})(\tau-I(\varepsilon_{it}<0))\}+o_{p}(N^{-1/2}%
),\label{EQ:fhat-f}%
\end{equation}%
\begin{equation}
\sup_{x_{j}\in\mathcal{X}_{j}}|\widehat{g}_{j}(x_{j})-g_{j}^{0}(x_{j}%
)|=O_{p}(\phi_{NT})+o_{p}(N^{-1/2})=o_{p}(N^{-1/2}).\label{EQ:supg}%
\end{equation}
Let $r\in(0,1]$. Let $\widetilde{S}_{[rN]t}=\sum\nolimits_{i=1}^{[rN]}%
G_{i}^{0}(X_{i})(\tau-I(\widehat{\varepsilon}_{it}^{0}<0))$, where
$\widehat{\varepsilon}_{it}^{0}=y_{it}-\{\widehat{f}_{ut}+\sum\nolimits_{j=1}%
^{J}g_{j}^{0}(X_{ji})\widehat{f}_{jt}\}$. By Lemma \ref{LEM:WNt1}, we have
\begin{equation}
||N^{-1/2}\widehat{S}_{[rN]t}-N^{-1/2}\widetilde{S}_{[rN]t}||=o_{p}%
(1).\label{EQ:Shat-Stilda}%
\end{equation}
For any given $f_{t}\in R^{J+1}$, define $S_{[rN]t}(f_{t})=\sum\nolimits_{i=1}%
^{[rN]}G_{i}^{0}(X_{i})(\tau-I(\varepsilon_{it}(f_{t})<0))$, \ where
$\varepsilon_{it}(f_{t})=y_{it}-\{f_{ut}+\sum\nolimits_{j=1}^{J}g_{j}%
^{0}(X_{ji})f_{jt}\}$. Following similar arguments to the proof in Lemma
\ref{LEM:VNt1}, we have
\[
\sup_{||f_{t}-f_{t}^{0}||\leq C(d_{NT}+N^{-1/2})}||N^{-1/2}[S_{[rN]t}%
(f_{t})-S_{[rN]t}(f_{t}^{0})-E[\{S_{[rN]t}(f_{t})-S_{[rN]t}(f_{t}%
^{0})\}|\mathbf{X,F]]||}=o_{p}(1).
\]
Moreover,
\begin{align}
&  N^{-1/2}E[\{S_{[rN]t}(f_{t})-S_{[rN]t}(f_{t}^{0})\}|\mathbf{X,F]}%
\nonumber\\
&  =\sum\nolimits_{i=1}^{[rN]}G_{i}^{0}(X_{i})E[(I(\varepsilon_{it}(f_{t}%
^{0})<0)-I(\varepsilon_{it}(f_{t})<0))|X_{i},f_{t}],\label{EQ:EShat-Stilda}%
\end{align}
and thus by Taylor's expansion, we have%
\begin{gather}
||N^{-1/2}E[\{S_{[rN]t}(f_{t})-S_{[rN]t}(f_{t}^{0})\}|\mathbf{X,F]}\nonumber\\
-N^{-1/2}\sum\nolimits_{i=1}^{[rN]}p_{i}\left(  0\left\vert X_{i}%
,f_{t}\right.  \right)  G_{i}^{0}(X_{i})G_{i}^{0}(X_{i})^{\intercal}(f_{t}%
^{0}-f_{t})||=o_{p}(1).\label{EQ:EShat-Stilda2}%
\end{gather}
Hence, by (\ref{EQ:Shat-Stilda}), (\ref{EQ:EShat-Stilda}) and
(\ref{EQ:EShat-Stilda2}), we have
\begin{align*}
N^{-1/2}\widehat{S}_{[rN]t} &  =N^{-1/2}\sum\nolimits_{i=1}^{[rN]}G_{i}%
^{0}(X_{i})(\tau-I(\varepsilon_{it}<0))\\
&  -N^{-1/2}\sum\nolimits_{i=1}^{[rN]}p_{i}\left(  0\left\vert X_{i}%
,f_{t}\right.  \right)  G_{i}^{0}(X_{i})G_{i}^{0}(X_{i})^{\intercal
}(\widehat{f}_{t}-f_{t}^{0})+o_{p}(1).
\end{align*}
This result, together with (\ref{EQ:fhat-f}), implies
\begin{equation}
N^{-1/2}\widehat{S}_{[rN]t}=\digamma_{Nt}(r)-\Lambda_{Nt}(r)\{\Lambda
_{Nt}(1)\}^{-1}\digamma_{Nt}(1)+o_{p}(1).\label{EQ:ShatrN}%
\end{equation}
Thus, $N^{-1/2}\widehat{S}_{Nt}=o_{p}(1)$. By following the argument above
again, we have $||N^{-1/2}\sum\nolimits_{j=1}^{N}\widehat{v}_{jt}K_{jN}^{\ast
}-N^{-1/2}\sum\nolimits_{j=1}^{N}v_{jt}K_{jN}^{\ast}||=O_{p}(1)$. Also
$||N^{-1/2}\sum\nolimits_{j=1}^{N}v_{jt}K_{jN}^{\ast}||=O_{p}(1)$ by the weak
law of large numbers. Hence, $||N^{-1/2}\sum\nolimits_{j=1}^{N}\widehat{v}%
_{jt}K_{jN}^{\ast}||=O_{p}(1)$. Therefore
\[
N^{-1}\sum\nolimits_{j=1}^{N}\widehat{v}_{jt}K_{jN}^{\ast}\widehat{S}%
_{N}^{\intercal}=O_{p}(1)o_{p}(1)=o_{p}(1).
\]
By (\ref{EQ:fhat-f}) and (\ref{EQ:supg}), $\widehat{w}_{Nt}=o_{p}(1)$. By this
result and also applying the identity that $\sum\nolimits_{l=1}^{N}a_{l}%
b_{l}=(\sum\nolimits_{l=1}^{N-1}(a_{l}-a_{l+1})\sum\nolimits_{j=1}^{l}%
b_{j})+a_{N}\sum\nolimits_{l=1}^{N}b_{l}$ to $\sum\nolimits_{j=1}^{N}%
K_{ij}^{\ast}\widehat{v}_{j}^{\intercal}$ and then again to the sum over $i$,
we obtain
\begin{align*}
\widehat{\Omega}_{Nt,M=bN} &  =N^{-1}\sum\nolimits_{i=1}^{N-1}N^{-1}%
\sum\nolimits_{j=1}^{N-1}N^{2}((K_{ij}^{\ast}-K_{i,j+1}^{\ast})-(K_{i+1,j}%
^{\ast}-K_{i+1,j+1}^{\ast}))N^{-1/2}\widehat{S}_{it}N^{-1/2}\widehat{S}%
_{jt}^{\intercal}\\
&  +N^{-1}\sum\nolimits_{j=1}^{N}\widehat{v}_{jt}K_{jN}^{\ast}\widehat{S}%
_{Nt}^{\intercal}+o_{p}(1),
\end{align*}
and thus
\begin{align}
&  \widehat{\Omega}_{Nt,M=bN}=\sum\nolimits_{i=1}^{N-1}\sum\nolimits_{j=1}%
^{N-1}((K_{ij}^{\ast}-K_{i,j+1}^{\ast})-(K_{i+1,j}^{\ast}-K_{i+1,j+1}^{\ast
}))\frac{\widehat{S}_{it}}{\sqrt{N}}\,\frac{\widehat{S}_{jt}^{\intercal}%
}{\sqrt{N}}\nonumber\\
&  +o_{p}(1).\label{EQ:Omegahat}%
\end{align}
Moreover,
\begin{equation}
N^{2}((K_{ij}^{\ast}-K_{i,j+1}^{\ast})-(K_{i+1,j}^{\ast}-K_{i+1,j+1}^{\ast
}))=-D_{bN}\{(i-j)/N\}.\label{EQ:Ddiff}%
\end{equation}
Also $\lim_{N\rightarrow\infty}D_{bN}(r)=\frac{1}{b^{2}}K^{\ast\prime\prime
}(\frac{r}{b})$,
$\vert$%
$\vert$%
$\Lambda_{Nt}(r)-r\Lambda_{t}^{0}||=o_{p}(1)$, where $\Lambda_{t}^{0}%
=\lim_{N\rightarrow\infty}\Lambda_{Nt}^{0}$ and $\digamma_{Nt}%
(r)\overset{\mathcal{D}}{\rightarrow}W_{J+1}(r)\Upsilon^{\intercal}$. Thus,
\begin{equation}
(\Lambda_{Nt}(r),\digamma_{Nt}(r)^{\intercal},D_{bN}(r))\overset{\mathcal{D}%
}{\rightarrow}\left(  r\Lambda_{t}^{0},\Upsilon W_{J+1}(r)^{\intercal}%
,\frac{1}{b^{2}}K^{\ast\prime\prime}\left(  \frac{r}{b}\right)  \right)
.\label{EQ:distr}%
\end{equation}
Hence, by (\ref{EQ:ShatrN}), (\ref{EQ:Omegahat}), and (\ref{EQ:Ddiff}), it
follows that
\begin{align}
\widehat{\Omega}_{Nt,M=bN} &  =\int_{0}^{1}\int_{0}^{1}-D_{bN}(r-s)[\digamma
_{Nt}(r)-\Lambda_{Nt}(r)\{\Lambda_{Nt}(1)\}^{-1}\digamma_{Nt}(1)]\nonumber\\
&  \times\lbrack\digamma_{Nt}(s)-\Lambda_{Nt}(s)\{\Lambda_{Nt}(1)\}^{-1}%
\digamma_{Nt}(1)]^{\intercal}drds+o_{p}(1).\label{EQ:OmegahatNtN}%
\end{align}
By the continuous mapping theorem,%
\[
\widehat{\Omega}_{N,M=bN}\overset{\mathcal{D}}{\rightarrow}\Upsilon\int%
_{0}^{1}\int_{0}^{1}-\frac{1}{b^{2}}K^{\ast\prime\prime}(\frac{r-s}%
{b})\{W_{J+1}(r)-rW_{J+1}(1)\}\{W_{J+1}(s)-sW_{J+1}(1)\}^{\intercal
}drds\Upsilon^{\intercal}.
\]
Then the proof is completed.
\end{proof}

\subsection{{\protect\normalsize Proofs of Theorems \ref{THM:null} and
\ref{THM:alternative}}}

\begin{proof}
By (\ref{EQ:fhat-f}), $\widehat{f}_{t}-f_{t}^{0}=N^{-1/2}\Lambda
_{Nt}(1)^{-1}\digamma_{Nt}(1)+o_{p}(N^{-1/2})$. Then under $H_{0}$, we have
\begin{equation}
N^{1/2}(R\widehat{f}_{t}-r)=R\Lambda_{Nt}(1)^{-1}\digamma_{Nt}(1)+o_{p}(1).
\label{EQ:Rfhat}%
\end{equation}
It directly follows from (\ref{EQ:distr}), (\ref{EQ:OmegahatNtN}) and
(\ref{EQ:Rfhat}) that
\begin{align*}
&  F_{Nt,b}\overset{\mathcal{D}}{\rightarrow}\{R\Lambda_{t}^{0-1}\Upsilon
W_{J+1}(1)\}^{\intercal}\{R\tau(1-\tau)\Lambda_{t}^{0-1}\\
&  \times(\Upsilon\int\nolimits_{0}^{1}\int\nolimits_{0}^{1}-\frac{1}{b^{2}%
}K^{\ast\prime\prime}(\frac{r-s}{b})B_{J+1}(r)B_{J+1}(s)^{\intercal
}drds\Upsilon^{\intercal})\Lambda_{t}^{0-1}R^{\intercal}\}^{-1}\\
&  \times R\Lambda_{t}^{0-1}\Upsilon W_{J+1}(1)/q.
\end{align*}
Since $R\Lambda_{t}^{0-1}\Upsilon W_{J+1}(1)$ is a $q\times1$ vector of normal
random variables with mean zero and variance $R\Lambda_{t}^{0-1}%
\Upsilon \Upsilon^{\intercal}\Lambda_{t}^{0-1}R^{\intercal}$, $R\Lambda
_{t}^{0-1}\Upsilon W_{J+1}(1)$ can be written as $\Upsilon_{t}^{\ast}W_{q}%
(1)$, where $\Upsilon_{t}^{\ast}\Upsilon_{t}^{\ast\intercal}=R\Lambda
_{t}^{0-1}\Upsilon \Upsilon^{\intercal}\Lambda_{t}^{0-1}R^{\intercal}$. Then
replacing $R\Lambda_{t}^{0-1}\Upsilon W_{J+1}(1)$ by $\Upsilon_{t}^{\ast}%
W_{q}(1)$ and canceling $\Upsilon_{t}^{\ast}$ in the above equation, we have
the result in Theorem \ref{THM:null}. Moreover, under the alternative that
$H_{1}$: $Rf_{t}^{0}=r+cN^{-1/2}$, we have%
\begin{align*}
N^{1/2}(R\widehat{f}_{t}-r)  &  =N^{1/2}(Rf_{t}^{0}-r)+R\Lambda_{Nt}%
(1)^{-1}\digamma_{Nt}(1)+o_{p}(1)\\
&  =c+R\Lambda_{Nt}(1)^{-1}\digamma_{Nt}(1)+o_{p}(1).
\end{align*}
Thus by (\ref{EQ:distr}), we have
\begin{align*}
&  F_{Nt,b}\overset{\mathcal{D}}{\rightarrow}\{c+R\Lambda_{t}^{0-1}\Upsilon
W_{J+1}(1)\}^{\intercal}\{R\tau(1-\tau)\Lambda_{t}^{0-1}\\
&  \times(\Upsilon\int\nolimits_{0}^{1}\int\nolimits_{0}^{1}-\frac{1}{b^{2}%
}K^{\ast\prime\prime}(\frac{r-s}{b})B_{J+1}(r)B_{J+1}(s)^{\intercal
}drds\Upsilon^{\intercal})\Lambda_{t}^{0-1}R^{\intercal}\}^{-1}\\
&  \times\{c+R\Lambda_{t}^{0-1}\Upsilon W_{J+1}(1)\}/q.
\end{align*}
Also $c+R\Lambda_{t}^{0-1}\Upsilon W_{J+1}(1)\equiv c+\Upsilon_{t}^{\ast}%
W_{q}(1)=\Upsilon_{t}^{\ast}(\Upsilon_{t}^{\ast-1}c+W_{q}(1))$. Then the
result in Theorem \ref{THM:alternative} follows from the above results. The
proof is completed.
\end{proof}

\section{Supplementary Material}

In this supplement, we present Lemmas \ref{LEM:Gtildatheta-theta0}-\ref{LEM:Gstar} which are used to prove Lemma \ref{LEM:bahardur} in Section \ref{SEC:ghat0}. We also give Lemmas \ref{LEM:LNtfg}-\ref{LEM:WNt1} which are used in the proofs of Lemmas \ref{LEM:fstar} and \ref{LEM:fhat0}, and Lemmas \ref{LEM:VNt11}-\ref{LEM:VNt1} which are used in the proofs of Lemma \ref{LEM:gjtilda} in Section \ref{SEC:fhatghat}.



\renewcommand{\thetheorem}{{\sc S.\arabic{theorem}}} \renewcommand{%
\theproposition}{{\sc S.\arabic{proposition}}}
\renewcommand{\thelemma}{{\sc
S.\arabic{lemma}}} \renewcommand{\thecorollary}{{\sc
S.\arabic{corollary}}} \renewcommand{\theequation}{S.\arabic{equation}} %
\renewcommand{\thesubsection}{{\it S.\arabic{subsection}}} %
\setcounter{equation}{0} \setcounter{lemma}{0} \setcounter{proposition}{0} %
\setcounter{theorem}{0} \setcounter{subsection}{0}\setcounter{corollary}{0}

\begin{lemma}
{\normalsize \label{LEM:Gtildatheta-theta0}Under Conditions (C1) and (C2), and
$K_{N}^{2}N^{-1}(\log NT)^{2}(\log N)^{8}=o(1)$ and $K_{N}^{-1}=o(1)$,
\begin{align*}
&  \sup_{1\leq t\leq T}\sup_{||\boldsymbol{\vartheta}_{t}%
-\boldsymbol{\vartheta}_{t}^{0}||\leq CK_{N}^{1/2}N^{-1/2}}||N^{-1}%
\sum\nolimits_{i=1}^{N}\widetilde{G}_{tN,i}(\boldsymbol{\vartheta}_{t}%
)-N^{-1}\sum\nolimits_{i=1}^{N}\widetilde{G}_{tN,i}(\boldsymbol{\vartheta}%
_{t}^{0})||\\
&  =O_{p}(K_{N}^{3/2}N^{-3/4}\sqrt{\log NT}).
\end{align*}
}
\end{lemma}

\begin{proof}
Let $B_{N}=\{\boldsymbol{\vartheta}_{t}:||\boldsymbol{\vartheta}%
_{t}-\boldsymbol{\vartheta}_{t}^{0}||\leq CK_{N}^{1/2}N^{-1/2}\}$. By taking
the same strategy as given in Lemma A.5 of Horowitz and Lee (2005), we cover
the ball $B_{N}$ with cubes $\mathcal{C=\{C(}\boldsymbol{\vartheta}_{t,v})\}$,
where $\mathcal{C(}\boldsymbol{\vartheta}_{t,v})$ is a cube containing
$(\boldsymbol{\vartheta}_{t,v}-\boldsymbol{\vartheta}_{t}^{0})$ with sides of
$C\{d(N)/N^{5}\}^{1/2}$ such that $\boldsymbol{\vartheta}_{t,v}\in B_{N}$.
Then the number of the cubes covering the ball $B_{N}$ is $V=(2N^{2})^{d(N)}$.
Moreover, we have $||(\boldsymbol{\vartheta}_{t}-\boldsymbol{\vartheta}%
_{t}^{0})-(\boldsymbol{\vartheta}_{t,v}-\boldsymbol{\vartheta}_{t}^{0})||\leq
C\{d(N)/N^{5/2}\}$ for any $\boldsymbol{\vartheta}_{t}-\boldsymbol{\vartheta
}_{t}^{0}\in\mathcal{C(}\boldsymbol{\vartheta}_{t,v})$, where $v=1,\ldots,V$.
First we can decompose
\begin{align}
&  \sup_{\boldsymbol{\vartheta}_{t}\in B_{N}}||N^{-1}\sum\nolimits_{i=1}%
^{N}\widetilde{G}_{tN,i}(\boldsymbol{\vartheta}_{t})-N^{-1}\sum\nolimits_{i=1}%
^{N}\widetilde{G}_{tN,i}(\boldsymbol{\vartheta}_{t}^{0})||\nonumber\\
&  \leq\max_{1\leq v\leq V}\sup_{(\boldsymbol{\vartheta}_{t}%
-\boldsymbol{\vartheta}_{t}^{0})\in\mathcal{C(}\boldsymbol{\vartheta}_{t,v}%
)}||N^{-1}\sum\nolimits_{i=1}^{N}\widetilde{G}_{tN,i}(\boldsymbol{\vartheta
}_{t})-N^{-1}\sum\nolimits_{i=1}^{N}\widetilde{G}_{tN,i}(\boldsymbol{\vartheta
}_{t,v})||\nonumber\\
&  +\max_{1\leq v\leq V}||N^{-1}\sum\nolimits_{i=1}^{N}\widetilde{G}%
_{tN,i}(\boldsymbol{\vartheta}_{t,v})-N^{-1}\sum\nolimits_{i=1}^{N}%
\widetilde{G}_{tN,i}(\boldsymbol{\vartheta}_{t}^{0})||\nonumber\\
&  =\Delta_{tN,1}+\Delta_{tN,2}\label{EQ:diff}%
\end{align}
Let $\gamma_{N}=C\{d(N)/n^{5/2}\}$. By the same argument as given in the proof
of Lemma A.5 in Horowitz and Lee (2005), we have
\begin{equation}
\Delta_{tN,1}\leq\max_{1\leq v\leq V}|\Gamma_{tN,1v}|+\max_{1\leq v\leq
V}|\Gamma_{tN,2v}|,\label{EQ:DetlaN1}%
\end{equation}
where
\begin{align*}
\Gamma_{tN,1v} &  =N^{-1}\sum\nolimits_{i=1}^{N}||Z_{i}||\left[  F_{i}%
[Z_{i}^{\intercal}(\boldsymbol{\vartheta}_{t,v}-\boldsymbol{\vartheta}_{t}%
^{0})-b_{t}(X_{i})+||Z_{i}||\gamma_{N}|X_{i},f_{t}]\right.  \\
&  \left.  -F_{i}[Z_{i}^{\intercal}(\boldsymbol{\vartheta}_{t,v}%
-\boldsymbol{\vartheta}_{t}^{0})-b_{t}(X_{i})-||Z_{i}||\gamma_{N}|X_{i}%
,f_{t}]\right]  ,\\
\Gamma_{tN,2v} &  =N^{-1}\sum\nolimits_{i=1}^{N}\Gamma_{tN,2v,i}=N^{-1}%
\sum\nolimits_{i=1}^{N}||Z_{i}||\left[  [I\{\varepsilon_{it}\leq
Z_{i}^{\intercal}(\boldsymbol{\vartheta}_{t,v}-\boldsymbol{\vartheta}_{t}%
^{0})-b_{t}(X_{i})+||Z_{i}||\gamma_{N}\}\right.  \\
&  -F_{i}\{Z_{i}^{\intercal}(\boldsymbol{\vartheta}_{t,v}%
-\boldsymbol{\vartheta}_{t}^{0})-b_{t}(X_{i})+||Z_{i}||\gamma_{N}|X_{i}%
,f_{t}\}]\\
&  \left.  -[I\{\varepsilon_{it}\leq Z_{i}^{\intercal}(\boldsymbol{\vartheta
}_{t,v}-\boldsymbol{\vartheta}_{t}^{0})-b_{t}(X_{i})\}-F_{i}\{Z_{i}%
^{\intercal}(\boldsymbol{\vartheta}_{t,v}-\boldsymbol{\vartheta}_{t}%
^{0})-b_{t}(X_{i})|X_{i},f_{t}\}]\right]  .
\end{align*}
By Condition (C2), we have that there are some constants $0<c^{\prime
},c^{\prime\prime}<\infty$ such that
\begin{equation}
\sup_{1\leq t\leq T}\max_{1\leq v\leq V}|\Gamma_{tN,1v}|\leq c^{\prime}%
\gamma_{N}\max_{1\leq i\leq N}||Z_{i}||||Z_{i}||\leq c^{\prime\prime
}\{d(N)/N^{5/2}\}K_{N}=O(K_{N}^{2}N^{-5/2}).\label{EQ:GammaN1v}%
\end{equation}
Next we will show the convergence rate for $\max_{1\leq v\leq V}%
|\Gamma_{tN,2v}|$. It is easy to see that $E(\Gamma_{tN,2v,i})=0$. Also
$|\Gamma_{tN,2v,i}|\leq4||Z_{i}||\leq c_{1}K_{N}^{1/2}$ for some constant
$0<c_{1}<\infty$. Moreover,
\begin{align*}
&  E\left[  ||Z_{i}||I\{\varepsilon_{it}\leq Z_{i}^{\intercal}%
(\boldsymbol{\vartheta}_{t,v}-\boldsymbol{\vartheta}_{t}^{0})-b_{t}%
(X_{i})+||Z_{i}||\gamma_{N}\}-I\{\varepsilon_{it}\leq Z_{i}^{\intercal
}(\boldsymbol{\vartheta}_{t,v}-\boldsymbol{\vartheta}_{t}^{0})-b_{t}%
(X_{i})\}\right]  ^{2}\\
&  \asymp E\{||Z_{i}||^{2}||Z_{i}||\gamma_{N}\}\leq c_{2}^{\ast}\gamma
_{N}K_{N}^{1/2}\leq c_{2}K_{N}^{3/2}N^{-5/2},
\end{align*}
for some constants $0<c_{2}^{\ast}<c_{2}<\infty$. Hence $E(\Gamma
_{tN,2v,i})^{2}\leq c_{2}K_{N}^{3/2}N^{-5/2}$. By Condition (C1), we have for
$i\neq j$,
\[
|E(\Gamma_{tN,2v,i}\Gamma_{tN,2v,j})|\leq2\phi(|j-i|)^{1/2}\{E(\Gamma
_{tN,2v,i}^{2})E(\Gamma_{tN,2v,j}^{2})\}^{1/2}\leq2c_{2}\phi(|j-i|)^{1/2}K_{N}%
^{3/2}N^{-5/2}.
\]
Hence
\begin{align*}
&  E(\Gamma_{tN,2v,i})^{2}+2\sum\nolimits_{j>i}|E(\Gamma_{tN,2v,i}%
\Gamma_{tN,2v,j})|\\
&  \leq c_{2}K_{N}^{3/2}N^{-5/2}+4c_{2}\sum\nolimits_{k=1}^{N}K_{1}%
e^{-\lambda_{1}k/2}K_{N}^{3/2}N^{-5/2}\\
&  \leq c_{2}K_{N}^{3/2}N^{-5/2}(1+4K_{1}(1-e^{-\lambda_{1}/2})^{-1}%
)=c_{3}K_{N}^{3/2}N^{-5/2},
\end{align*}
where $c_{3}=c_{2}(1+4K_{1}(1-e^{-\lambda_{1}/2})^{-1})$. By Condition (C1),
for each fixed $t$, the sequence $\{(X_{i},f_{t},\varepsilon_{it}),1\leq i\leq
N\}$ has the $\phi$-mixing coefficient $\phi(k)\leq K_{1}e^{-\lambda_{1}k}$
for $K_{1},\lambda_{1}>0$. Thus, by the Bernstein's inequality given in Lemma
\ref{LEM:Bernstein}, we have for $N$ sufficiently large,
\begin{align*}
&  P\left(  |\Gamma_{tN,2v}|\geq aK_{N}^{3/2}N^{-1}(\log NT)^{3}\right)  \\
&  \leq\exp(-\frac{C_{1}(aK_{N}^{3/2}(\log NT)^{3})^{2}}{c_{3}K_{N}%
^{3/2}N^{-5/2}N+c_{1}^{2}K_{N}+aK_{N}^{3/2}(\log NT)^{3}c_{1}K_{N}^{1/2}%
\log(N)^{2}})\leq(NT)^{-c_{4}a^{2}K_{N}}%
\end{align*}
for some constant $0<c_{4}<\infty$. By the union bound of probability, we
have
\begin{align*}
&  P\left(  \sup_{1\leq t\leq T}\max_{1\leq v\leq V}|\Gamma_{tN,2v}|\geq
aK_{N}^{3/2}N^{-1}(\log NT)^{3}\right)  \\
&  \leq(2N^{2})^{d(N)}T(NT)^{-c_{4}a^{2}K_{N}}\leq2^{d(N)}N^{2(1+JK_{N}%
)-c_{4}a^{2}K_{N}}T^{1-c_{4}a^{2}K_{N}}.
\end{align*}
Hence, taking $a$ large enough, one has
\[
P\left(  \sup_{1\leq t\leq T}\max_{1\leq v\leq V}|\Gamma_{tN,2v}|\geq
aK_{N}^{3/2}N^{-1}(\log N)^{3}\right)  \leq2^{K_{N}}N^{-K_{N}}T^{-K_{N}}.
\]
Then we have
\begin{equation}
\sup_{1\leq t\leq T}\max_{1\leq v\leq V}|\Gamma_{tN,2v}|=O_{p}\{K_{N}%
^{3/2}N^{-1}(\log NT)^{3}\}.\label{EQ:GammaN2v}%
\end{equation}
Next we will show the convergence rate for $\Delta_{tN,2}$. Let $\widetilde{g}%
_{tN,i,\ell}(\boldsymbol{\vartheta}_{t,v})$ be the $\ell^{\text{th}}$ element
in $\widetilde{G}_{tN,i}(\boldsymbol{\vartheta}_{t,v})-\widetilde{G}%
_{tN,i}(\boldsymbol{\vartheta}_{t}^{0})$ for $\ell=1,\ldots,d(N)$. It is easy
to see that $E\{\widetilde{g}_{tN,i,\ell}(\boldsymbol{\vartheta}_{t,v})\}=0$.
Also $|\widetilde{g}_{tN,i,\ell}(\boldsymbol{\vartheta}_{t,v})|\leq4|Z_{i\ell
}|\leq c_{1}K_{N}^{1/2}$ for some constant $0<c_{1}<\infty$. Moreover,%
\begin{align*}
&  E\left[  [I\{\varepsilon_{it}\leq Z_{i}^{\intercal}(\boldsymbol{\vartheta
}_{t,v}-\boldsymbol{\vartheta}_{t}^{0})-b_{t}(X_{i})\}-I\{\varepsilon_{it}%
\leq-b_{t}(X_{i})\}]Z_{i\ell}\right]  ^{2}\\
&  \leq c_{1}^{\prime}||\boldsymbol{\vartheta}_{t,v}-\boldsymbol{\vartheta
}_{t}^{0}||K_{N}^{1/2}\leq c_{1}^{\prime}CK_{N}^{1/2}N^{-1/2}K_{N}^{1/2}%
=c_{1}^{\prime}CK_{N}N^{-1/2}%
\end{align*}
for some constant $0<c_{1}^{\prime}<\infty$. Hence $E(\widetilde{g}%
_{tN,i,\ell}(\boldsymbol{\vartheta}_{t,v}))^{2}\leq c_{1}^{\prime}%
CK_{N}N^{-1/2}$. By Condition (C1), we have for $i\neq j$,
\[
|E(\widetilde{g}_{tN,i,\ell}(\boldsymbol{\vartheta}_{t,v})\widetilde{g}%
_{tN,j,\ell}(\boldsymbol{\vartheta}_{t,v})|\leq4\phi(|j-i|)^{1/2}%
\{E(\Gamma_{tN,2v,i}^{2})E(\Gamma_{tN,2v,j}^{2})\}^{1/2}.
\]
Hence
\begin{align*}
&  E(\widetilde{g}_{tN,i,\ell}(\boldsymbol{\vartheta}_{t,v}))^{2}%
+2\sum\nolimits_{j>i}|E(\widetilde{g}_{tN,i,\ell}(\boldsymbol{\vartheta}%
_{t,v})\widetilde{g}_{tN,j,\ell}(\boldsymbol{\vartheta}_{t,v})|\\
&  \leq c_{1}^{\prime}CK_{N}N^{-1/2}+4\sum\nolimits_{k=1}^{N}K_{1}%
e^{-\lambda_{1}k/2}c_{1}^{\prime}CK_{N}N^{-1/2}\\
&  \leq c_{1}^{\prime}CK_{N}N^{-1/2}(1+4K_{1}(1-e^{-\lambda_{1}/2}%
)^{-1})=c_{2}K_{N}N^{-1/2},
\end{align*}
where $c_{2}=c_{1}^{\prime}C(1+4K_{1}(1-e^{-\lambda_{1}/2})^{-1})$. Thus, by
the Bernstein's inequality given in Lemma \ref{LEM:Bernstein} and $K_{N}%
^{2}N^{-1}(\log NT)^{2}(\log N)^{8}=o(1)$, we have for $N$ sufficiently
large,
\begin{align*}
&  P\left(  |N^{-1}\sum\nolimits_{i=1}^{N}\widetilde{g}_{tN,i,\ell
}(\boldsymbol{\vartheta}_{t,v})|\geq aK_{N}N^{-3/4}\sqrt{\log NT}\right)  \\
&  \leq\exp(-\frac{C_{1}(aK_{N}N^{1/4}\sqrt{\log NT})^{2}}{c_{2}K_{N}%
N^{-1/2}N+c_{1}^{2}K_{N}+aK_{N}N^{1/4}(\log NT)^{1/2}c_{1}K_{N}^{1/2}(\log
N)^{2}})\leq(NT)^{-c_{3}a^{2}K_{N}}%
\end{align*}
for some constant $0<c_{3}<\infty$. By the union bound of probability, we have%
\[
P\left(  \sup_{1\leq t\leq T}\sup_{1\leq\ell\leq d\left(  N\right)  }%
|N^{-1}\sum\nolimits_{i=1}^{N}\widetilde{g}_{tN,i,\ell}(\boldsymbol{\vartheta
}_{t,v})|\geq aK_{N}N^{-3/4}\sqrt{\log NT}\right)  \leq d(N)T(NT)^{-c_{3}%
a^{2}K_{N}}.
\]
Hence,
\begin{align*}
&  P\left(  \sup_{1\leq t\leq T}||N^{-1}\sum\nolimits_{i=1}^{N}\widetilde{G}%
_{tN,i}(\boldsymbol{\vartheta}_{t,v})-N^{-1}\sum\nolimits_{i=1}^{N}%
\widetilde{G}_{tN,i}(\boldsymbol{\vartheta}_{t}^{0})||\geq aK_{N}%
^{3/2}N^{-3/4}\sqrt{\log NT}\right)  \\
&  \leq d(N)T(NT)^{-c_{3}a^{2}K_{N}}.
\end{align*}
By the union bound of probability again, we have
\[
P\left(  \sup_{1\leq t\leq T}|\Delta_{tN,2}|\geq aK_{N}^{3/2}N^{-3/4}%
\sqrt{\log NT}\right)  \leq(2N^{2})^{d(N)}d(N)T(NT)^{-c_{3}a^{2}K_{N}}.
\]
Hence, taking $a$ large enough, one has
\[
P\left(  \sup_{1\leq t\leq T}|\Delta_{tN,2}|\geq aK_{N}^{3/2}N^{-3/4}%
\sqrt{\log NT}\right)  \leq2^{K_{N}}K_{N}N^{-K_{N}}T^{-K_{N}}.
\]
Then we have
\begin{equation}
\sup_{1\leq t\leq T}|\Delta_{tN,2}|=O_{p}\{K_{N}^{3/2}N^{-3/4}\sqrt{\log
NT}\}.\label{EQ:delta2}%
\end{equation}
Therefore, by (\ref{EQ:diff}), (\ref{EQ:DetlaN1}), (\ref{EQ:GammaN1v}),
(\ref{EQ:GammaN2v}) and (\ref{EQ:delta2}), we have%
\begin{align*}
\sup_{1\leq t\leq T} &  \sup_{\boldsymbol{\vartheta}_{t}\in B_{N}}||N^{-1}%
\sum\nolimits_{i=1}^{N}\widetilde{G}_{tN,i}(\boldsymbol{\vartheta}_{t}%
)-N^{-1}\sum\nolimits_{i=1}^{N}\widetilde{G}_{tN,i}(\boldsymbol{\vartheta}%
_{t}^{0})||\\
&  =O_{p}\{K_{N}^{2}N^{-5/2}+K_{N}^{3/2}N^{-1}(\log NT)^{3}+K_{N}%
^{3/2}N^{-3/4}\sqrt{\log NT}\}\\
&  =O_{p}(K_{N}^{3/2}N^{-3/4}\sqrt{\log NT}).
\end{align*}
\end{proof}

\begin{lemma}
{\normalsize \label{LEM:Gtildatilda}Under Conditions (C1) and (C2), $\sup_{1\leq t\leq T}||N^{-1}\sum
\nolimits_{i=1}^{N}G_{tN,i}(\widetilde{\boldsymbol{\vartheta}}_{t}%
)||=O_{p}(K_{N}^{3/2}N^{-1})$. }
\end{lemma}

\begin{lemma}
\label{LEM:Gstar}Under Conditions (C2) and (C3), as $N\rightarrow
\infty$,%
\[
\Psi_{Nt}^{-1}G_{tN,i}^{\ast}(\boldsymbol{\vartheta}_{t}%
)=-(\boldsymbol{\vartheta}_{t}-\boldsymbol{\vartheta}_{t}^{0})+N^{-1}\Psi
_{Nt}^{-1}\sum\nolimits_{i=1}^{N}p_{i}\left(  0\left\vert X_{i},f_{t}\right.
\right)  Z_{i}b_{t}(X_{i})+R_{Nt}^{\ast},
\]
where $||R_{Nt}^{\ast}||\leq C^{\ast}\{K_{N}^{1/2}||\boldsymbol{\vartheta}%
_{t}-\boldsymbol{\vartheta}_{t}^{0}||^{2}+K_{N}^{1/2-2r}\}$ for some constant
$0<C^{\ast}<\infty$, uniformly in $t$.
\end{lemma}

\begin{proof}
The proofs of Lemmas \ref{LEM:Gtildatilda} and \ref{LEM:Gstar} follow
the same procedure as in Lemmas A.4 and A.7 of Horowitz and Lee (2005) by
using the results (\ref{EQ:Rjt}) and (\ref{EQ:ZZ}) which hold uniformly in
$t=1,...,T$.
\end{proof}

\begin{lemma}
\label{LEM:LNtfg} Under Conditions (C2) and (C3),
\[
E\{L_{Nt}(f_{t},g)|\mathbf{X,F}\}=-(f_{t}-f_{t}^{0})^{\intercal}%
E(W_{Nt,1}|\mathbf{X,F)+}\frac{1}{2}(f_{t}-f_{t}^{0})^{\intercal}\Lambda
_{Nt}^{0}(f_{t}-f_{t}^{0})+o_{p}(||f_{t}-f_{t}^{0}||^{2}),
\]
uniformly in $||\boldsymbol{\lambda}_{j}-\boldsymbol{\lambda}_{j}^{0}%
||\leq\widetilde{C}$$d_{NT}^{\ast}$ and $||f_{t}-f_{t}^{0}||\leq
\varpi_{N}$, where $W_{Nt,1}$ is defined in \ref{EQ:WNt1ap} and $g_{j}(x_{j})=B_{j}(x_{j})^{\intercal}\boldsymbol{\lambda
}_{j}$.
\end{lemma}

\begin{proof}
By using the identity of Knight (1998) that
\[
\rho_{\tau}(u-v)-\rho_{\tau}(u)=-v\psi_{\tau}(u)+\int\nolimits_{0}^{v}(I(u\leq
s)-I(u\leq0))ds,
\]
we have%
\begin{align}
&  \rho_{\tau}(y_{it}-f_{t}^{\intercal}G_{i}(X_{i}))-\rho_{\tau}(y_{it}%
-f_{t}^{0\intercal}G_{i}(X_{i}))\nonumber\\
&  =-(f_{t}-f_{t}^{0})^{\intercal}G_{i}(X_{i})\psi_{\tau}(y_{it}%
-f_{t}^{0\intercal}G_{i}(X_{i}))\nonumber\\
&  +\int\nolimits_{0}^{(f_{t}-f_{t}^{0})^{\intercal}G_{i}(X_{i})}\left(
I(y_{it}-f_{t}^{0\intercal}G_{i}(X_{i})\leq s)-I(y_{it}-f_{t}^{0\intercal
}G_{i}(X_{i})\leq0)\right)  ds.\label{EQ:rhotau}%
\end{align}
By Lipschitz continuity of $p_{i}(\varepsilon|X_{i},f_{t})$ given in Condition
(C2) and boundedness of $f_{jt}^{0}$ in Condition (C3), we have%
\begin{align*}
&  F_{i}\{f_{t}^{0\intercal}(G_{i}(X_{i})-G_{i}^{0}(X_{i}))+s|X_{i}%
,f_{t}\}-F_{i}\{f_{t}^{0\intercal}(G_{i}(X_{i})-G_{i}^{0}(X_{i}))|X_{i}%
,f_{t}\}\\
&  =sp_{i}\{f_{t}^{0\intercal}(G_{i}(X_{i})-G_{i}^{0}(X_{i}))|X_{i}%
,f_{t}\}+o(s),
\end{align*}
where $o(\cdot)$ holds uniformly in $||\boldsymbol{\lambda}_{j}%
-\boldsymbol{\lambda}_{j}^{0}||\leq\widetilde{C}$$d_{NT}^{\ast}$ and
$||f_{t}-f_{t}^{0}||\leq\varpi_{N}$. Then we have
\begin{align}
&  E\{L_{Nt}(f_{t},g)|\mathbf{X,F}\}\nonumber\\
&  =-(f_{t}-f_{t}^{0})^{\intercal}E(W_{Nt,1}|\mathbf{X,F)}+N^{-1}%
\sum\nolimits_{i=1}^{N}\int\nolimits_{0}^{(f_{t}-f_{t}^{0})^{\intercal}%
G_{i}(X_{i})}[F_{i}\{f_{t}^{0\intercal}(G_{i}(X_{i})-G_{i}^{0}(X_{i}%
))+s|X_{i},f_{t}\}\nonumber\\
&  -F_{i}\{f_{t}^{0\intercal}(G_{i}(X_{i})-G_{i}^{0}(X_{i}))|X_{i}%
,f_{t}\}]ds\nonumber\\
&  =-(f_{t}-f_{t}^{0})^{\intercal}E(W_{Nt,1}|\mathbf{X,F)}+N^{-1}%
\sum\nolimits_{i=1}^{N}\int\nolimits_{0}^{(f_{t}-f_{t}^{0})^{\intercal}%
G_{i}(X_{i})}[sp_{i}\{f_{t}^{0\intercal}(G_{i}(X_{i})-G_{i}^{0}(X_{i}%
))|X_{i},f_{t}\}]ds\nonumber\\
&  +o\left[  (f_{t}-f_{t}^{0})^{\intercal}\{N^{-1}\sum\nolimits_{i=1}^{N}%
G_{i}(X_{i})G_{i}(X_{i})^{\intercal}\}(f_{t}-f_{t}^{0})\right]  \nonumber\\
&  =-(f_{t}-f_{t}^{0})^{\intercal}E(W_{Nt,1}|\mathbf{X,F)}+\frac{1}{2}%
(f_{t}-f_{t}^{0})^{\intercal}\times\nonumber\\
&  \left[  N^{-1}\sum\nolimits_{i=1}^{N}p_{i}\{f_{t}^{0\intercal}(G_{i}%
(X_{i})-G_{i}^{0}(X_{i}))|X_{i},f_{t}\}G_{i}(X_{i})G_{i}(X_{i})^{\intercal
}\right]  (f_{t}-f_{t}^{0})\nonumber\\
&  +o\left[  (f_{t}-f_{t}^{0})^{\intercal}\{N^{-1}\sum\nolimits_{i=1}^{N}%
G_{i}(X_{i})G_{i}(X_{i})^{\intercal}\}(f_{t}-f_{t}^{0})\right]
.\label{EQ:ELNtftg}%
\end{align}
Since $\sup\nolimits_{x_{j}\in\lbrack a,b]}|g_{j}(x_{j})-g_{j}^{0}%
(x_{j})|=o(1)$, then $\sup_{x\in\mathcal{X}}|f_{t}^{0\intercal}(G_{i}%
(x)-G_{i}^{0}(x))|=o(1)$. By similar reasoning to the proof for Theorem 2 in
Lee and Robinson (2016), we have \newline $N^{-1}\sum\nolimits_{i=1}^{N}G_{i}%
(X_{i})G_{i}(X_{i})^{\intercal}=N^{-1}\sum\nolimits_{i=1}^{N}E\{G_{i}(X_{i})G_{i}(X_{i})^{\intercal
}\}+o_{p}(1)$. Hence, by these results, we have the result
in Lemma \ref{LEM:LNtfg}.
\end{proof}

\begin{lemma}
 \label{LEM:WNt2} Under Conditions (C2) and (C3), we have
\[
W_{Nt,2}(f_{t},g)-E(W_{Nt,2}(f_{t},g)|\mathbf{X,F)=}o_{p}(||f_{t}-f_{t}%
^{0}||^{2}+N^{-1})
\]
uniformly in $||\boldsymbol{\lambda}_{j}-\boldsymbol{\lambda}_{j}^{0}%
||\leq\widetilde{C}$$d_{NT}^{\ast}$ and $||f_{t}-f_{t}^{0}||\leq
\varpi_{N}$, where $W_{Nt,2}(f_{t},g)$ is defined in (\ref{EQ:WNt2}) and
$g_{j}(x_{j})=B_{j}(x_{j})^{\intercal}\boldsymbol{\lambda}_{j}$.
\end{lemma}

\begin{proof}
By (\ref{EQ:rhotau}), we have%
\[
W_{Nt,2i}(f_{t},g)=\int\nolimits_{0}^{(f_{t}-f_{t}^{0})^{\intercal}G_{i}%
(X_{i})}\left(  I(y_{it}-f_{t}^{0\intercal}G_{i}(X_{i})\leq s)-I(y_{it}%
-f_{t}^{0\intercal}G_{i}(X_{i})\leq0)\right)  ds,
\]
and thus
\begin{align*}
E(W_{Nt,2i}(f_{t},g)|X_{i},f_{t}\mathbf{)} &  \mathbf{=}\int\nolimits_{0}%
^{(f_{t}-f_{t}^{0})^{\intercal}G_{i}(X_{i})}[F_{i}\{f_{t}^{0\intercal}%
(G_{i}(X_{i})-G_{i}^{0}(X_{i}))+s|X_{i},f_{t}\}\\
&  -F_{i}\{f_{t}^{0\intercal}(G_{i}(X_{i})-G_{i}^{0}(X_{i}))|X_{i},f_{t}\}]ds.
\end{align*}
By following the same reasoning as the proof for (\ref{EQ:ELNtftg}), we have
\[
\sup_{X_{i}\mathbf{\in\lbrack}a,b\mathbf{]}^{J}}|E(W_{Nt,2i}(f_{t}%
,g)|X_{i},f_{t}\mathbf{)-}\frac{1}{2}(f_{t}-f_{t}^{0})^{\intercal}%
p_{i}(0|X_{i},f_{t})G_{i}(X_{i})G_{i}(X_{i})^{\intercal}(f_{t}-f_{t}%
^{0})|=o_{p}(||f_{t}-f_{t}^{0}||^{2}).
\]
Hence with probability approaching $1$, as $N\rightarrow\infty$,
\[
\sup_{X_{i}\mathbf{\in\lbrack}a,b\mathbf{]}^{J}}|E(W_{Nt,2i}(f_{t}%
,g)|X_{i},f_{t}\mathbf{)|}\leq C_{W}||f_{t}-f_{t}^{0}||^{2},
\]
for some constant $0<C_{W}<\infty$. Moreover,
\begin{align*}
&  E\{W_{Nt,2i}(f_{t},g)\}^{2}\\
&  =E[E[\{\int\nolimits_{0}^{(f_{t}-f_{t}^{0})^{\intercal}G_{i}(X_{i}%
)}(I(y_{it}-f_{t}^{0\intercal}G_{i}(X_{i})\leq s)-I(y_{it}-f_{t}^{0\intercal
}G_{i}(X_{i})\leq0))ds\}^{2}|X_{i},f_{t}]]\\
&  \leq E[E[|I(y_{it}-f_{t}^{0\intercal}G_{i}(X_{i})\leq(f_{t}-f_{t}%
^{0})^{\intercal}G_{i}(X_{i}))-I(y_{it}-f_{t}^{0\intercal}G_{i}(X_{i}%
)\leq0)|\\
&  \times\{(f_{t}-f_{t}^{0})^{\intercal}G_{i}(X_{i})\}^{2}|X_{i},f_{t}]]\\
&  =E[E[|I(\varepsilon_{it}\leq f_{t}{}^{\intercal}G_{i}(X_{i})-f_{t}%
^{0\intercal}G_{i}(X_{i})^{0})-I(\varepsilon_{it}\leq f_{t}^{0\intercal}%
(G_{i}(X_{i})-G_{i}(X_{i})^{0})|\\
&  \times\{(f_{t}-f_{t}^{0})^{\intercal}G_{i}(X_{i})\}^{2}|X_{i},f_{t}]]\\
&  \leq C^{\prime\prime}E|(f_{t}{}-f_{t}^{0})^{\intercal}G_{i}(X_{i})|^{3}\leq
C^{\prime\prime\prime}E||f_{t}{}-f_{t}^{0}||^{3}%
\end{align*}
for some constants $0<C^{\prime\prime}<\infty$ and $0<C^{\prime\prime\prime
}<\infty$. Therefore, for $N\rightarrow\infty$,
\begin{align*}
&  E\{W_{Nt,2}(f_{t},g)-E(W_{Nt,2}(f_{t},g)|\mathbf{X,F)}\}^{2}\\
&  =N^{-2}\sum\nolimits_{i=1}^{N}E\left[  W_{Nt,2i}(f_{t},g)-E(W_{Nt,2i}%
(f_{t},g)|X_{i},f_{t}\mathbf{)}\right]  ^{2}\\
&  \leq N^{-2}\sum\nolimits_{i=1}^{N}[2E\{W_{Nt,2i}(f_{t},g)\}^{2}%
+2E[E(W_{Nt,2i}(f_{t},g)|X_{i},f_{t}\mathbf{)}]^{2}]\\
&  \leq N^{-1}(2C^{\prime\prime\prime}E||f_{t}{}-f_{t}^{0}||^{3}+2C_{W}%
^{2}E||f_{t}{}-f_{t}^{0}||^{4})\leq C^{\prime\prime\prime\prime}N^{-1}%
E||f_{t}{}-f_{t}^{0}||^{3},
\end{align*}
for some constant $0<C^{\prime\prime\prime\prime}<\infty$. Following the same
routine procedure as the proof in Lemma \ref{LEM:Gtildatheta-theta0} by
applying the Bernstein's inequality, we have
\[
\sup_{||\boldsymbol{\lambda}_{j}-\boldsymbol{\lambda}_{j}^{0}||\leq
\widetilde{C}d_{NT}^{\ast},||f_{t}-f_{t}^{0}||\leq\varpi_{N}}||f_{t}-f_{t}%
^{0}||^{-3/2}|W_{Nt,2}(f_{t},g)-E(W_{Nt,2}(f_{t},g)|\mathbf{X,F)}%
|=O_{p}(N^{-1/2}).
\]
Hence, we have $|W_{Nt,2}(f_{t},g)-E(W_{Nt,2}(f_{t},g)|\mathbf{X,F)}%
|=O_{p}(||f_{t}-f_{t}^{0}||^{-3/2}N^{-1/2})$, uniformly in
$||\boldsymbol{\lambda}_{j}-\boldsymbol{\lambda}_{j}^{0}||\leq\widetilde{C}$%
$d_{NT}^{\ast}$  and $||f_{t}-f_{t}^{0}||\leq\varpi_{N}$. Since
\begin{align*}
N^{-1/2}||f_{t}-f_{t}^{0}||^{3/2} &  \leq N^{-1}||f_{t}-f_{t}^{0}%
||^{1/2}+||f_{t}-f_{t}^{0}||^{2}||f_{t}-f_{t}^{0}||^{1/2}\\
&  \leq N^{-1}\varpi_{N}+||f_{t}-f_{t}^{0}||^{2}\varpi_{N}\text{,}%
\end{align*}
then we have $W_{Nt,2}(f_{t},g)-E(W_{Nt,2}(f_{t},g)|\mathbf{X,F})=%
o_{p}(||f_{t}-f_{t}^{0}||^{2}+N^{-1})$, uniformly in $||\boldsymbol{\lambda
}_{j}-\boldsymbol{\lambda}_{j}^{0}||\leq\widetilde{C}$$d_{NT}^{\ast}$
and $||f_{t}-f_{t}^{0}||\leq\varpi_{N}$.
\end{proof}

\begin{lemma}
 \label{LEM:WNt1} Under Conditions (C1)-(C3), for any
$t$ there is a stochastically bounded sequence $\delta_{N,jt}$ such that as
$N\rightarrow\infty$,
\[
W_{Nt,1}=N^{-1}\sum\nolimits_{i=1}^{N}G_{i}^{0}(X_{i})\psi_{\tau}%
(\varepsilon_{it})+d_{NT}\delta_{N,t}+o_{p}(N^{-1/2}),
\]
uniformly in $||\boldsymbol{\lambda}_{j}-\boldsymbol{\lambda}_{j}^{0}%
||\leq\widetilde{C}$$d_{NT}^{\ast}$, where $W_{Nt,1}$ is defined in
(\ref{EQ:WNt1}), $\delta_{N,t}=(\delta_{N,jt},0\leq j\leq J)^{\intercal}$and
$g_{j}(x_{j})=B_{j}(x_{j})^{\intercal}\boldsymbol{\lambda}_{j}$.
\end{lemma}

\begin{proof}
Write
\begin{equation}
W_{Nt,1}=W_{Nt,11}+W_{Nt,12}+W_{Nt,13},\label{Wnt123}%
\end{equation}
where
\begin{align*}
W_{Nt,11} &  =N^{-1}\sum\nolimits_{i=1}^{N}G_{i}^{0}(X_{i})\psi_{\tau}%
(y_{it}-f_{t}^{0\intercal}G_{i}^{0}(X_{i})),\\
W_{Nt,12} &  =(W_{Ntj,12},0\leq j\leq J)^{\intercal}=N^{-1}\sum\nolimits_{i=1}%
^{N}(G_{i}(X_{i})-G_{i}^{0}(X_{i}))\psi_{\tau}(y_{it}-f_{t}^{0\intercal}%
G_{i}^{0}(X_{i})),\\
W_{Nt,13} &  =(W_{Ntj,13},0\leq j\leq J)^{\intercal}\\
&  =N^{-1}\sum\nolimits_{i=1}^{N}G_{i}(X_{i})\{\psi_{\tau}(y_{it}%
-f_{t}^{0\intercal}G_{i}(X_{i}))-\psi_{\tau}(y_{it}-f_{t}^{0\intercal}%
G_{i}^{0}(X_{i}))\}.
\end{align*}
It is easy to see that $E(W_{Ntj,12})=0$. Also by the $\phi$-mixing
distribution condition given in Condition (C1), we have var$\left(
W_{Ntj,12}\right)  \leq C_{W_{12}}N^{-1}d_{NT}^{2}$ for some constant
$0<C_{W_{12}}<\infty$, then by following the routine procedure as the proof in
Lemma \ref{LEM:Gtildatheta-theta0}, we have
\begin{equation}
\sup\nolimits_{||\boldsymbol{\lambda}_{j}-\boldsymbol{\lambda}_{j}^{0}%
||\leq\widetilde{C}d_{NT}^{\ast}}|W_{Ntj,12}|=o_{p}(N^{-1/2}).\label{EQ:WNt12}%
\end{equation}
Moreover,
\begin{align*}
E(W_{Ntj,13}|\mathbf{X,F)} &  \mathbf{=}N^{-1}\sum\nolimits_{i=1}^{N}%
g_{j}(X_{ji})E\{I(y_{it}-f_{t}^{0\intercal}G_{i}^{0}(X_{i})\leq0)-I(y_{it}%
-f_{t}^{0\intercal}G_{i}(X_{i})\leq0)|X_{i},f_{t}\}\\
&  =N^{-1}\sum\nolimits_{i=1}^{N}g_{j}(X_{ji})\int_{f_{t}^{0\intercal}%
(G_{i}(X_{i})-G_{i}^{0}(X_{i}))}^{0}p_{i}(s|X_{i},f_{t})ds\\
&  =N^{-1}\sum\nolimits_{i=1}^{N}g_{j}(X_{ji})p_{i}(0|X_{i},f_{t}%
)f_{t}^{0\intercal}(G_{i}^{0}(X_{i})-G_{i}(X_{i}))+O(d_{NT}^{2})+o(N^{-1}).
\end{align*}
Let
\[
d_{NT}\delta_{N,jt}=N^{-1}\sum\nolimits_{i=1}^{N}g_{j}(X_{ji})p_{i}%
(0|X_{i},f_{t})f_{t}^{0\intercal}(G_{i}^{0}(X_{i})-G_{i}(X_{i}))+O(d_{NT}%
^{2}).
\]
Since $N^{-1}\sum\nolimits_{i=1}^{N}\{g_{j}(X_{ji})-g_{j}^{0}(X_{ji}%
)\}^{2}\leq(\widetilde{C}d_{NT}^{\ast})^{2}$, then as $N\rightarrow\infty$,
$|d_{NT}\delta_{N,jt}|\leq C_{\delta}$$d_{NT}^{\ast}$ for some
constant $0<C_{\delta}<\infty$. Therefore,
\begin{equation}
E(W_{Ntj,13}|\mathbf{X,F})=d_{NT}\delta_{N,jt}+o(N^{-1/2}).\label{EQ:EWNt13}%
\end{equation}
Also by the $\phi$--mixing condition given in Condition (C1), we have
$E\{W_{Ntj,13}-E(W_{Ntj,13}|\mathbf{X,F)}\}^{2}\leq C_{\delta}^{\prime}%
N^{-1}d_{NT}$ for some constant $0<C_{\delta}^{\prime}<\infty$. \ Therefore,
by following the procedure as the proof in Lemma \ref{LEM:Gtildatheta-theta0},
we have
\begin{equation}
\sup\nolimits_{||\boldsymbol{\lambda}_{j}-\boldsymbol{\lambda}_{j}^{0}%
||\leq\widetilde{C}d_{NT}^{\ast}}|W_{Ntj,13}-E(W_{Ntj,13}|\mathbf{X,F)|=}%
o_{p}(N^{-1/2}).\label{EQ:WNt13}%
\end{equation}
Therefore, the result in Lemma \ref{LEM:WNt1} is proved by (\ref{Wnt123}),
(\ref{EQ:WNt12}), (\ref{EQ:EWNt13}) and (\ref{EQ:WNt13}).
\end{proof}

\begin{lemma}
\label{LEM:VNt11} Let Conditions (C1)-(C4) hold. If, in addition,
$K_{N}^{4}N^{-1}=o(1)$, $K_{N}^{-r+2}(\log T)=o(1)$ and $K_{N}^{-1}(\log
NT)(\log N)^{4}=o(1)$, then we have%
\[
||\widehat{\boldsymbol{\lambda}}^{[1]}-\boldsymbol{\lambda}^{0}-\Psi_{NT}%
^{-1}U_{N,1}||=O_{p}(d_{NT})+o_{p}(N^{-1/2}),
\]
where $U_{NT,1}$ is defined in (\ref{EQ:UNT1}) and $\Psi_{NT}$ is defined in (\ref{EQ:CphNT}).
\end{lemma}

\begin{proof}
By Lemma \ref{LEM:fhat0} and (\ref{EQ:fstar-f0}), we have
$||\widehat{f}_{t}^{[0]}-f_{t}^{0}||\leq C_{f}(d_{NT}+N^{-1/2})$ for some
constant $0<C_{f}<\infty$. Let $Q_{it}=\{B_{j}(X_{ji})^{\intercal}f_{jt},1\leq
j\leq J\}^{\intercal}$. Let $f=(f_{1}^{\intercal},\ldots,f_{T}^{\intercal
})^{\intercal}$ satisfy that $||f_{t}-f_{t}^{0}||\leq C_{f}(d_{NT}+N^{-1/2})$.
Write
\begin{align}
&  L_{NT}^{\ast}(f,\boldsymbol{\lambda})\nonumber\\
&  =E\{L_{NT}^{\ast}(f,\boldsymbol{\lambda})|\mathbf{X,F}%
\}-(\boldsymbol{\lambda}-\boldsymbol{\lambda}^{0})^{\intercal}\{V_{NT,1}%
(f)-E(V_{NT,1}(f)|\mathbf{X,F})\}\nonumber\\
&  +\,V_{NT,2}(f,\boldsymbol{\lambda})-E(V_{NT,2}(f,\boldsymbol{\lambda
})|\mathbf{X,F}),\label{EQ:LNTstar}%
\end{align}
where
\begin{align}
V_{NT,1}(f) &  =(NT)^{-1}\sum\nolimits_{i=1}^{N}\sum\nolimits_{t=1}^{T}%
Q_{it}\psi_{\tau}(y_{it}-f_{ut}-\boldsymbol{\lambda}^{0\intercal}%
Q_{it}),\label{VNt1}\\
V_{NT,2}(f,\boldsymbol{\lambda}) &  =(NT)^{-1}\sum\nolimits_{i=1}^{N}%
\sum\nolimits_{t=1}^{T}\{\rho_{\tau}(y_{it}-f_{ut}-\boldsymbol{\lambda
}^{\intercal}Q_{it})-\rho_{\tau}(y_{it}-f_{ut}-\boldsymbol{\lambda
}^{0\intercal}Q_{it})\nonumber\\
&  +(\boldsymbol{\lambda}-\boldsymbol{\lambda}^{0})^{\intercal}Q_{it}%
\psi_{\tau}(y_{it}-f_{ut}-\boldsymbol{\lambda}^{0\intercal}Q_{it})\}.\nonumber
\end{align}
By following the same reasoning as in the proofs of Lemmas \ref{LEM:LNtfg} and
\ref{LEM:WNt2}, we have
\begin{equation}
E\{L_{NT}^{\ast}(f,\boldsymbol{\lambda})|\mathbf{X}\}=-(\boldsymbol{\lambda
}-\boldsymbol{\lambda}^{0})^{\intercal}E(V_{NT,1}(f)|\mathbf{X,F)+}\frac{1}%
{2}(\boldsymbol{\lambda}-\boldsymbol{\lambda}^{0})^{\intercal}\Psi
_{NT}(\boldsymbol{\lambda}-\boldsymbol{\lambda}^{0})+o_{p}%
(||\boldsymbol{\lambda}-\boldsymbol{\lambda}^{0}||^{2}),\label{EQ:ELNTstar}%
\end{equation}%
\begin{equation}
V_{NT,2}(f,\boldsymbol{\lambda})-E(V_{NT,2}(f,\boldsymbol{\lambda
})|\mathbf{X,F)=}o_{p}(||\boldsymbol{\lambda}-\boldsymbol{\lambda}^{0}%
||^{2}+(NT)^{-1}),\label{VT2}%
\end{equation}
uniformly in $||f_{t}-f_{t}^{0}||\leq C_{f}(d_{NT}+N^{-1/2})$ and
$||\boldsymbol{\lambda}-\boldsymbol{\lambda}^{0}||\leq\varsigma_{NT}$, where
$\varsigma_{NT}$ is any sequence of positive numbers satisfying $\varsigma
_{NT}=o(1)$.Thus, by (\ref{EQ:LNTstar}), (\ref{EQ:ELNTstar}) and (\ref{VT2}),
we have
\[
L_{NT}^{\ast}(f,\boldsymbol{\lambda})=-(\boldsymbol{\lambda}%
-\boldsymbol{\lambda}^{0})^{\intercal}V_{NT,1}(f)\mathbf{+}\frac{1}%
{2}(\boldsymbol{\lambda}-\boldsymbol{\lambda}^{0})^{\intercal}\Psi
_{NT}(\boldsymbol{\lambda}-\boldsymbol{\lambda}^{0})\mathbf{+}o_{p}%
(||\boldsymbol{\lambda}-\boldsymbol{\lambda}^{0}||^{2}+(NT)^{-1}),
\]
uniformly in $||f_{t}-f_{t}^{0}||\leq C_{f}(d_{NT}+N^{-1/2})$ and
$||\boldsymbol{\lambda}-\boldsymbol{\lambda}^{0}||\leq\varsigma_{NT}$.
Therefore, we have
\[
\widehat{\boldsymbol{\lambda}}^{[1]}-\boldsymbol{\lambda}^{0}=\Psi_{NT}%
^{-1}V_{NT,1}(\widehat{f}^{[0]})+o_{p}\{(NT)^{-1/2}\}.
\]
By following the same reasoning as the proof for (\ref{EQ:ZZ}), as
$(N,T)\rightarrow\infty$ with probability approaching $1$, we have
$||\Psi_{NT}^{-1}||\leq C_{\Psi}^{\prime}$ for some constant $0<C_{\Psi
}^{\prime}<\infty$. In Lemma \ref{LEM:VNt1}, we will show that $||V_{NT,1}%
(\widehat{f}^{[0]})-U_{NT,1}||=O_{p}(d_{NT})+o_{p}(N^{-1/2})$. Therefore, the
result in Lemma \ref{LEM:VNt11} follows from the above results, and thus the
proof is completed.
\end{proof}

\begin{lemma}
\label{LEM:VNt1} Let Conditions (C1)-(C4) hold. If, in addition,
$K_{N}^{4}N^{-1}=o(1)$, $K_{N}^{-r+2}(\log T)=o(1)$ and $K_{N}^{-1}(\log
NT)(\log N)^{4}=o(1)$, then we have%
\[
||V_{NT,1}(\widehat{f}^{[0]})-U_{NT,1}||=O_{p}(d_{NT})+o_{p}(N^{-1/2}),
\]
where $V_{NT,1}$ and $U_{NT,1}$ are defined in (\ref{VNt1}) and (\ref{EQ:UNT1}%
), respectively.
\end{lemma}

\begin{proof}
Write
\begin{equation}
V_{NT,1}(f)=V_{NT,11}+V_{NT,12}(f)+V_{NT,13}(f),\label{VNT1}%
\end{equation}
where
\begin{align*}
V_{NT,11} &  =U_{NT,1}=(NT)^{-1}\sum\nolimits_{i=1}^{N}\sum\nolimits_{t=1}%
^{T}Q_{it}^{0}\psi_{\tau}(\varepsilon_{it}),\\
V_{NT,12}(f) &  =(NT)^{-1}\sum\nolimits_{i=1}^{N}\sum\nolimits_{t=1}%
^{T}(Q_{it}-Q_{it}^{0})\psi_{\tau}(\varepsilon_{it})),\\
V_{NT,13}(f) &  =(NT)^{-1}\sum\nolimits_{i=1}^{N}\sum\nolimits_{t=1}^{T}%
Q_{it}\{\psi_{\tau}(y_{it}-f_{ut}-\boldsymbol{\lambda}^{0\intercal}%
Q_{it}))-\psi_{\tau}(\varepsilon_{it})\}.
\end{align*}
Since $||N^{-1}\sum\nolimits_{i=1}^{N}B(X_{i})\psi_{\tau}(\varepsilon
_{it})||=O_{p}(N^{-1/2})$, we have with probability approaching 1,
\begin{align}
&  & \sup_{||f_{t}-f_{t}^{0}||\leq C_{f}(d_{NT}+N^{-1/2})}||V_{NT,12}|| &
\leq T^{-1}\sum\nolimits_{t=1}^{T}||N^{-1}\sum\nolimits_{i=1}^{N}B(X_{i}%
)\psi_{\tau}(\varepsilon_{it})||\nonumber\\
&  & \times\ \sup_{||f_{t}-f_{t}^{0}||\leq C_{f}(d_{NT}+N^{-1/2})}%
||f_{t}-f_{t}^{0}|| &  =O\{N^{-1/2}(d_{NT}+N^{-1/2})\}=o(N^{-1/2}%
+d_{NT}).\label{Vj12}%
\end{align}
By following the same procedure as the proof for (\ref{EQ:VLamtilda}), we have
for any vector $\mathbf{a}\in R^{K_{N}J}$ with $||\mathbf{a||=}1$,
\[
\text{var}(\mathbf{a}^{\intercal}V_{NT,13}(f)\mathbf{a)=}O\mathbf{\{}%
K_{N}(d_{NT}+N^{-1/2})(NT)^{-1}\},
\]
uniformly in $||f_{t}-f_{t}^{0}||\leq C_{f}(d_{NT}+N^{-1/2})$. Then by the
procedure as the proof in Lemma \ref{LEM:Gtildatheta-theta0}, we have%
\begin{align*}
\sup_{||f_{t}-f_{t}^{0}||\leq C_{f}(d_{NT}+N^{-1/2})}||V_{NT,13}%
(f)-E\{V_{NT,13}(f)\}|| &  \mathbf{=}O_{p}\mathbf{\{}K_{N}^{1/2}%
(d_{NT}+N^{-1/2})^{1/2}(NT)^{-1/2}\}\\
&  =o_{p}(d_{NT}).
\end{align*}
Hence,
\begin{equation}
||V_{NT,13}(\widehat{f}^{[0]})-E\{V_{NT,13}(\widehat{f}^{[0]})\}||=o_{p}%
(d_{NT}).\label{V-EV}%
\end{equation}
Let
\[
\kappa_{it}(f)=f_{ut}^{0}-f_{ut}+\sum\nolimits_{j=1}^{J}(\widetilde{g}_{j}%
^{0}(X_{ji})(f_{jt}^{0}-f_{jt})+r_{j,it}^{\ast}).
\]
Then there exist constants $0<C,C^{\prime}<\infty$ such that
\begin{align}
||E\{V_{NT,13}(f)|\mathbf{X,F}\}|| &  \leq C||E[(NT)^{-1}\sum\nolimits_{i=1}%
^{N}\sum\nolimits_{t=1}^{T}B_{i}(X_{i})\{I(\varepsilon_{it}\leq
0)-I(\varepsilon_{it}\leq\kappa_{it}(f))\}|\mathbf{X,F}]||\nonumber\\
&  \leq C^{\prime}||(NT)^{-1}\sum\nolimits_{i=1}^{N}\sum\nolimits_{t=1}%
^{T}B_{i}(X_{i})\kappa_{it}(f)p_{i}(0|X_{i},f_{t})||\label{EQ:EV}%
\end{align}
uniformly in $||f_{t}-f_{t}^{0}||\leq C_{f}(d_{NT}+N^{-1/2})$. Moreover, by
(\ref{EQ:fstar-f0}) and Lemma \ref{LEM:fhat0}, we have
\begin{align}
&  \left\Vert (NT)^{-1}\sum\nolimits_{i=1}^{N}\sum\nolimits_{t=1}^{T}%
B_{i}(X_{i})\kappa_{it}(\widehat{f}^{[0]})p_{i}(0|X_{i},f_{t})\right.
\nonumber\\
&  \left.  +(NT)^{-1}\sum\nolimits_{i=1}^{N}\sum\nolimits_{t=1}^{T}B_{i}%
(X_{i})p_{i}(0|X_{i},f_{t})\widetilde{g}^{0}(X_{i})^{\intercal}[\Lambda
_{N}^{-1}\{N^{-1}\sum\nolimits_{i=1}^{N}G_{i}^{0}(X_{i})(\tau-I(\varepsilon
_{it}<0))\}]\right\Vert \nonumber\\
&  =O(d_{NT})+o_{p}(N^{-1/2}).\label{EQ:Bkit}%
\end{align}
Since $||(NT)^{-1}\sum\nolimits_{t=1}^{T}\sum\nolimits_{i=1}^{N}G_{i}%
^{0}(X_{i})(\tau-I(\varepsilon_{it}<0))||=O_{p}\{(NT)^{-1/2}\}$, and
\[
||(NT)^{-1}\sum\nolimits_{i=1}^{N}\sum\nolimits_{t=1}^{T}B_{i}(X_{i}%
)p_{i}(0|X_{i},f_{t})||=O_{p}(1),
\]
we have
\begin{align*}
&  \left\Vert (NT)^{-1}\sum\nolimits_{i=1}^{N}\sum\nolimits_{t=1}^{T}%
B_{i}(X_{i})p_{i}(0|X_{i},f_{t})\widetilde{g}^{0}(X_{i})^{\intercal}%
[\Lambda_{N}^{-1}\{N^{-1}\sum\nolimits_{i=1}^{N}G_{i}^{0}(X_{i})(\tau
-I(\varepsilon_{it}<0))\}]\right\Vert \\
&  =O_{p}\{(NT)^{-1/2}\}.
\end{align*}
Therefore, by (\ref{EQ:EV}) and (\ref{EQ:Bkit}), we have with probability
approaching 1,
\begin{equation}
||E\{V_{NT,13}(\widehat{f}^{[0]})|\mathbf{X,F}\}||=O(d_{NT})+o(N^{-1/2}%
).\label{EV}%
\end{equation}
By (\ref{V-EV}) and (\ref{EV}), we have
\begin{equation}
||V_{NT,13}(\widehat{f}^{[0]})||=O_{p}(d_{NT})+o_{p}(N^{-1/2}).\label{Vjk13}%
\end{equation}
\ Therefore, the result in Lemma \ref{LEM:VNt1} follows from (\ref{VNT1}),
(\ref{Vj12}), and (\ref{Vjk13}) directly.
\end{proof}

\end{document}